\documentclass{IEEEtran}

\usepackage{graphicx}
\usepackage{amsmath,amssymb, amsthm, bbm}
\usepackage{epsfig}
\usepackage{varioref}
\usepackage{subfigure}
\usepackage[sort&compress,square,numbers]{natbib}
\usepackage{arydshln}
\usepackage{paralist}
\usepackage[colorlinks]{hyperref}
\usepackage{algpseudocode}
\usepackage{subfigure}

\def\Csp{{\mathbb{C}}}

\def\Rsp{{\mathbb{R}}}

\def\px{P_{\Xmatsc}}

\def\pz{P_{\Zmatsc}}

\def\qx{Q^{\left(u\right)}_{\Xmatsc}}

\def\uGausss{u_{\Gausssc}}

\newcommand{\Gausssc}{{\mbox{\tiny${\rm{G}}$}}}

\newcommand{\avec}{{\bf{a}}}

\newcommand{\bvec}{{\bf{b}}}

\newcommand{\dvec}{{\bf{d}}}

\newcommand{\yvec}{{\bf{y}}}

\newcommand{\xvec}{{\bf{x}}}

\newcommand{\gvec}{{\bf{g}}}

\newcommand{\etavec}{{\bf{\eta}}}

\newcommand{\zerovec}{{\bf{0}}}

\newcommand{\Amat}{{\bf{A}}}

\newcommand{\Bmat}{{\bf{B}}}

\newcommand{\Cmat}{{\bf{C}}}

\newcommand{\Dmat}{{\bf{D}}}

\newcommand{\Gmat}{{\bf{G}}}

\newcommand{\Hmat}{{\bf{H}}}

\newcommand{\Imat}{{\bf{I}}}

\newcommand{\Smat}{{\bf{S}}}

\newcommand{\Rmat}{{\bf{R}}}

\newcommand{\Umat}{{\bf{U}}}

\newcommand{\Vmat}{{\bf{V}}}

\newcommand{\Wmat}{{\bf{W}}}

\newcommand{\Xmat}{{\bf{X}}}

\newcommand{\Xmatsc}{{\mbox{\boldmath \tiny $\Xmat$}}}

\newcommand{\deltavec}{{\mbox{\boldmath  $\delta$}}}

\newcommand{\Zmatsc}{{\mbox{\boldmath \tiny $\Zmat$}}}

\newcommand{\Smatsc}{{\mbox{\boldmath \tiny $\Smat$}}}

\newcommand{\Wmatsc}{{\mbox{\boldmath \tiny $\Wmat$}}}

\newcommand{\Zmat}{{\bf{Z}}}

\newcommand{\Psimat}{\mbox{\boldmath $\Psi$}}

\newcommand{\bzeta}{\mbox{\boldmath $\zeta$}}

\def\bzeta{{\mbox{\boldmath $\zeta$}}}

\def\bxi{{\mbox{\boldmath $\xi$}}}

\def\bSigma{{\mbox{\boldmath $\Sigma$}}}

\def\bxi{{\mbox{\boldmath $\xi$}}}

\def\psivec{{\mbox{\boldmath $\psi$}}}

\def\bSigmasc{{\mbox{\boldmath \tiny{$\bSigma$}}}}

\def\etavec{{\mbox{\boldmath $\eta$}}}

\def\XCalsc{{\mbox{\tiny $\mathcal{X}$}}}

\def\XCal{{\mbox{$\mathcal{X}$}}}

\def\muvec{{\mbox{\boldmath $\mu$}}}

\newcommand{\be}{\begin{equation}}

\newcommand{\ee}{\end{equation}}

\newcommand{\beqna}{\begin{eqnarray}}

\newcommand{\eeqna}{\end{eqnarray}}

\begin{document}
\title{Robust Multiple Signal Classification via Probability Measure Transformation}
\author{{Koby Todros, Member, IEEE, and Alfred O. Hero, Fellow, IEEE}
\thanks{Koby Todros is with the Dept. of ECE, Ben-Gurion University of the Negev, Beer-Sheva 84105, Israel. Email: todros@ee.bgu.ac.il.}
\thanks{Alfred O. Hero is with the Dept. of EECS, University of Michigan, Ann-Arbor 48105, MI, U.S.A. Email: hero@eecs.umich.edu.}}

\maketitle
\newtheorem{Theorem}{Theorem}
\newtheorem{Lemma}{Lemma}
\newtheorem{Corollary}{Corollary}
\newtheorem{Conclusion}{Conclusion}
\newtheorem{Proposition}{Proposition}
\newtheorem{Definition}{Definition}
\newtheorem{Remark}{Remark}
\newtheorem{Identity}{Identity}
\begin{abstract}
In this paper, we introduce a new framework for robust multiple signal classification (MUSIC). The proposed framework, called robust measure-transformed (MT) MUSIC, is based on applying a transform to the probability distribution of the received signals, i.e., transformation of the probability measure defined on the observation space. In robust MT-MUSIC, the sample covariance is replaced by the empirical MT-covariance. By judicious choice of the transform we show that:
\begin{inparaenum}[\upshape(1\upshape)]
\item
the resulting empirical MT-covariance is B-robust, with bounded influence function that takes negligible values for large norm outliers, and
\item
under the assumption of spherically contoured noise distribution, the noise subspace can be determined from the eigendecomposition of the MT-covariance. 
\end{inparaenum}
Furthermore, we derive a new robust measure-transformed minimum description length (MDL) criterion for estimating the number of signals, and
extend the MT-MUSIC framework to the case of coherent signals. The proposed approach is illustrated in simulation examples that show its advantages as compared to other robust MUSIC and MDL generalizations. 
\end{abstract}
\section{Introduction}
\label{sec:intro}
The multiple signal classification (MUSIC) algorithm \cite{Schmidt}, \cite{Viberg} is a well established technique for estimating direction-of-arrivals (DOAs) of signals impinging on an array of sensors. It operates by finding DOAs with corresponding array steering vectors that have minimal projections onto the empirical noise subspace, whose spanning vectors are obtained via eigendecomposition of the sample covariance matrix (SCM) of the array outputs.

In the presence of outliers, possibly caused by heavy-tailed impulsive noise, the SCM poorly estimates the covariance of the array outputs, resulting in unreliable DOAs estimates. In order to overcome this limitation, several MUSIC generalizations have been proposed in the literature that replace the SCM with robust association or scatter matrix estimators, for which the empirical noise subspace can be determined from their eigendecomposition. 

Under the assumption that the signal and noise components are jointly $\alpha$-stable processes \cite{AlphaStable}, it was proposed in \cite{ROCMUSIC} to replace the SCM with empirical covariation matrices that involve fractional lower-order statistics. Although $\alpha$-stable processes are appropriate for modelling impulsive noise  \cite{AlphaStable2}, the assumption that the signal and noise components are jointly $\alpha$-stable is restrictive. In \cite{LIU}, a less restrictive approach considering circular signals contaminated by additive $\alpha$-stable noise was developed that replaces the SCM with matrices comprised of empirical fractional-lower-order-moments. Although this approach is less restrictive than the one proposed in \cite{ROCMUSIC}, violation of the signal circularity assumption, e.g., in the case of BPSK signals, results in poor DOA estimation performance \cite{LIU}.  

In \cite{swami1997}-\cite{swami2002} it was proposed to apply MUSIC after passing the data through a zero-memory non-linear (ZMNL) function that suppresses outliers by clipping the amplitude of the received signals. The ZMNL approach has simple implementation having low computational complexity, and unlike the methods proposed in \cite{ROCMUSIC}, \cite{LIU} it does not require restrictive assumptions on the signal and noise probability distributions. Although the ZMNL preprocessing may result in more accurate DOA estimation than the methods in \cite{ROCMUSIC}, \cite{LIU}, it may not preserve the noise subspace which can lead to performance degradation \cite{Sadler}. 

Under the assumption of normally distributed signals in heavy tailed noise, a similar approach was proposed in \cite{Zoubir} that is based on successive outlier trimming until the remaining data is Gaussian. Normality  of the data is tested using the Shapiro-Wilk's test. Similarly to the ZMNL preprocessing, the noise subspace may not be preserved after the trimming procedure. Moreover, the key assumption that the signals are Gaussian may not be satisfied in some practical scenarios. 

In \cite{Visuri}, a different MUSIC generalization was proposed that replaces the SCM with empirical sign or rank covariances. Using only the assumption of spherically distributed noise, it was shown that convergent estimates of the noise subspace can be obtained from their eigendecomposition.  The influence functions \cite{Hampel} of the empirical sign and rank covariance matrices, that measure their sensitivity to an outlier, are bounded \cite{croux2002sign}. In other words, these estimators are B-robust \cite{Hampel}. However, it can be shown that the Frobenius norms of their matrix valued influence functions do not approach zero as the magnitude of the outlier approaches infinity, i.e., they do not reject large outliers. Indeed, the empirical sign and rank covariance matrices have influence functions with constant Frobenius norms for spherically symmetric distributions. 

In \cite{Visa}, robust M-estimators of scatter \cite{Marrona}, \cite{Huber}, such as the maximum-likelihood, Huber's \cite{Huber}, and Tyler's \cite{Tyler} M-estimators, extended to complex elliptically symmetric (CES) distributions, were proposed as alternatives to the SCM. Under the class of CES distributions having finite second-order moments, these estimators provide consistent estimation of the covariance up to a positive scalar, resulting in consistent estimation of the noise subspace. Although this approach can provide robustness against outliers with negligible loss in efficiency when the observations are normally distributed, it may suffer from the following drawbacks. First, when the observations are not elliptically distributed, M-estimators may lose asymptotic consistency \cite{Hallin}, which may lead to poor estimation of the noise subspace. Second, M-estimators of scatter are often computed using an iterative fixed-point algorithm that converges to a unique solution under some regularity conditions. Each iteration involves matrix inversion which may be computationally demanding in high dimensions, or unstable when the scatter matrix is close to singular. Moreover, although the influence functions of M-estimators may be bounded, they may not behave well for large norm outliers that can negatively affect estimation performance. Indeed, similarly to the method of \cite{Visuri}, Tyler's scatter M-estimator does not reject large outliers and its matrix valued influence function \cite{Visa} has constant Frobenius norm for spherically symmetric distributions.

In \cite{LP_MUSIC}, a robust MUSIC generalization called $l_{p}$-MUSIC was proposed that estimates the noise subspace by minimizing the $l_{p}$-norm ($1<p<2$) of the residual between the data matrix and its low-rank representation. The $l_{p}$-norm with $p<2$ is less sensitive to outliers than the $l_{2}$-norm. Therefore, $l_{p}$-MUSIC is more robust against impulsive noise as compared to MUSIC that is based on $l_{2}$-norm minimization of the data fitting error matrix \cite{LP_MUSIC}. However, unlike MUSIC and other robust generalizations, in $l_{p}$-MUSIC 
the empirical noise subspace is not determined  by solving a simple eigendecomposition problem. Indeed, in \cite{LP_MUSIC} the non-convex $l_{p}$-norm minimization is performed by alternating convex optimization scheme that may converge to undesired local minima.  

In this paper, we introduce a new framework for robust MUSIC. The proposed framework, called robust measure-transformed MUSIC (MT-MUSIC), is inspired by a measure transformation approach that was recently applied to canonical correlation analysis \cite{MTCCA} and independent component analysis \cite{MTICA}. Robust MT-MUSIC is based on applying a transform to the probability distribution of the data. The proposed transform is defined by a non-negative function, called the MT-function, and maps the probability distribution into a set of new probability measures on the observation space. By modifying the MT-function, classes of measure transformations can be obtained that have useful properties. Under the proposed transform we define the measure-transformed (MT) covariance and derive its strongly consistent estimate, which is also shown to be Fisher consistent \cite{Cox}. Robustness of the empirical MT-covariance is established in terms of boundedness of its influence function. A sufficient condition on the MT-function that guarantees B-robustness of the empirical MT-covariance is also obtained. 

In robust MT-MUSIC, the SCM is replaced by the empirical MT-covariance. The MT-function is selected such that the resulting empirical MT-covariance is B-robust, and the noise subspace can be determined from the eigendecomposition of the MT-covariance. By modifying the MT-function such that these conditions are satisfied a class of robust MT-MUSIC algorithms can be obtained. 

Selection of the MT-function under the family of zero-centered Gaussian functions, parameterized by a scale parameter, results in a new algorithm called Gaussian MT-MUSIC. We show that the empirical Gaussian MT-covariance is B-robust with influence function that approaches zero as the outlier magnitude approaches infinity. Under the additional assumption that the noise component has a spherically contoured distribution, we show that the noise subspace can be determined from the eigendecomposition of Gaussian MT-covariance. Note that this spherically contoured noise distribution assumption is weaker than the standard i.i.d. Gaussian noise assumption. We propose a data-driven procedure for selecting the scale parameter of the Gaussian MT-function. This procedure has the property that it prevents significant transform-domain Fisher-information loss when the observations are normally distributed. 

In this paper, a robust estimate of the number of signals is proposed that is based on minimization of a measure-transformed version of the minimum description length (MDL) criterion \cite{Wax}. This criterion, called MT-MDL, is obtained by replacing the eigenvalues of the SCM with the eigenvalues of the empirical MT-covariance. We show that under some mild conditions, minimization of the MT-MDL criterion results in strongly consistent estimation of the number of signals regardless the underlying distribution of the data. These conditions are satisfied when the Gaussian MT-function is implemented and the noise component has a spherically contoured distribution. The MT-MDL criterion with the Gaussian MT-function is called the Gaussian MT-MDL.

The proposed Gaussian MT-MUSIC algorithm is extended to the case of coherent signals impinging on a uniform linearly spaced array (ULA) \cite{VanTrees}. This extension is carried out through forward-backward spatial smoothing of the empirical Gaussian MT-covariance matrix.  

The Gaussian MT-MUSIC algorithm and the Gaussian MT-MDL criterion are evaluated by simulations to illustrate their advantages relative to other robust MUSIC and MDL generalizations. We examine scenarios of non-coherent and coherent signals contaminated by several types of spherically contoured noise distributions arising from the compound Gaussian (CG) family. This family encompasses common heavy-tailed distributions, such as the $t$-distribution, the $K$-distribution, and the CG-distribution with inverse Gaussian texture, and have been widely adopted for modeling radar clutter \cite{CG1}-\cite{CG4}. 

The paper is organized as follows. In Section \ref{MUSIC_ALG}, the robust MT-MUSIC framework is presented. In Section \ref{GMTMUSIC}, the Gaussian MT-MUSIC algorithm is derived. In Section \ref{MTMDLEST}, we propose a measure-transformed generalization of the MDL criterion for estimating the number of signals.  In Section \ref{CMTMUSIC}, a spatially smoothed version of the Gaussian MT-MUSIC algorithm for coherent signals is developed. The proposed approach is illustrated by simulation in Section \ref{Examples}. In Section \ref{Conclusions}, the main points of this contribution are summarized. The proofs of the propositions and theorems stated throughout the paper are given in the Appendix.
\section{Robust measure-transformed MUSIC}
\label{MUSIC_ALG}
In this section, the robust MT-MUSIC procedure is presented. First, the sensor array model is introduced. Second, a general transformation on probability measures is established. Under the proposed transform, we define the MT-covariance matrix and derive its strongly consistent estimate. Robustness of the empirical MT-covariance is studied by analyzing its influence function. Finally, based on the assumed array model, we propose a robust MT-MUSIC procedure that replaces the SCM with the empirical MT-covariance of the received signals. 
\subsection{Array model}
Consider an array of $p$ sensors that receive signals generated by $q<p$ narrowband incoherent far-field point sources with distinct azimuthal DOAs $\theta_{1},\ldots,\theta_{q}$. Under this model, the array output satisfies \cite{Viberg}:
\begin{equation}  
\label{ArrayModel}
\Xmat\left(n\right)=\Amat\Smat\left(n\right)+\Wmat\left(n\right),
\end{equation}
where $n\in\mathbb{N}$ is a discrete time index, $\Xmat\left(n\right)\in\Csp^{p}$ is the vector of received signals, $\Smat\left(n\right)\in\Csp^{q}$ is a zero-mean latent random vector, comprised of the emitted signals, with non-singular covariance, and $\Wmat\left(n\right)\in\Csp^{p}$ is an additive spatially white noise with zero location parameter. The matrix $\Amat\triangleq\left[\avec\left(\theta_{1}\right),\ldots,\avec\left(\theta_{q}\right)\right]\in\Csp^{{p}\times{q}}$ is the array steering matrix, where $\avec\left(\theta\right)$ is the steering vector of the array toward direction $\theta$. We assume that the array is unambiguous, i.e., any collection of $p$ steering vectors corresponding to distinct DOAs forms a linearly independent set. Therefore, $\Amat$ has full column rank, and identification of its column vectors is equivalent to the problem of identifying the DOAs. We also assume that $\Smat\left(n\right)$ and $\Wmat\left(n\right)$ are statistically independent and first-order stationary. To simplify notation, the time index $n$ will be omitted in the sequel except where noted.
\subsection{Probability measure transform}
We define the measure space $\left(\XCal,\mathcal{S}_{\XCalsc},\px\right)$, where $\XCal$ is the observation space of a random vector $\Xmat\in\Csp^{p}$, $\mathcal{S}_{\XCalsc}$ is a $\sigma$-algebra over $\XCal$, and $\px$ is a probability measure on $\mathcal{S}_{\XCalsc}$. Let $g:\XCal\rightarrow\Csp$ denote an integrable scalar function on $\XCal$. The expectation of $g\left(\Xmat\right)$ under $\px$ is defined as
\begin{equation}
\label{ExpDef}
{\rm{E}}\left[g\left(\Xmat\right);\px\right]\triangleq\int\limits_{\XCalsc}g\left(\xvec\right)d\px\left(\xvec\right),
\end{equation}
where $\xvec\in\XCal$. The empirical probability measure $\hat{P}_{\Xmatsc}$ given a sequence of samples $\Xmat\left(n\right)$, $n=1,\ldots,N$ from $\px$ is specified by
\begin{equation}
\label{EmpProbMes}
\hat{P}_{\Xmatsc}\left(A\right)=\frac{1}{N}\sum\limits_{n=1}^{N}\delta_{\Xmatsc(n)}\left(A\right),
\end{equation}
where $A\in\mathcal{S}_{\XCalsc}$, and $\delta_{\Xmatsc(n)}\left(\cdot\right)$ is the Dirac probability measure at $\Xmat\left(n\right)$ \cite{Folland}.
\begin{Definition}
\label{Def1}
Given a non-negative function $u:\Csp^{p}\rightarrow\Rsp_{+}$ satisfying 
\begin{equation} 
\label{Cond}
0<{{{\rm{E}}}\left[u\left(\Xmat\right);\px\right]}<\infty,
\end{equation}
a transform on $\px$ is defined via the relation:
\begin{equation}
\label{MeasureTransform} 
\qx\left(A\right)\triangleq{\rm{T}}_{u}\left[\px\right]\left(A\right)=\int\limits_{A}\varphi_{u}\left(\xvec\right)d\px\left(\xvec\right),
\end{equation}
where $A\in\mathcal{S}_{\XCalsc}$, $\xvec\in\XCal$, and
\begin{equation}
\label{VarPhiDef} 
\varphi_{u}\left(\xvec\right)\triangleq\frac{u\left(\xvec\right)}{{{\rm{E}}}\left[u\left(\Xmat\right);\px\right]}.
\end{equation}
The function $u\left(\cdot\right)$, associated with the transform ${\rm{T}}_{u}\left[\cdot\right]$, is called the MT-function.
\end{Definition}  
\begin{Proposition}[Properties of the transform]
\label{Prop1}
Let $\qx$ be defined by relation (\ref{MeasureTransform}). 
Then
\begin{enumerate}
\item
\label{P1}
$\qx$ is a probability measure on $\mathcal{S}_{\XCalsc}$.
\item
\label{P2}
$\qx$ is absolutely continuous w.r.t. $\px$, with Radon-Nikodym derivative \cite{Folland}:
\begin{equation}
\label{MeasureTransformRadNik}     
\frac{d\qx\left(\xvec\right)}{d\px\left(\xvec\right)}=\varphi_{u}\left(\xvec\right).
\end{equation}
\item 
\label{P3} 
Assume that the MT-function $u\left(\cdot\right)$ is strictly positive, and let $\gvec:\XCal\rightarrow\Csp^{m}$ denote an integrable function over $\XCal$.
If the covariance of $\gvec\left(\Xmat\right)$ under $\px$ is non-singular, then it is non-singular under the transformed probability measure $\qx$. 
\end{enumerate} 
[A proof is given in Appendix \ref{Prop1Proof}]
\end{Proposition}
The probability measure $\qx$ is said to be generated by the MT-function $u\left(\cdot\right)$. By modifying $u\left(\cdot\right)$, such that the condition (\ref{Cond}) is satisfied, virtually any probability measure on $\mathcal{S}_{\XCalsc}$ can be obtained. 
\subsection{The measure-transformed covariance}
According to (\ref{MeasureTransformRadNik}) the covariance of $\Xmat$ under $\qx$ is given by  
\begin{equation} 
\label{MTCovZ}   
\bSigma^{\left(u\right)}_{\Xmatsc}\triangleq{\rm{E}}\left[\left(\Xmat-\muvec^{\left(u\right)}_{\Xmatsc}\right)
\left(\Xmat-\muvec^{\left(u\right)}_{\Xmatsc}\right)^H\varphi_{u}\left(\Xmat\right);\px\right],
\end{equation}
where 
\begin{equation} 
\label{MTMean} 
\muvec^{\left(u\right)}_{\Xmatsc}\triangleq{\rm{E}}\left[\Xmat\varphi_{u}\left(\Xmat\right);\px\right]
\end{equation}
is the expectation of $\Xmat$ under $\qx$. Equation (\ref{MTCovZ}) implies that $\bSigma^{\left(u\right)}_{\Xmatsc}$ is a weighted covariance matrix of $\Xmat$ under $\px$, with weighting function $\varphi_{u}\left(\cdot\right)$. Hence, $\bSigma^{\left(u\right)}_{\Xmatsc}$ can be estimated using only samples from the distribution $\px$. By modifying the MT-function $u\left(\cdot\right)$, such that the condition (\ref{Cond}) in definition \ref{Def1} is satisfied,  the MT-covariance matrix under $\qx$ is modified. In particular, by choosing $u\left(\xvec\right)\equiv{1}$, we have $\qx=\px$, for which the standard covariance matrix $\bSigma_{\Xmatsc}$  is obtained. 

In the following proposition, a strongly consistent estimate of $\bSigma^{\left(u\right)}_{\Xmatsc}$ is constructed, based on $N$ i.i.d. samples of $\Xmat$.
Unlike the empirical MT-covariance proposed in \cite{MTCCA}, the construction is based on complex observations and its almost sure convergence conditions are different.
\begin{Proposition}[Strongly consistent estimate of the MT-covariance]
\label{ConsistentEst}
Let $\Xmat\left(n\right)$, $n=1,\ldots,N$ denote a sequence of i.i.d. samples from $\px$, and define the empirical covariance estimate
\begin{equation}   
\label{Rx_u_Est}   
\hat{\bSigma}^{\left(u\right)}_{\Xmatsc}\triangleq\sum\limits_{n=1}^{N}\left(\Xmat\left(n\right)-\hat{\muvec}^{\left(u\right)}_{\Xmatsc}\right)\left(\Xmat\left(n\right)-\hat{\muvec}^{\left(u\right)}_{\Xmatsc}\right)^{H}\hat{\varphi}_{u}\left(\Xmat\left(n\right)\right)
\end{equation}
where 
\begin{equation}   
\label{Mu_u_Est} 
\hat{\muvec}^{\left(u\right)}_{\Xmatsc}\triangleq\sum_{n=1}^{N}\Xmat\left(n\right)\hat{\varphi}_{u}\left(\Xmat\left(n\right)\right)
\end{equation}
and
\begin{equation}
\label{hat_varphi}  
\hat{\varphi}_{u}\left(\Xmat\left(n\right)\right)\triangleq\frac{u\left(\Xmat\left(n\right)\right)}{\sum\limits_{n=1}^{N}u\left(\Xmat\left(n\right)\right)}.
\end{equation}
If 
\begin{equation}
\label{Cond11}
{\rm{E}}\left[\left\|\Xmat\right\|^{2}_{2}u\left(\Xmat\right);\px\right]<\infty,
\end{equation}
where ${\|\cdot\|}_{2}$ denotes the $l_{2}$-norm, then $\hat{\bSigma}^{\left(u\right)}_{\Xmatsc}\rightarrow{\bSigma}^{\left(u\right)}_{\Xmatsc}$ almost surely (a.s.) as $N\rightarrow\infty$. 
[A proof is given in Appendix \ref{Prop2Proof}]
\end{Proposition}
Note that for $u(\Xmat)\equiv{1}$ the estimator $\frac{N}{N-1}\hat{\bSigma}^{(u)}_{\Xmatsc}$ reduces to the standard unbiased SCM. Also notice that $\hat{\bSigma}^{(u)}_{\Xmatsc}$ can be written as a statistical functional ${\Psimat}_{\Xmatsc}^{(u)}[\cdot]$ of the empirical probability measure $\hat{P}_{\Xmatsc}$ defined in (\ref{EmpProbMes}), i.e.,
\begin{equation}
\label{StatFunc} 
\hat{\bSigma}^{\left(u\right)}_{\Xmatsc}=\frac{{\rm{E}}[(\Xmat-\etavec_{\Xmatsc}^{(u)}[\hat{P}_{\Xmatsc}])(\Xmat-\etavec_{\Xmatsc}^{(u)}[\hat{P}_{\Xmatsc}])^{H};\hat{P}_{\Xmatsc}]}{{\rm{E}}[{u}\left(\Xmat\right);\hat{P}_{\Xmatsc}]}
\triangleq{\Psimat}_{\Xmatsc}^{(u)}[\hat{P}_{\Xmatsc}],
\end{equation} 
where 
\begin{equation}
\label{StatFuncMean}
\etavec_{\Xmatsc}^{(u)}[\hat{P}_{\Xmatsc}]\triangleq\frac{{\rm{E}}[\Xmat{u}\left(\Xmat\right);\hat{P}_{\Xmatsc}]}{{\rm{E}}[{u}\left(\Xmat\right);\hat{P}_{\Xmatsc}]}. 
\end{equation}
By (\ref{VarPhiDef}), (\ref{MTCovZ}) and  (\ref{StatFunc}), when $\hat{P}_{\Xmatsc}$ is replaced by ${P}_{\Xmatsc}$ we have ${\Psimat}_{\Xmatsc}^{(u)}[P_{\Xmatsc}]=\bSigma^{(u)}_{\Xmatsc}$, which implies that $\hat{\bSigma}^{(u)}_{\Xmatsc}$ is Fisher consistent \cite{Cox}. 
\subsection{Robustness of the empirical MT-covariance}
Here, we study the robustness of the empirical MT-covariance $\hat{\bSigma}^{(u)}_{\Xmatsc}$  using its matrix valued influence function \cite{Hampel}.
Define the probability measure $P_{\epsilon}\triangleq(1-\epsilon)\px+\epsilon\delta_{\yvec}$, where $0\leq\epsilon\leq1$, $\yvec\in\Csp^{p}$, and
$\delta_{\yvec}$ is the Dirac probability measure at $\yvec$. The influence function of a Fisher consistent estimator with statistical functional $\textrm{H}[\cdot]$ at probability distribution $\px$ is defined as \cite{Hampel}:
\begin{equation}
\label{IFDef}
{\rm{IF}}_{\textrm{H}}\left(\yvec;\px\right)\triangleq\lim\limits_{\epsilon\rightarrow{0}}\frac{\textrm{H}\left[P_{\epsilon}\right]-\textrm{H}\left[\px\right]}{\epsilon}
=\left.\frac{\partial{\textrm{H}}\left[{P}_{\epsilon}\right]}{\partial{\epsilon}}\right|_{\epsilon=0}.
\end{equation}
The influence function describes the effect on the estimator of an infinitesimal contamination at the point $\yvec$. An estimator is said to be B-robust if its influence function is bounded \cite{Hampel}. Using (\ref{StatFunc}) and (\ref{IFDef}) one can verify that the influence function of $\hat{\bSigma}^{(u)}_{\Xmatsc}$ is given by
\begin{equation}
\label{MT_COV_INF}
{\rm{IF}}_{{\tiny{\Psimat}}_{\xvec}^{(u)}}\left(\yvec;\px\right)
=\frac{{u\left(\yvec\right)}
[(\yvec-\muvec^{(u)}_{\Xmatsc})(\yvec-\muvec^{(u)}_{\Xmatsc})^{H} -  {\bSigma}^{(u)}_{\Xmatsc}]}
{{\rm{E}}[u(\Xmat);\px]}.
\end{equation}
The following proposition states a sufficient condition for boundedness of (\ref{MT_COV_INF}).  This condition is satisfied for the Gaussian MT-function proposed in Section \ref{GMTMUSIC}.
\begin{Proposition}
\label{RobustnessConditions}
The influence function (\ref{MT_COV_INF}) is bounded if the MT-function $u(\yvec)$ and the product $u(\yvec)\|\yvec\|^{2}_{2}$ are bounded over $\Csp^p$. [A proof is given in Appendix \ref{InfBoundProof}]
\end{Proposition}
\subsection{The robust MT-MUSIC procedure}
\label{RobMTMUS}
In robust MT-MUSIC the measure transformation (\ref{MeasureTransform}) is applied to the probability distribution $\px$ of the array output $\Xmat\left(n\right)$ (\ref{ArrayModel}).  
The MT-function $u\left(\cdot\right)$ is selected such that the following conditions are satisfied:
\begin{enumerate}
\item
\label{Cond1} 
The resulting empirical MT-covariance $\hat{\bSigma}^{\left(u\right)}_{\Xmatsc}$ is B-robust.
\item
\label{Cond2}
Let $\lambda^{\left(u\right)}_{1}\geq\cdots\geq\lambda^{\left(u\right)}_{p}$ denote the eigenvalues of  ${\bSigma}^{(u)}_{\Xmatsc}$. The $p-q$ smallest eigenvalues of ${\bSigma}^{(u)}_{\Xmatsc}$ satisfy:
\begin{equation}
\label{EigCond}
\lambda_{q}^{(u)}>\lambda^{\left(u\right)}_{q+1}=\cdots=\lambda^{\left(u\right)}_{p},
\end{equation}
and their corresponding eigenvectors span the null-space of $\Amat^{H}$, also called the noise subspace.
\end{enumerate} 

Let $\hat{\Vmat}^{(u)}\in\Csp^{p\times{(p-q)}}$ denote the matrix comprised of $p-q$ eigenvectors of $\hat{\bSigma}^{(u)}_{\Xmatsc}$ corresponding to its smallest eigenvalues. The DOAs are estimated by finding the $q$ highest maxima of the measure-transformed pseudo-spectrum:
\begin{equation}
\label{MTEmpSpectrum}
\hat{P}^{(u)}(\theta)\triangleq{\left\|\hat{\Vmat}^{(u)H}\avec(\theta)\right\|}^{-2}_{2}.
\end{equation}
By modifying the MT-function $u(\cdot)$ such that the stated conditions \ref{Cond1} and \ref{Cond2}  are satisfied a family of robust MT-MUSIC algorithms can be obtained. 
In particular, by choosing $u\left(\xvec\right)\propto\left\|\xvec\right\|^{-2}_{2}$ one can verify using (\ref{MTCovZ}) that for zero-centered symmetric distributions the resulting MT-covariance is proportional to the sign-covariance, proposed for the robust MUSIC generalization in \cite{Visuri}. Another particular choice of MT-function leading to the Gaussian MT-MUSIC algorithm is discussed in the following section.
\section{The Gaussian MT-MUSIC}
\label{GMTMUSIC}
In this section, we parameterize the MT-function $u\left(\cdot;\tau\right)$, with scale parameter $\tau\in\Rsp_{++}$
under the Gaussian family of functions centered at the origin. This results in a B-robust empirical MT-covariance matrix that rejects large outliers. Under the assumption of spherically contoured noise distribution, we show that the noise subspace can be determined from the eigendecomposition of the MT-covariance. Choice of the scale parameter $\tau$ is also discussed.
\subsection{The Gaussian MT-function}
\label{GaussMTFunc}
We define the Gaussian MT-function $u_{\rm{G}}\left(\cdot;\cdot\right)$ as 
\begin{eqnarray}
\label{GaussKernel}  
\uGausss\left(\xvec;\tau\right)\triangleq\left(\pi\tau^{2}\right)^{-p}\exp\left(-{\left\|\xvec\right\|^{2}_{2}}/{\tau^{2}}\right),\hspace{0.2cm}\tau\in\Rsp_{++}.
\end{eqnarray}
Using (\ref{VarPhiDef})-(\ref{MTCovZ}) and (\ref{GaussKernel}) one can verify that the resulting Gaussian MT-covariance always takes finite values. 
Additionally, notice that the Gaussian MT-function satisfies the condition (\ref{Cond11}) in Proposition \ref{ConsistentEst}, and therefore, the empirical Gaussian MT-covariance, based on i.i.d. samples from any probability distribution $\px$,  is strongly consistent. For any fixed scale parameter $\tau$, the Gaussian MT-function also satisfies the condition in proposition \ref{RobustnessConditions}, resulting in a B-robust empirical Gaussian MT-covariance $\hat{\bSigma}^{\left(\uGausss\right)}_{\Xmatsc}\left(\tau\right)$. The following proposition, which follows directly from (\ref{MT_COV_INF}) and (\ref{GaussKernel}), states that the Frobenius norm of the corresponding influence function approaches zero as the contamination norm approaches infinity.
\begin{Proposition}
\label{InfFuncLim}
For any fixed scale parameter $\tau$ of the Gaussian MT-function (\ref{GaussKernel}), the influence function of the resulting empirical Gaussian MT-covariance satisfies
\begin{equation}
\label{IFLim}
{\left\|{\rm{IF}}_{{\tiny{\Psimat}}_{\xvec}^{(u_{\rm{G}})}}\left(\yvec;\px\right)\right\|}_{\rm{Fro}}\rightarrow{0}\hspace{0.2cm}\textrm{as}\hspace{0.2cm}{\left\|\yvec\right\|}_{2}\rightarrow\infty, 
\end{equation}
where $\left\|\cdot\right\|_{\rm{Fro}}$ denotes the Frobenius norm. [A proof is given in Appendix \ref{InfFuncLimProof}]
\end{Proposition}
Thus, unlike the SCM and other robust covariance approaches, the empirical Gaussian MT-covariance rejects large outliers. This property is illustrated in Fig. \ref{INF_FRO} for a standard bivariate complex normal distribution, as compared to  the empirical sign-covariance, Tyler's scatter M-estimator, and the SCM.
\begin{figure}[htbp!]
  \begin{center}
    {\includegraphics[scale=0.35]{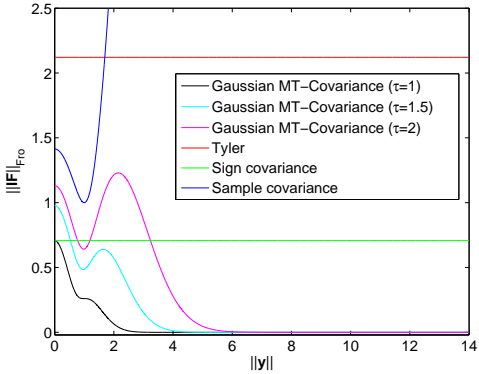}}
  \end{center}
  \caption{Frobenius norms of the influence functions associated with the empirical Gaussian MT-covariance for $\tau=1$, $\tau=1.5$ and $\tau=2$, Tyler's scatter M-estimator, the empirical sign-covariance, and the SCM, versus the contamination norm, for a bivariate standard complex normal distribution. Notice that the influence function approaches zero for large $\|\mathbf y\|$ only for the proposed Gaussian MT-covariance estimator, indicating enhanced robustness to outliers.}
\label{INF_FRO}
\end{figure}

Notice that as the scale parameter $\tau$ of the Gaussian MT-function (\ref{GaussKernel}) approaches infinity, the corresponding empirical Gaussian MT-covariance $\hat{\bSigma}^{\left(\uGausss\right)}_{\Xmatsc}\left(\tau\right)$ approaches the non-robust standard SCM $\hat{\bSigma}_{\Xmatsc}$, whose influence function is unbounded. On the other hand, as $\tau$ decreases it can be shown using the upper bound in (\ref{InfFuncGauss}) that the influence function of $\hat{\bSigma}^{\left(\uGausss\right)}_{\Xmatsc}\left(\tau\right)$ has a faster asymptotic decay, as illustrated in Fig. \ref{INF_FRO}, i.e., the empirical Gaussian MT-covariance becomes more resilient to large outliers. However, we note that this may come at the expense of information loss. The trade-off between robustness and information loss is discussed in Subsection \ref{TauChoice}.
\subsection{The Gaussian MT-covariance for spherically distributed noise}
\label{CGNoise}
We assume that the noise component in (\ref{ArrayModel}) has a complex spherically contoured distribution, also known as a spherical distribution \cite{Visa} having stochastic representation:
\begin{equation}
\label{CompGauss}
\Wmat\left(n\right)=\nu\left(n\right)\bzeta\left(n\right),
\end{equation}
where $\nu\left(n\right)\triangleq\rho\left(n\right)/\left\|\bzeta\left(n\right)\right\|_{2}$, $\rho\left(n\right)\in\Rsp_{++}$ is a first-order stationary process, and $\bzeta\left(n\right)\in\Csp^{p}$ is a proper-complex wide-sense stationary Gaussian process with zero-mean and unit covariance, which is statistically independent of $\rho\left(n\right)$. The stochastic representation (\ref{CompGauss}) is a direct consequence of the following properties \cite{Visa}:
\begin{inparaenum}
\item
Any spherically distributed complex random vector $\Wmat$ can be represented as $\Wmat=\rho\Umat$, where $\rho$ is a strictly positive random variable, and $\Umat$  is uniformly distributed on the complex unit sphere and statistically independent of $\rho$.
\item
Any random vector $\Umat$ that is uniformly distributed on the complex unit sphere can be represented as $\Umat=\bzeta/\left\|\bzeta\right\|_{2}$, where $\bzeta$ is a complex random vector with zero-mean spherically contoured distribution, for example a complex Gaussian random vector with zero-mean and unit covariance.
\end{inparaenum}

The structure of the resulting Gaussian MT-covariance of the array output is given in the following theorem. 
\begin{Theorem}
\label{CompGaussStruct}   
Under the array model (\ref{ArrayModel}) and the spherical noise assumption (\ref{CompGauss}), the Gaussian MT-covariance of $\Xmat\left(n\right)$  takes the form:
\begin{equation}
\label{GaussMTCov}
\bSigma^{\left(\uGausss\right)}_{\Xmatsc}\left(\tau\right)=\Amat\bSigma^{\left(g\right)}_{\alpha^{2}\Smatsc}\left(\tau\right)\Amat^{H}+\sigma^{2\left(h\right)}_{\alpha\Wmatsc}\left(\tau\right)\Imat,
\end{equation} 
where $\bSigma^{\left(g\right)}_{\alpha^{2}\Smatsc}\left(\tau\right)$ is a non-singular covariance matrix of the scaled signal component $\alpha^{2}\left(n\right)\Smat\left(n\right)$, $\alpha\left(n\right)\triangleq\sqrt{\frac{\tau^{2}}{\tau^{2}+\nu^{2}\left(n\right)}}$, under the transformed joint probability measure $Q^{\left(g\right)}_{\alpha,\Smatsc}$ with the MT-function  
$g\left(\alpha,\Smat;\tau\right)\triangleq{(\frac{\pi\tau^{2}}{\alpha^{2}})}^{-p}\exp(-{\alpha^{2}\|\Amat\Smat\|^{2}_{2}}/{\tau^{2}})$. The term 
$\sigma^{2\left(h\right)}_{\alpha\Wmatsc}\left(\tau\right)$, multiplying the identity matrix $\Imat$, is the variance of the scaled noise component $\alpha\left(n\right)\Wmat\left(n\right)$ under the transformed joint probability measure $Q^{\left(h\right)}_{\alpha,\Wmat}$ with the MT-function $h\left(\alpha;\tau\right)\triangleq{\rm{E}}\left[g\left(\alpha,\Smat;\tau\right);P_{\Smatsc}\right]$. [A proof is given in Appendix \ref{CompGaussStructProof}]
\end{Theorem}
Thus, by the structure (\ref{GaussMTCov}) and the facts that the steering matrix $\Amat$ has full column rank and the MT-covariance $\bSigma^{\left(g\right)}_{\alpha^{2}\Smatsc}\left(\tau\right)$ is non-singular, we conclude that Condition \ref{Cond2}
in Subsection \ref{RobMTMUS} is satisfied.
\subsection{The Gaussian MT-MUSIC algorithm}
The empirical Gaussian MT-covariance is B-robust, and, under the spherical noise assumption (\ref{CompGauss}), the noise subspace can be determined from the eigendecomposition of the  Gaussian MT-covariance. The Gaussian MT-MUSIC algorithm is implemented by replacing the MT-function in (\ref{MTEmpSpectrum}) with the Gaussian MT-function (\ref{GaussKernel}). 
\subsection{Choosing the scale parameter of the Gaussian MT-function}
\label{TauChoice}
While de-emphasizing non-informative outliers, e.g., caused by heavy-tailed distributions, the empirical Gaussian MT-covariance is less informative than the standard sample-covariance when the observations are normally distributed. This is seen in the following theorem that follows from the Gaussian Fisher information formula \cite{Schreier} and elementary trace inequalities \cite{TraceIneq}.
\begin{Theorem}
\label{TauProp}
Assume that the probability distribution $\px$ of the array outputs (\ref{ArrayModel}) is proper complex normal. The ratio between the Fisher information for estimating $\theta_{k}\in\{\theta_{1},\ldots,\theta_{q}\}$ under the transformed probability measure $Q^{\left(\uGausss\right)}_{\Xmatsc}$ (with the MT-function (\ref{GaussKernel})) and the corresponding Fisher information under $\px$ satisfy:
\begin{equation}
\label{FIMRatio}
\frac{\tau^{4}}{({\lambda_{\rm{max}}}\left(\bSigma_{\Xmatsc}\right)+\tau^{2})^{2}}\leq\frac{{F}(\theta_{k};{Q^{\left(\uGausss\right)}_{\xvec}})}{F\left(\theta_{k};{P_{\Xmatsc}}\right)}  
\leq\frac{\tau^{4}}{(\lambda_{\rm{min}}\left(\bSigma_{\Xmatsc}\right)+\tau^{2})^{2}},
\end{equation}
where $\lambda_{\rm{min}}\left(\cdot\right)$ and $\lambda_{\rm{max}}\left(\cdot\right)$ are the minimum and maximum eigenvalues, respectively.
[A proof is given in Appendix \ref{TauPropProof}]
\end{Theorem}
Therefore, in order to prevent a significant transform-domain Fisher information loss when the observations are normally distributed, we propose to choose the following safe-guard scale parameter:
\begin{equation}
\label{TAUSAFE}  
\tau=\sqrt{c{\lambda_{\rm{max}}}\left(\bSigma_{\Xmatsc}\right)}, 
\end{equation}
where $c$ is some positive constant that guarantees that the Fisher information ratio (\ref{FIMRatio}) is greater than $(c/(1+c))^{2}$. Since in practice $\bSigma_{\Xmatsc}$ is unknown, it is replaced by the following empirical robust estimate that is based on its relation (\ref{GaussCovGauss}) to the Gaussian MT-covariance for normally distributed observations:
\begin{equation}
\label{RobCovEst}
\hat{\bSigma}_{\Xmatsc}=\tau^{2}\hat{\bSigma}^{(\uGausss)}_{\Xmatsc}\left(\tau\right)\left(\tau^{2}\Imat-\hat{\bSigma}^{(\uGausss)}_{\Xmatsc}\left(\tau\right)\right)^{-1},
\end{equation}
where the empirical Gaussian MT-covariance $\hat{\bSigma}^{(\uGausss)}_{\Xmatsc}\left(\tau\right)$ is obtained using (\ref{Rx_u_Est}), and $\tau$ must be greater than ${\lambda_{\rm{max}}}\left(\hat{\bSigma}^{(\uGausss)}_{\Xmatsc}\left(\tau\right)\right)$ in order to guarantee positive definiteness of $\hat{\bSigma}_{\Xmatsc}$. Therefore, substitution of (\ref{RobCovEst}) into (\ref{TAUSAFE}) results in the following data-driven selection rule:
\begin{equation}
\label{TAUSAFE_EMP}
\tau=\sqrt{\left(c+1\right){\lambda_{\rm{max}}}\left(\hat{\bSigma}^{(\uGausss)}_{\Xmatsc}\left(\tau\right)\right)},
\end{equation}
which can be solved numerically, e.g., using fixed-point iteration \cite{FXP}.

In the general case, when the observations are not necessarily Gaussian, the selection rule (\ref{TAUSAFE}) controls the amount of second-order statistical information loss caused by the measure transformation. Increasing the constant $c$ increases the scale parameter $\tau$ and reduces the information loss, while on the other hand, makes the estimator more sensitive to large-norm outliers, as illustrated in Fig. \ref{INF_FRO}.
\section{Estimation of the number of signals}
\label{MTMDLEST}
We estimate the number of signals using a measure-transformed version of the minimum description length (MDL) criterion \cite{Wax}, called MT-MDL, that is obtained by replacing the eigenvalues of the SCM with the eigenvalues of the empirical MT-covariance. The MT-MDL criterion takes the form:
\begin{eqnarray}
\label{MTMDL}
{\rm{MDL}}^{\left(u\right)}\left(k\right)&=&-\log\left(\frac{\left(\prod\limits_{m=k+1}^{p}\hat{\lambda}^{\left(u\right)}_{m}\right)^{\frac{1}{p-k}}}{\frac{1}{p-k}\sum\limits_{m=k+1}^{p}\hat{\lambda}^{\left(u\right)}_{m}}\right)^{({p-k}){N}}
\\\nonumber&+&\frac{1}{2}k\left(2p-k\right)\log{N},
\end{eqnarray}
where $\hat{\lambda}^{\left(u\right)}_{1}\geq\ldots\geq\hat{\lambda}^{\left(u\right)}_{p}$ denote the eigenvalues of $\hat{\bSigma}^{\left(u\right)}_{\Xmatsc}$ and $N$ is the number of observations (snapshots). The estimated number of signals, $\hat{q}$, is obtained by minimizing (\ref{MTMDL}) over $k\in\left\{0,\ldots,p-1\right\}$. 

Under the conditions that the eigenvalues of the SCM are strongly consistent with asymptotic convergence rate of $O\left(\sqrt{N^{-1}\log\log{N}}\right)$ and that the $p-q$ smallest eigenvalues of the covariance matrix are equal and separated from its $q$ largest eigenvalues, it has been shown in \cite{Zhao} that minimization of the MDL criterion leads to strongly consistent estimates of the number of signals for any underlying probability distribution of the data. Thus, when the eigenvalues of $\hat{\bSigma}^{\left(u\right)}_{\Xmatsc}$ and ${\bSigma}^{\left(u\right)}_{\Xmatsc}$  satisfy these conditions, namely $\hat{\lambda}^{\left(u\right)}_{k}$ converges almost surely to ${\lambda}^{\left(u\right)}_{k}$ for all $k=1,\ldots,p$ with the same asymptotic convergence rate as the eigenvalues of the SCM, and the eigenvalues of ${\bSigma}^{\left(u\right)}_{\Xmatsc}$ satisfy (\ref{EigCond}), the resulting MT-MDL based estimator, $\hat{q}$, must be strongly consistent. This rationale is used for proving the following Theorem that states a sufficient condition for strong consistency of the estimator $\hat{q}$. 
\begin{Theorem}
\label{Th2}
Let $\Xmat\left(n\right)$, $n=1,\ldots,N$ denote a sequence of i.i.d. samples from the probability distribution $\px$ of the array output (\ref{ArrayModel}), with MT-covariance $\bSigma^{\left(u\right)}_{\Xmatsc}$ whose eigenvalues satisfy (\ref{EigCond}). If 
\begin{equation}
\label{Cond3}
{\rm{E}}\left[u^{2}\left(\Xmat\right);\px\right]<\infty \hspace{0.2cm} {\rm{and}}\hspace{0.2cm}{\rm{E}}\left[\left\|\Xmat\right\|^{4}_{2}u^{2}\left(\Xmat\right);\px\right]<\infty,
\end{equation}
then $\hat{q}\rightarrow{q}$ a.s. as $N\rightarrow\infty$. [A proof is given in Appendix \ref{Th2Proof}]
\end{Theorem}

Notice that the Gaussian MT-function (\ref{GaussKernel}) always satisfies the condition (\ref{Cond3}). Furthermore, as shown in subsection \ref{CGNoise}, the Gaussian MT-covariance $\bSigma^{\left(\uGausss\right)}_{\Xmatsc}\left(\tau\right)$ satisfies Condition \ref{Cond2} in Subsection \ref{RobMTMUS} for spherically distributed noise. Therefore, in this case, minimization of the MT-MDL criterion with the Gaussian MT-function (Gaussian MT-MDL) results in robust and strongly consistent estimate of the number of signals. We propose to choose the scale parameter $\tau$ of the Gaussian MT-function using the same selection rule (\ref{TAUSAFE_EMP}) that prevents significant information loss for estimation of the DOAs (model parameters), and does not require any knowledge about the number of signals (model order). The idea of estimating the DOAs and the number of signals using the same scale parameter $\tau$, i.e., under the same transformed probability measure, arises from the intuition that if there is no significant information loss for estimating the model parameters, then there will be no significant information loss for estimating the model order.   
\section{The spatially smoothed Gaussian MT-MUSIC for coherent signals}
\label{CMTMUSIC}
In this section, we consider the case of coherent signals contaminated by spherically distributed noise. In this scenario, the components of the latent vector $\Smat\left(n\right)$ in (\ref{ArrayModel}) are phase-delayed amplitude-weighted replicas of a single first-order stationary random signal $s(n)$, i.e.,  
\begin{equation}
\label{CoherentModel}
\Smat\left(n\right)=\bxi{s}\left(n\right),
\end{equation}
where $\bxi\in\Csp^{q}$ is a vector of deterministic complex attenuation coefficients. Similarly to the standard covariance $\bSigma_{\Xmatsc}$, the noise subspace cannot be determined from the eigendecomposition of  the Gaussian MT-covariance $\bSigma^{\left(\uGausss\right)}_{\Xmatsc}\left(\tau\right)$, and therefore, the Gaussian MT-MUSIC will fail in estimating the DOAs. Fortunately, similarly to \cite{Pillai},  we show that for uniform linearly spaced array (ULA) \cite{VanTrees} the DOAs can be determined using a spatially smoothed version of the Gaussian MT-covariance. 

We partition $\bSigma^{\left(\uGausss\right)}_{\Xmatsc}\left(\tau\right)$ into $L=p-r+1$ overlapping forward and backward square sub-matrices of dimension $r<q$. The entries of the $l$-th forward sub-matrix are given by
\begin{equation}
\label{FOR}
\left[\Cmat^{\left(\uGausss\right)}_{f,l}\left(\tau\right)\right]_{j,k}\triangleq\left[\bSigma^{\left(\uGausss\right)}_{\Xmatsc}\left(\tau\right)\right]_{l+j-1,l+k-1},
\end{equation}
$j,k=1,\ldots,r$, and the entries of the $l$-th backward sub-matrix are given by
\begin{equation}
\label{BCK}   
\left[\Cmat^{\left(\uGausss\right)}_{b,l}\left(\tau\right)\right]_{j,k}\triangleq\left[\bSigma^{\left(\uGausss\right)}_{\Xmatsc}\left(\tau\right)\right]^{*}_{p-l-j+2,p-l-k+2},
\end{equation}
$j,k=1,\ldots,r$, where $\left(\cdot\right)^{*}$ denotes the complex conjugate. These matrices correspond to overlapping forward and backward subarrays of size $r$, respectively. The forward spatially smoothed Gaussian MT-covariance is defined as the average over the $L$ forward sub-matrices:
\begin{equation}
\label{FGMTC}
\Cmat^{(\uGausss)}_{f}\left(\tau\right)\triangleq\frac{1}{L}\sum\limits_{l=1}^{L}\Cmat^{\left(\uGausss\right)}_{f,l}\left(\tau\right).
\end{equation}
Similarly, the backward spatially smoothed Gaussian MT-covariance is defined as the average over the $L$ backward sub-matrices:
\begin{equation}
\label{BGMTC} 
\Cmat^{(\uGausss)}_{b}\left(\tau\right)\triangleq\frac{1}{L}\sum\limits_{l=1}^{L}\Cmat^{\left(\uGausss\right)}_{b,l}\left(\tau\right).
\end{equation}
The forward/backward spatially smoothed Gaussian MT-covariance matrix is given by
\begin{equation}
\label{FBGMTC}  
\Cmat^{\left(\uGausss\right)}_{f/b}\left(\tau\right)\triangleq\frac{1}{2}\left({\Cmat^{(\uGausss)}_{f}\left(\tau\right) + \Cmat^{(\uGausss)}_{b}\left(\tau\right)}\right).
\end{equation}

We define $\Bmat\triangleq\left[\bvec\left(\theta_{1}\right),\ldots,\bvec\left(\theta_{q}\right)\right]\in\Csp^{r\times{q}}$ as the steering matrix of a ULA with $r<p$ sensors, where $\bvec\left(\theta\right)\triangleq\left[1,e^{-i2\pi{(d/\lambda)}\sin\left(\theta\right)},\ldots,e^{-i2\pi(r-1){(d/\lambda)}\sin\left(\theta\right)}\right]^{T}$ is the array steering vector toward direction $\theta$, and $d$ is the sensors spacing, i.e., $\Bmat$ is a sub-matrix of the steering matrix $\Amat$ in (\ref{ArrayModel}) comprised of its first $r$ rows.
The following Proposition states sufficient conditions under which the null-space of $\Bmat^{H}$ can be determined from the eigendecomposition of $\Cmat^{\left(\uGausss\right)}_{f/b}\left(\tau\right)$ (\ref{FBGMTC}). The same conditions were proved in  \cite{Pillai} for the spatially smoothed SCM.
\begin{Proposition}
\label{FBSTh}
Define $\Hmat\triangleq\Gmat^{-1}_{f}\Gmat_{b}$, where $\Gmat_{f}\triangleq{\rm{diag}}\left(\bxi\right)$, $\Gmat_{b}\triangleq{\rm{diag}}\left(\deltavec\right)$, $\deltavec\triangleq\left(\Dmat^{p-1}\bxi\right)^{*}$, and $\Dmat^{l}$ is the $l$-th power of the diagonal matrix 
$\Dmat\triangleq{\rm{diag}}\left(\left[v\left(\theta_{1}\right),\ldots,v\left(\theta_{q}\right)\right]\right),\hspace{0.2cm}v\left(\theta\right)\triangleq{\exp}\left({-i2\pi{(d/\lambda)}\sin\left(\theta\right)}\right)$.
Additionally, let $\lambda^{\left(\uGausss\right)}_{1}\geq\cdots\geq\lambda^{\left(\uGausss\right)}_{r}$ denote the eigenvalues of  $\Cmat^{\left(\uGausss\right)}_{f/b}\left(\tau\right)$. If the dimension $r$ of the forward and backward sub-matrices (\ref{FOR}) and (\ref{BCK}) is chosen such that the resulting number of sub-matrices in each direction (forward or backward) $L$ satisfies: 
\begin{enumerate}
\item
$L\geq{q}$, or
\item
$2L\geq{q}$ and  the largest subset of equal diagonal entries of $\Hmat$ is at most of size $L$,
\end{enumerate}
then the $r-q$ smallest eigenvalues of $\Cmat^{\left(\uGausss\right)}_{f/b}\left(\tau\right)$ satisfy $\lambda^{\left(\uGausss\right)}_{q}>\lambda^{\left(\uGausss\right)}_{q+1}=\cdots=\lambda^{\left(\uGausss\right)}_{r}$, and their corresponding eigenvectors span the null-space of $\Bmat^{H}$.  [A proof is given in Appendix \ref{FBSThProof}]
\end{Proposition}

Hence, by proper choice of $r$, such that either one of the stated conditions above is satisfied, the spatially smoothed Gaussian MT-MUSIC is obtained by replacing the empirical Gaussian MT-covariance $\hat{\bSigma}^{\left(\uGausss\right)}_{\Xmatsc}\left(\tau\right)$ with its spatially smoothed version $\hat{\Cmat}^{\left(\uGausss\right)}_{f/b}\left(\tau\right)$. 

The number of signals is estimated using a measure-transformed version of the modified MDL (MMDL) criterion used in \cite{Xu} for cases where forward/backward spatial smoothing is performed. Similarly to (\ref{MTMDL}), this criterion, called here Gaussian MT-MMDL, is obtained by replacing the eigenvalues of the spatially smoothed SCM with the eigenvalues of $\hat{\Cmat}^{\left(\uGausss\right)}_{f/b}\left(\tau\right)$. The Gaussian MT-MMDL criterion takes the form: 
\begin{eqnarray}
\label{GMTMDL}
{\rm{MDL}}^{\left(\uGausss\right)}_{f/b}\left(k\right)&=&-\log\left(\frac{\left(\prod\limits_{m=k+1}^{r}\hat{\lambda}^{\left(\uGausss\right)}_{m}\right)^{\frac{1}{r-k}}}{\frac{1}{r-k}\sum\limits_{m=k+1}^{r}\hat{\lambda}^{\left(\uGausss\right)}_{m}}\right)^{({r-k}){N}}
\\\nonumber&+&\frac{1}{4}k\left(2r-k+1\right)\log{N},
\end{eqnarray}
where $\hat{\lambda}^{\left(\uGausss\right)}_{1}\geq\ldots\geq\hat{\lambda}^{\left(\uGausss\right)}_{r}$ denote the eigenvalues of $\hat{\Cmat}^{\left(\uGausss\right)}_{f/b}\left(\tau\right)$ and $N$ is the number of observations (snapshots). The estimated number of signals, $\hat{q}$, is obtained by minimizing (\ref{GMTMDL}) over $k\in\left\{0,\ldots,r-1\right\}$.

Finally, we choose the scale parameter $\tau$ of the Gaussian MT-function using the selection rule (\ref{TAUSAFE_EMP}) with $\hat{\Cmat}^{\left(\uGausss\right)}_{f/b}\left(\tau\right)$ instead of $\hat{\bSigma}^{\left(\uGausss\right)}_{\Xmatsc}\left(\tau\right)$. 
\section{Numerical examples}
\label{Examples}
We evaluate and compare the performance of the proposed MT-MUSIC DOA estimator and the MT-MDL order estimator. What follows is a  summary of these comparisons.
The DOAs estimation performances are evaluated under the assumption that the number of signals is known. We perform a separate evaluation of the  proposed MT-MDL estimator of the number of signals.  We examine scenarios of non-coherent and coherent signals. For non-coherent signals, the Gaussian MT-MUSIC algorithm is compared to the standard SCM-based MUSIC (SCM-MUSIC) \cite{Schmidt} and to its robust generalizations based on the ZMNL preprocessing (ZMNL-MUSIC) \cite{swami2002}, the empirical sign-covariance (SGN-MUSIC) \cite{Visuri}, Tyler's scatter M-estimator (TYLER-MUSIC) \cite{Visa}, \cite{Tyler} and 
the maximum-likelihood (ML) estimators of scatter corresponding to each of the considered non-Gaussian noise distributions (ML-MUSIC) \cite{Visa}, \cite{MestMDL}. The estimation performance of the number of signals using the MT-MDL criterion (\ref{MTMDL}) with the Gaussian MT-function is compared to estimators using the standard MDL criterion \cite{Wax} based on the standard SCM (SCM-MDL), and the MDL variants based on the SCM of the pre-processed data with the ZMNL function (ZMNL-MDL) \cite{Sadler}, the empirical sign-covariance (SGN-MDL) \cite{Visuri}, Tyler's scatter M-estimator (TYLER-MDL) \cite{MestMDL} and the ML estimators of scatter corresponding to each of the considered non-Gaussian noise distributions (ML-MDL) \cite{Visa}, \cite{MestMDL}.  For coherent signals, the spatially smoothed (SS) Gaussian MT-MUSIC algorithm, discussed in Section \ref{CMTMUSIC}, is compared to the spatially smoothed versions of the SCM-MUSIC \cite{Pillai}, ZMNL-MUSIC,  SGN-MUSIC \cite{Visuri}, TYLER-MUSIC and ML-MUSIC. Estimation performance of the number of signals using the Gaussian MT-MMDL criterion (\ref{GMTMDL}) is compared to those obtained by the modified MDL criterion for coherent signals \cite{Xu} based on the forward/backward spatially smoothed versions of the standard SCM, the SCM of the preprocessed data with the ZMNL function \cite{swami2002}, the empirical sign-covariance \cite{Visuri}, Tyler's scatter M-estimator and the ML estimators of scatter corresponding to each non-Gaussian noise distribution considered in the simulation examples. 

We consider the following $p$-variate complex spherical compound Gaussian noise distributions with zero location parameter and isotropic dispersion $\sigma^{2}_{\Wmatsc}\Imat$:  
Gaussian, Cauchy, $K$-distribution with shape parameter $\nu=0.75$, and compound Gaussian distribution with inverse Gaussian texture and shape parameter $\lambda=0.1$. Notice that unlike the Gaussian distribution, the other noise distributions are heavy tailed. Random sampling from the considered noise distributions and their applicability for modelling radar clutter are discussed in detail in \cite{Visa} and \cite{CG4}. Let $\sigma^{2}_{S_{k}}$, $k=1,\ldots,q$ denote the variances of the received signals, the generalized signal-to-noise-ratio (GSNR) is defined as $\textrm{GSNR}\triangleq{10}\log_{10}{\frac{1}{q}\sum_{k=1}^{q}\sigma^{2}_{S_{k}}}/{\sigma^{2}_{\Wmatsc}}$ and is used to index the estimation performance.

Following the approach proposed in subsection \ref{TauChoice}, for non-coherent signals, we select the scale parameter $\tau$ of the Gaussian MT-function (\ref{GaussKernel}) as the solution of (\ref{TAUSAFE_EMP}) with $c=5$. This choice of the constant $c$ guarantees relative transform-domain Fisher information loss of no more than $\approx30\%$. The solution of (\ref{TAUSAFE_EMP}) is obtained using fixed-point iteration with initial condition $\tau_{0}=5\sqrt{\frac{1}{p}\sum_{k=1}^{p}\hat{\sigma}^{2}_{X_{k}}}$, where $\hat{\sigma}^{2}_{X_{k}} = \gamma^{2}[(\textrm{MAD}(\{\textrm{Re}(X_{k,n})\}_{n=1}^{N}))^{2} + (\textrm{MAD}(\{\textrm{Im}(X_{k,n})\}_{n=1}^{N}))^{2}]$, $\gamma\triangleq{1}/\textrm{erf}^{-1}(3/4)$, is a robust median absolute deviation (MAD) estimate of variance \cite{Huber}. The maximum number of iterations and the stopping criterion were set to 100 and ${\left|\tau_{l}-\tau_{l-1}\right|}/{\tau_{l-1}}<10^{-6}$, respectively, where $l$ is an iteration index. For coherent signals, we replaced the empirical Gaussian MT-covariance in (\ref{TAUSAFE_EMP}) with its spatially smoothed version and applied the same selection procedure for $\tau$. 

The maximum number of iterations and the stopping criterion in Tyler's scatter M-estimator and the ML estimators of scatter were set to 100 and 
$${\|\hat{\bSigma}^{\textrm{Tyler/ML}}_{\Xmatsc,l}-\hat{\bSigma}^{\textrm{Tyler/ML}}_{\Xmatsc,l-1}\|_{\rm{Fro}}}/{\|\hat{\bSigma}^{\textrm{Tyler/ML}}_{\Xmatsc,l-1}\|_{\rm{Fro}}}<10^{-6},$$ respectively. 
 
In all examples, the performances versus GSNR were evaluated for $N=1000$ i.i.d. snapshots.
The performances versus the number of snapshots were evaluated at the threshold GSNR point obtained by the Gaussian MT-MUSIC algorithm for $N=1000$ i.i.d. snapshots. The parameter space $\Theta\triangleq\left[-90^{\circ},90^{\circ}\right)$ was sampled uniformly with sampling interval 
$\Delta=0.0018^{\circ}$. All performance measures were averaged over $10^{4}$ Monte-Carlo simulations.
\subsection{Non-coherent signals}
\label{NonCoherentSig} 
In this example, we considered five independent 4-QAM signals with equal power $\sigma^{2}_{\Smatsc}$ impinging on a $16$-element uniform linear array with $\lambda/2$ spacing from DOAs $\theta_{1}=-10^{\circ}$, $\theta_{2}=0^{\circ}$, $\theta_{3}=5^{\circ}$, $\theta_{4}=15^{\circ}$, and $\theta_{5}=35^{\circ}$. The average RMSEs for estimating the DOAs and the error rates for estimating the number of signals versus GSNR and the number of snapshots 
are depicted in Figs. \ref{GAUSS_NON_COHERENT}-\ref{IG_NON_COHERENT} for each noise distribution. Notice that for the Gaussian noise case, there is no significant performance gap between the compared methods. 
For the other noise distributions, the proposed Gaussian MT-MUSIC and the Gaussian MT-MDL based estimation of the number of signals outperform all other robust MUSIC and MDL generalizations in the low GSNR and low sample size regimes, with significantly lower threshold regions. This performance advantage may be  attributed to the fact that unlike the empirical sign-covariance, Tyler's scatter M-estimator, and the ML estimators of scatter corresponding to each non-Gaussian noise distribution, the influence function of the empirical Gaussian MT-covariance is negligible for large norm outliers (as illustrated in Fig. \ref{INF_FRO}), which are likely in low GSNRs and become more defective when the sample size decreases. Furthermore, unlike the ZMNL preprocessing based technique, the proposed measure-transformation approach preserves the noise subspace and effectively suppresses outliers without significant information loss for estimating the DOAs and the number of signals.
\begin{figure}[htp]
  \begin{center}
    {{\subfigure[]{\includegraphics[scale = 0.235]{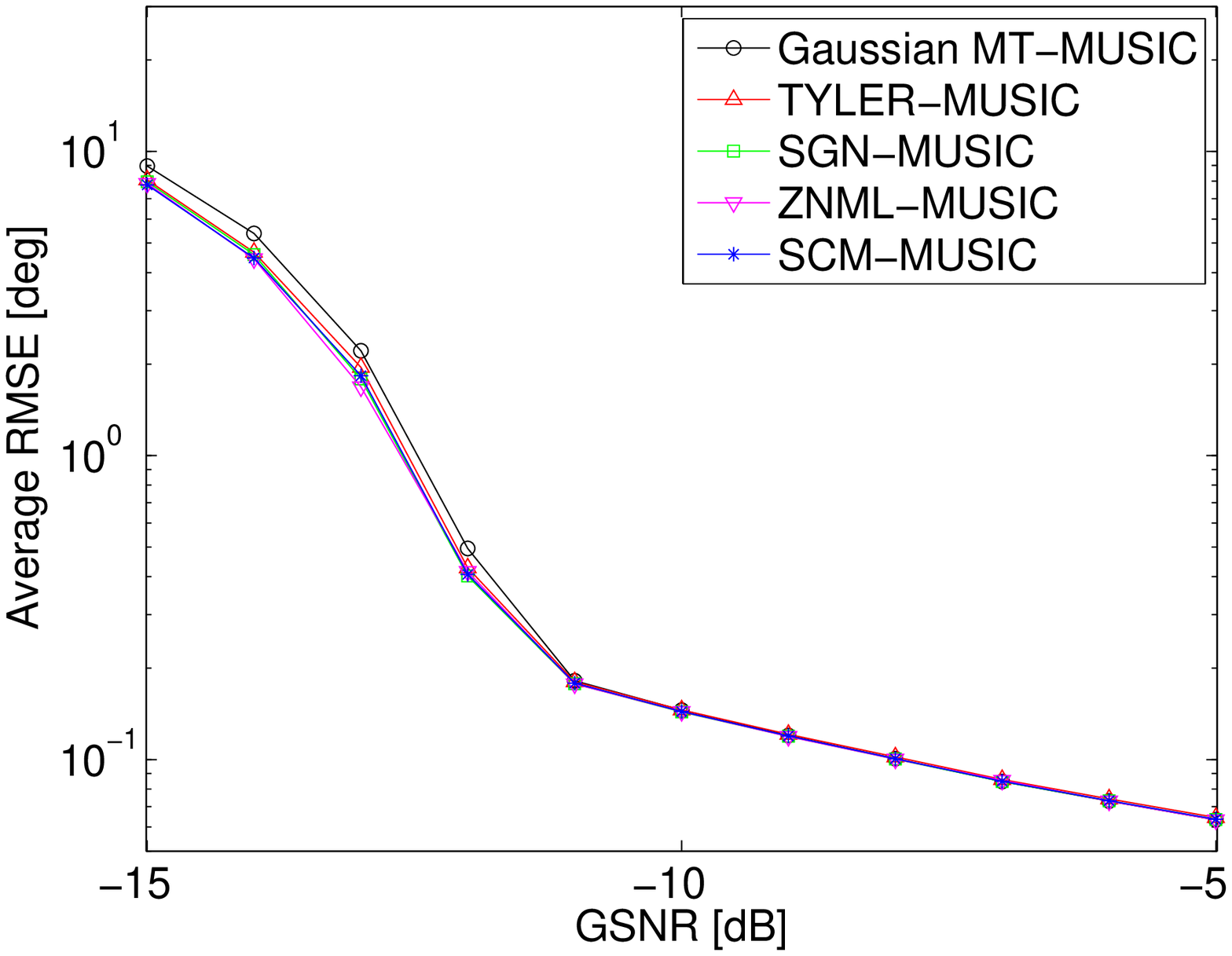}}}}
    {{\subfigure[]{\includegraphics[scale = 0.235]{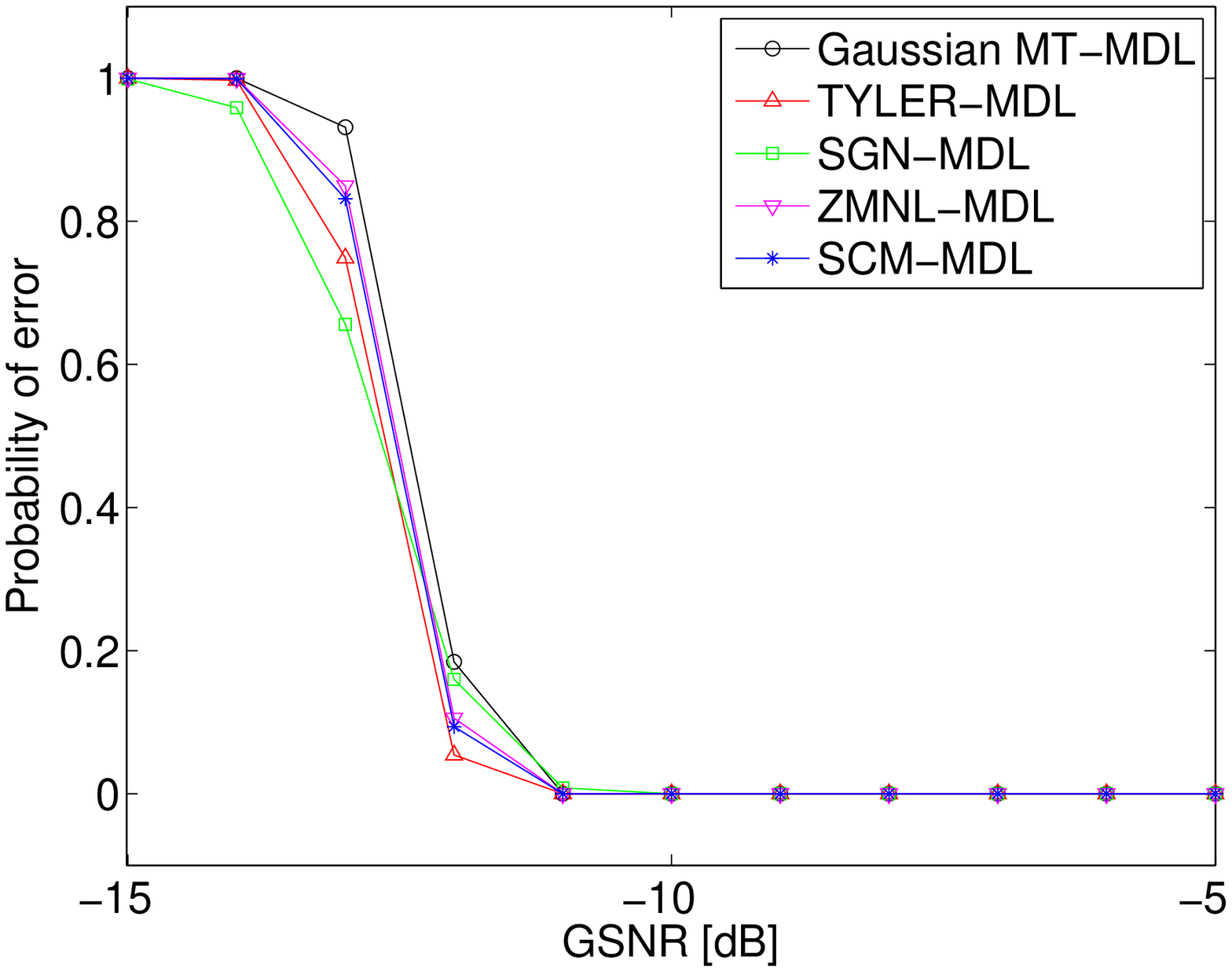}}}}
    {{\subfigure[]{\includegraphics[scale = 0.235]{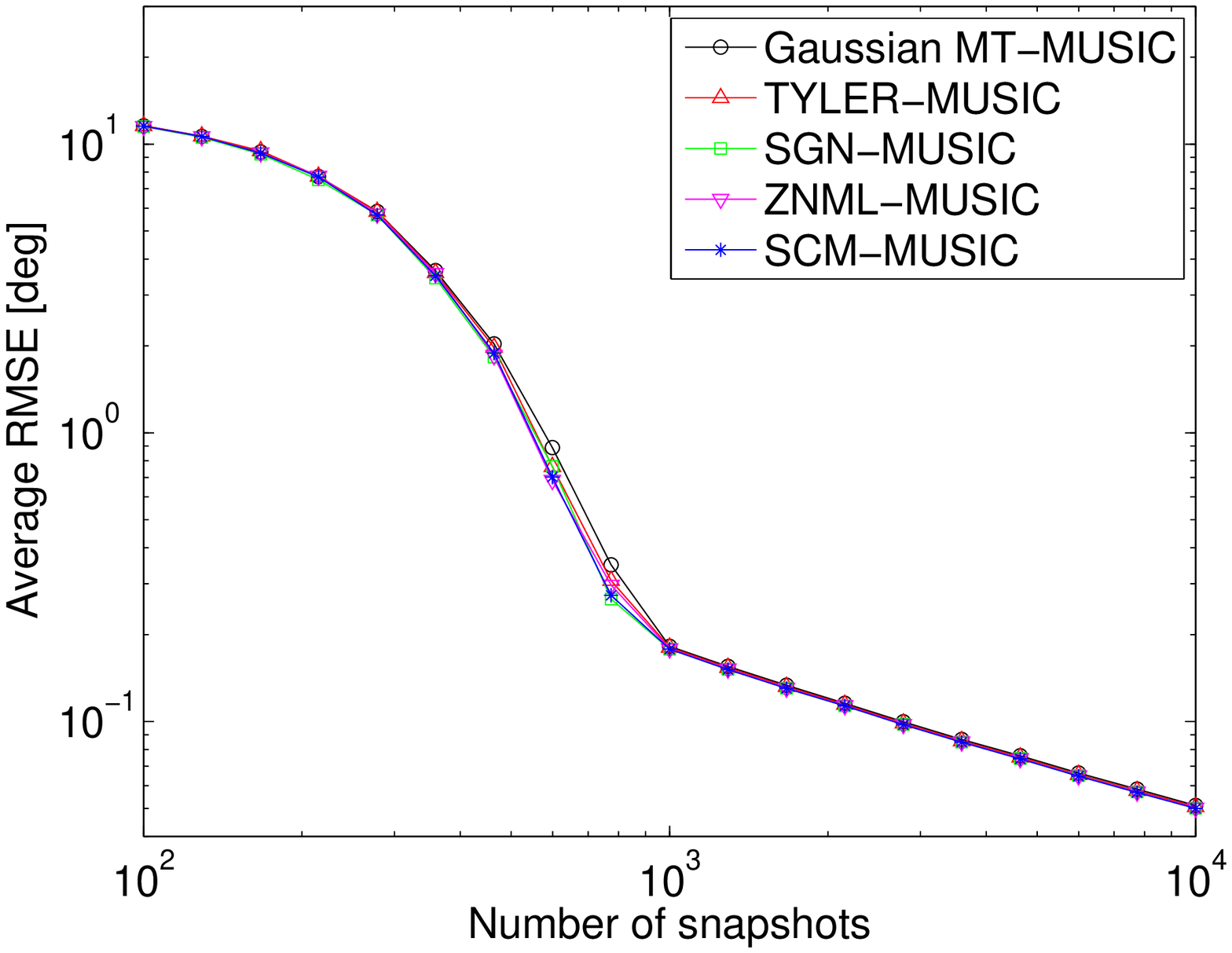}}}}
    {{\subfigure[]{\includegraphics[scale = 0.235]{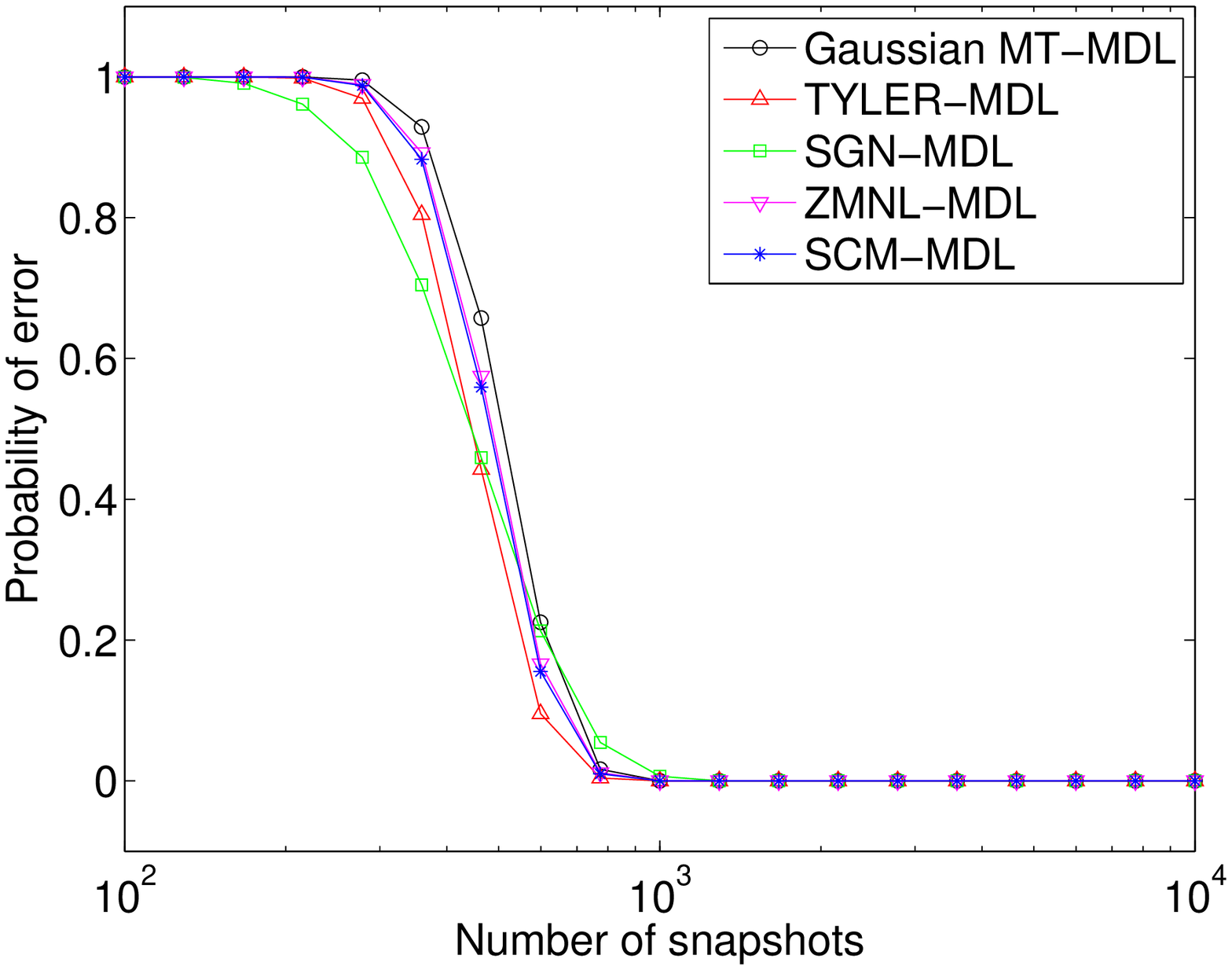}}}}
      \end{center}  
  \caption{\textbf{Non-coherent signals in Gaussian noise:}
   (a) Average RMSE versus GSNR. (b) Probability of error for estimating the number of signals versus GSNR. 
   (c) Average RMSE versus the number of snapshots. 
   (d) Probability of error for estimating the number of signals versus the number of snapshots.
   The performance measures versus GSNR were evaluated for $N=1000$ i.i.d. snapshots. The performance measures versus the number of snapshots were evaluated for ${\rm{GSNR}}=-11$ [dB].
   Notice that all algorithms perform similarly.}
\label{GAUSS_NON_COHERENT}
\end{figure}
\begin{figure}[htp]
  \begin{center}
    {{\subfigure[]{\includegraphics[scale = 0.235]{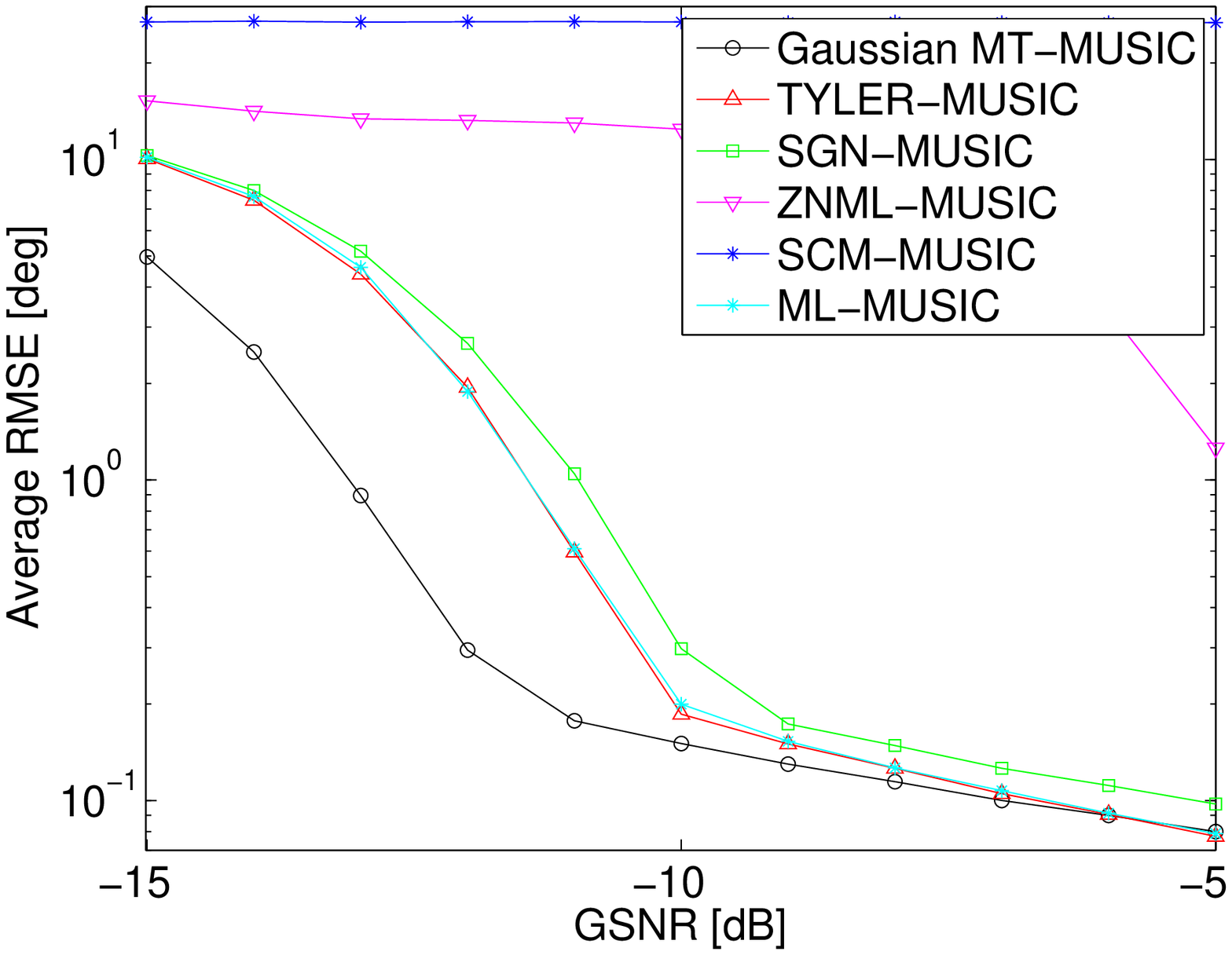}}}}
    {{\subfigure[]{\includegraphics[scale = 0.235]{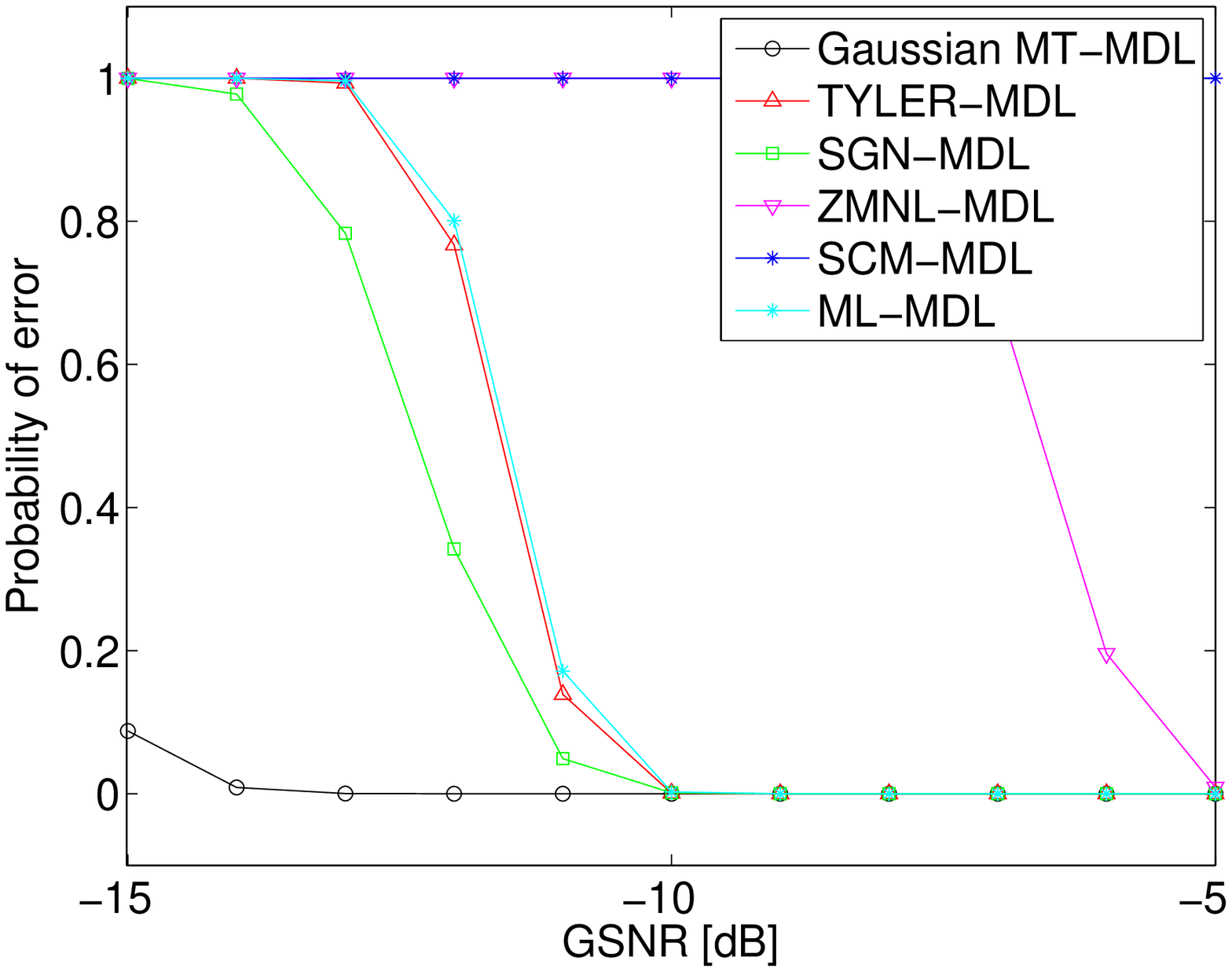}}}}
    {{\subfigure[]{\includegraphics[scale = 0.235]{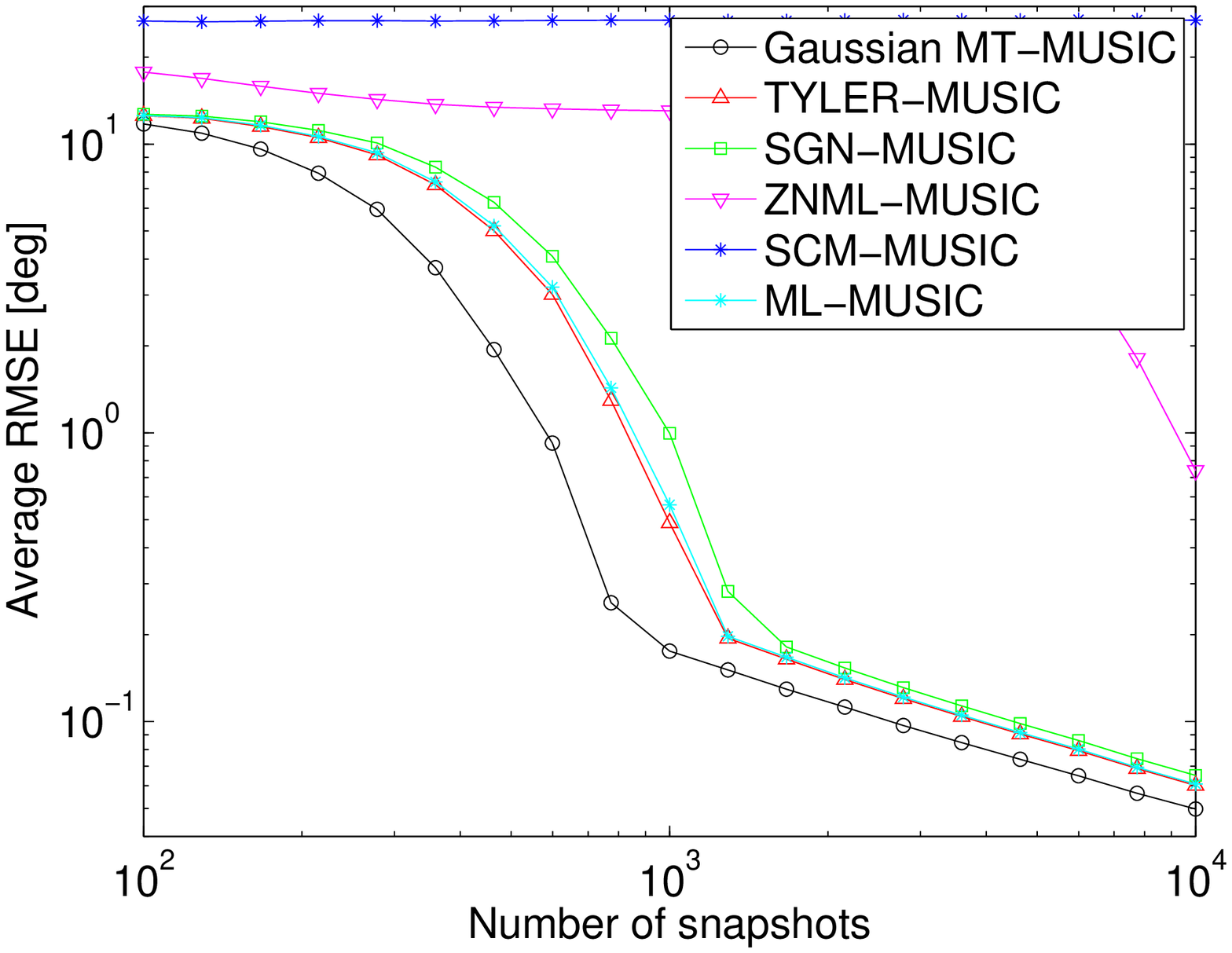}}}}
    {{\subfigure[]{\includegraphics[scale = 0.235]{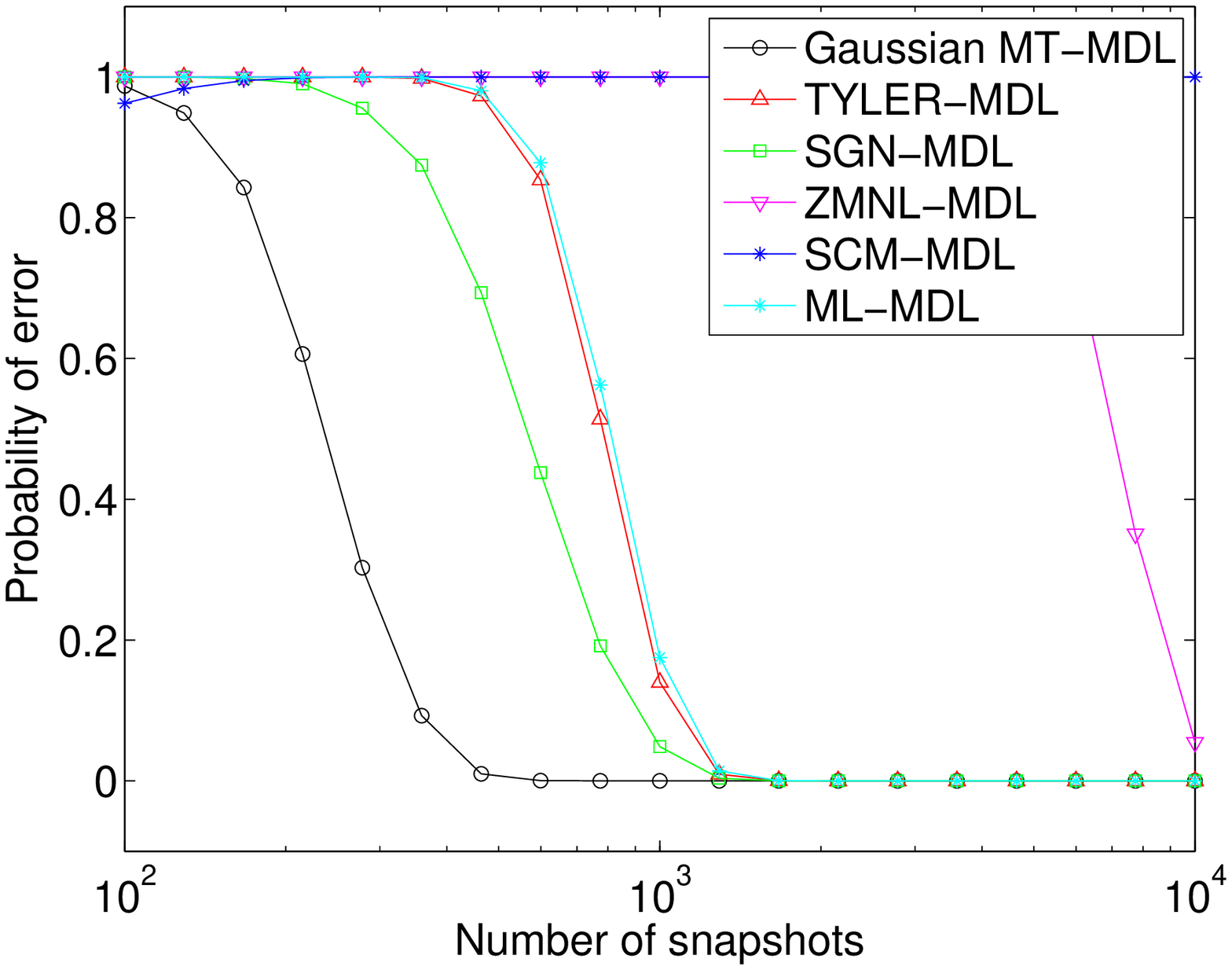}}}}
      \end{center}  
    \caption{\textbf{Non-coherent signals in Cauchy noise:}
   (a) Average RMSE versus GSNR. (b) Probability of error for estimating the number of signals versus GSNR. (c) Average RMSE versus the number of snapshots.  
   (d) Probability of error for estimating the number of signals versus the number of snapshots. The performance measures versus GSNR were evaluated for $N=1000$ i.i.d. snapshots. The performance measures versus the number of     snapshots were evaluated for ${\rm{GSNR}}=-11$ [dB].
   Notice that the Gaussian MT-MUSIC outperforms all other compared algorithms in the low GSNR and low sample size regimes.  Also notice that the Gaussian MT-MDL criterion leads to significantly lower error rates for estimating the number of signals.}
\label{CAUCHY_NON_COHERENT} 
\end{figure}
\begin{figure}[htp]
  \begin{center}
    {{\subfigure[]{\includegraphics[scale = 0.235]{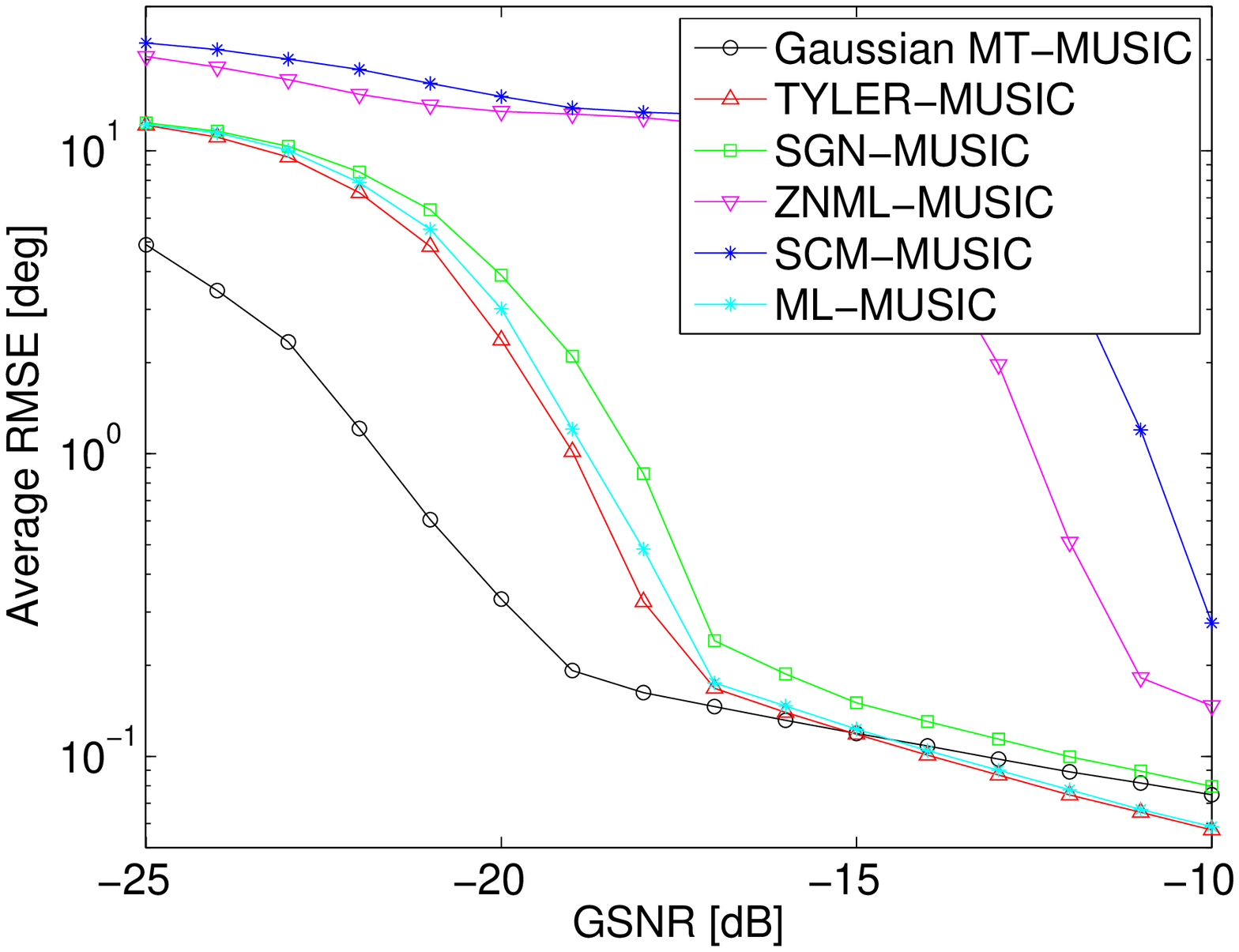}}}}
    {{\subfigure[]{\includegraphics[scale = 0.235]{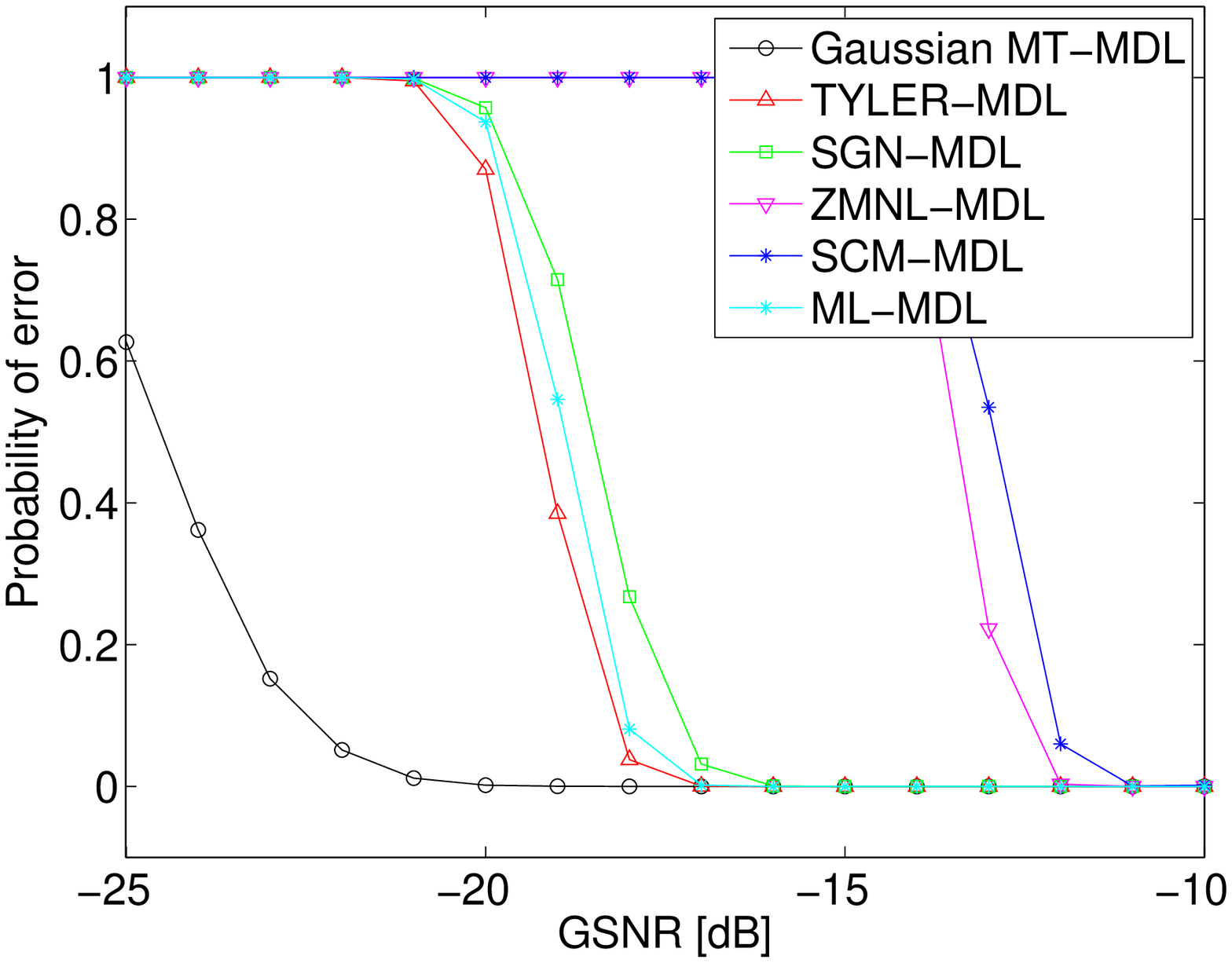}}}}
    {{\subfigure[]{\includegraphics[scale = 0.235]{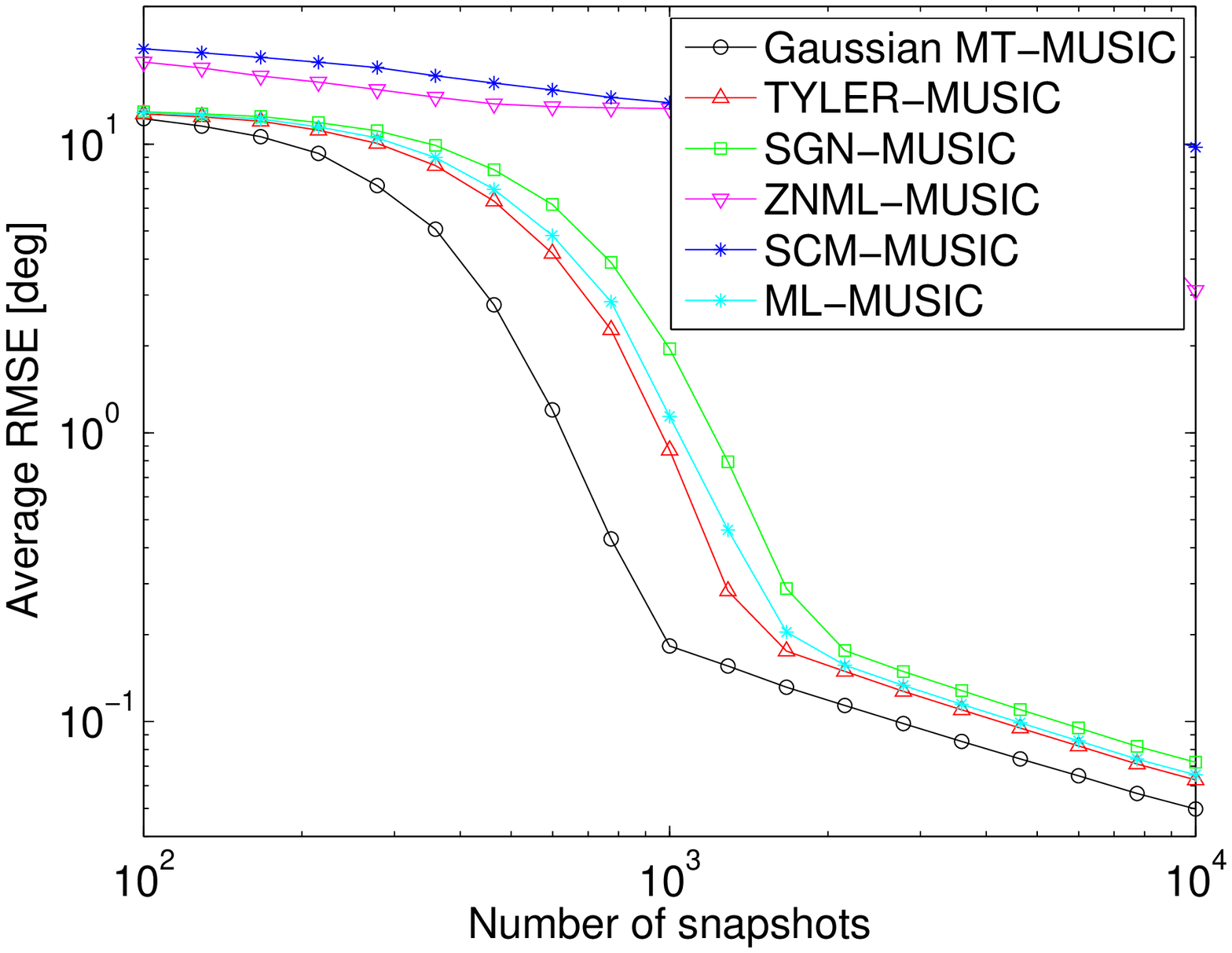}}}}
    {{\subfigure[]{\includegraphics[scale = 0.235]{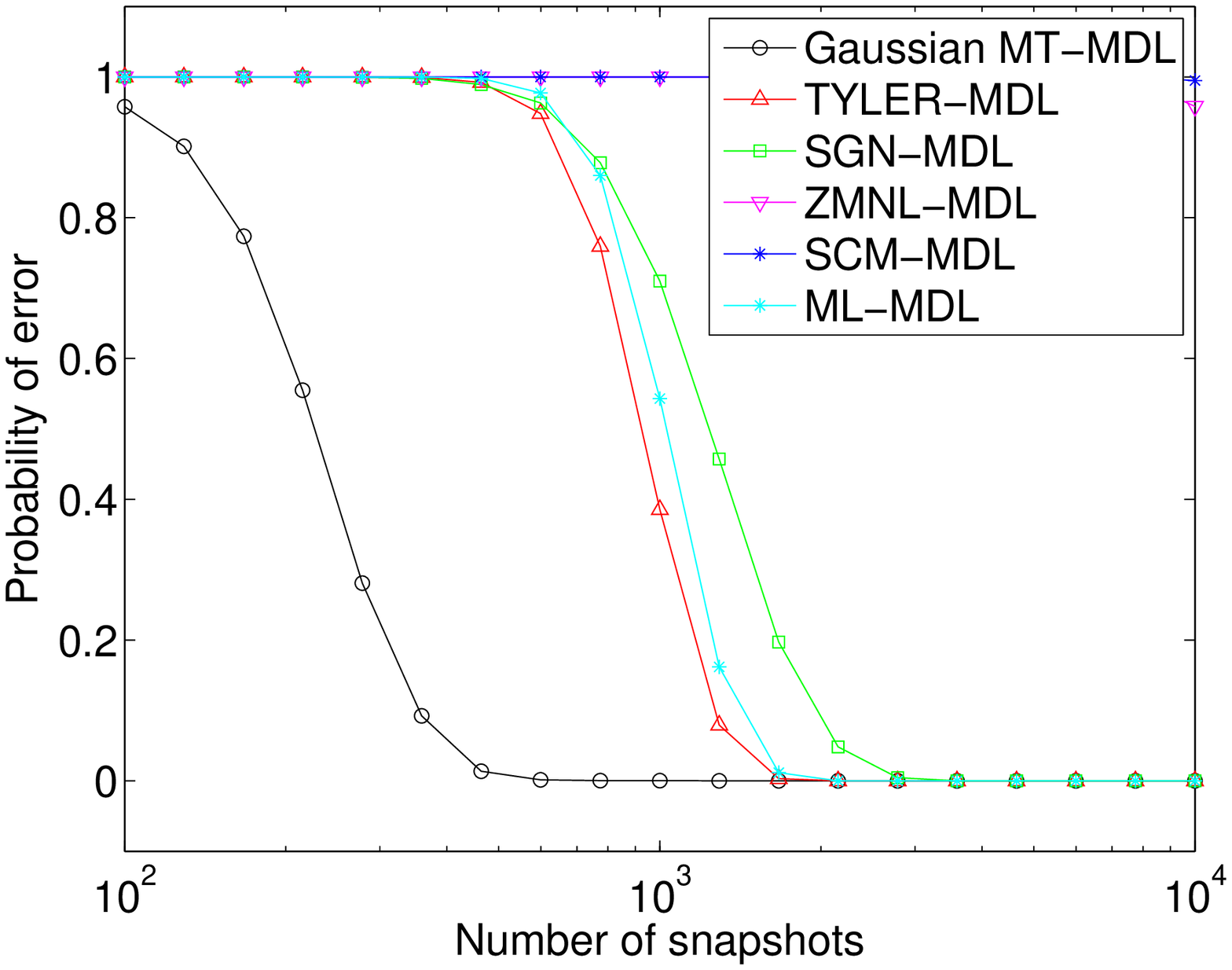}}}}
      \end{center}  
 \caption{{\textbf{Non-coherent signals in K-distributed noise with shape parameter $\nu=0.75$:}}
   (a) Average RMSE versus GSNR. (b) Probability of error for estimating the number of signals versus GSNR. 
   (c) Average RMSE versus the number of snapshots.  
   (d) Probability of error for estimating the number of signals versus the number of snapshots. The performance measures versus GSNR were evaluated for $N=1000$ i.i.d. snapshots. The performance measures versus the number of     snapshots were evaluated for ${\rm{GSNR}}=-19$ [dB]. Notice that the Gaussian MT-MUSIC outperforms all other compared algorithms in the low GSNR and low sample size regimes.  Also notice that the Gaussian MT-MDL criterion leads to significantly lower error rates for estimating the number of signals.}
\label{K_NON_COHERENT}
\end{figure}
\begin{figure}[htp]
  \begin{center}
    {{\subfigure[]{\includegraphics[scale = 0.235]{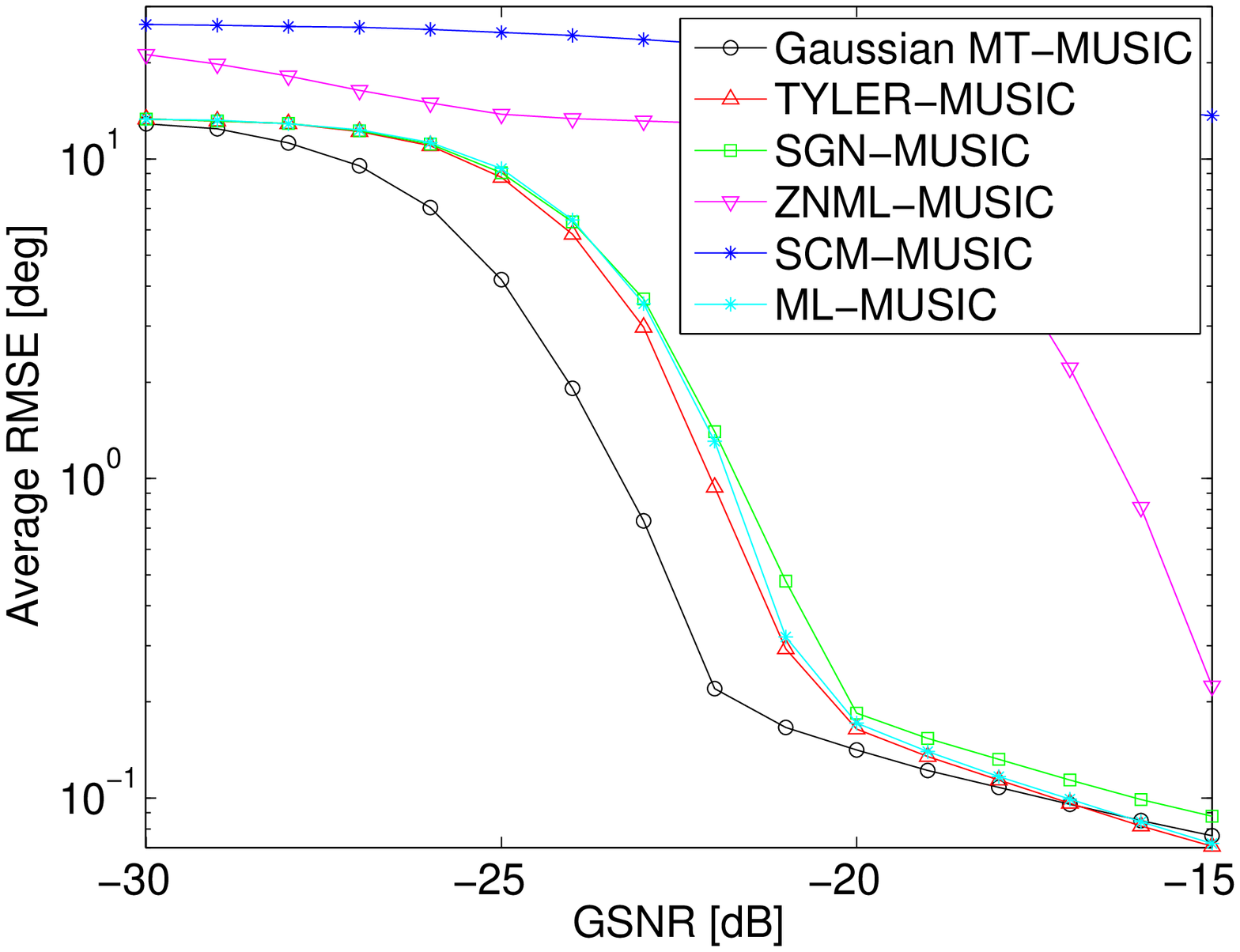}}}}
    {{\subfigure[]{\includegraphics[scale = 0.235]{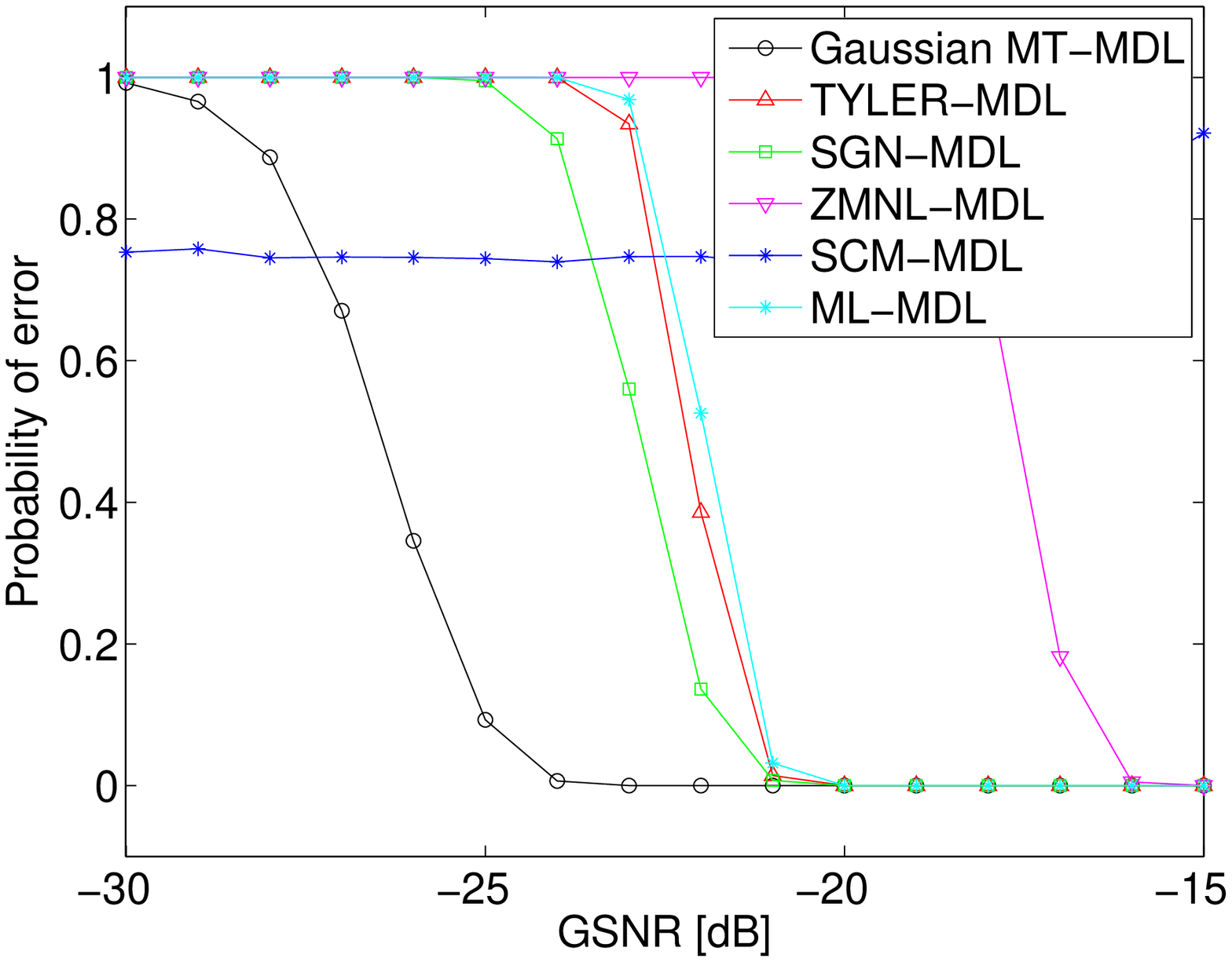}}}}
    {{\subfigure[]{\includegraphics[scale = 0.235]{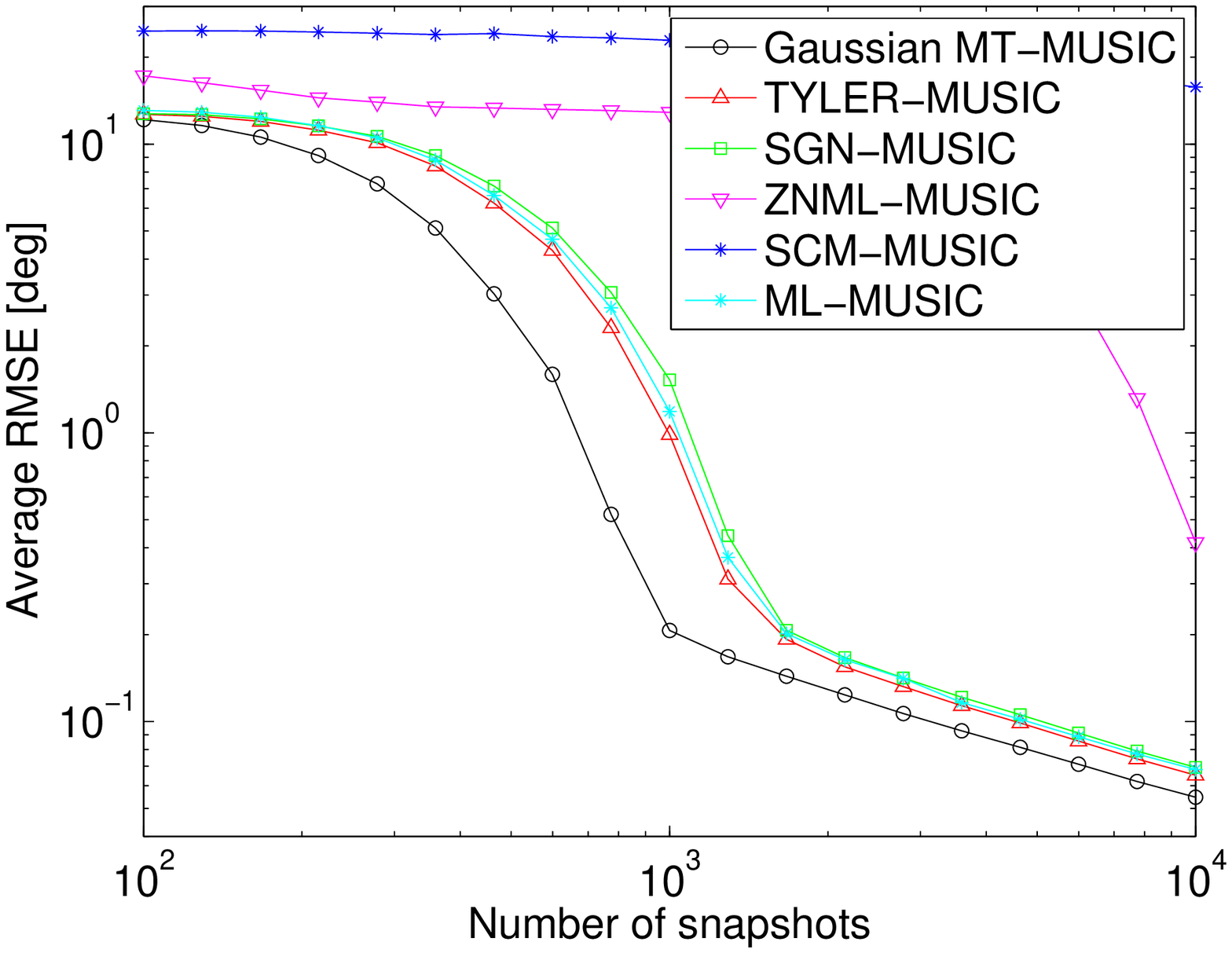}}}}
    {{\subfigure[]{\includegraphics[scale = 0.235]{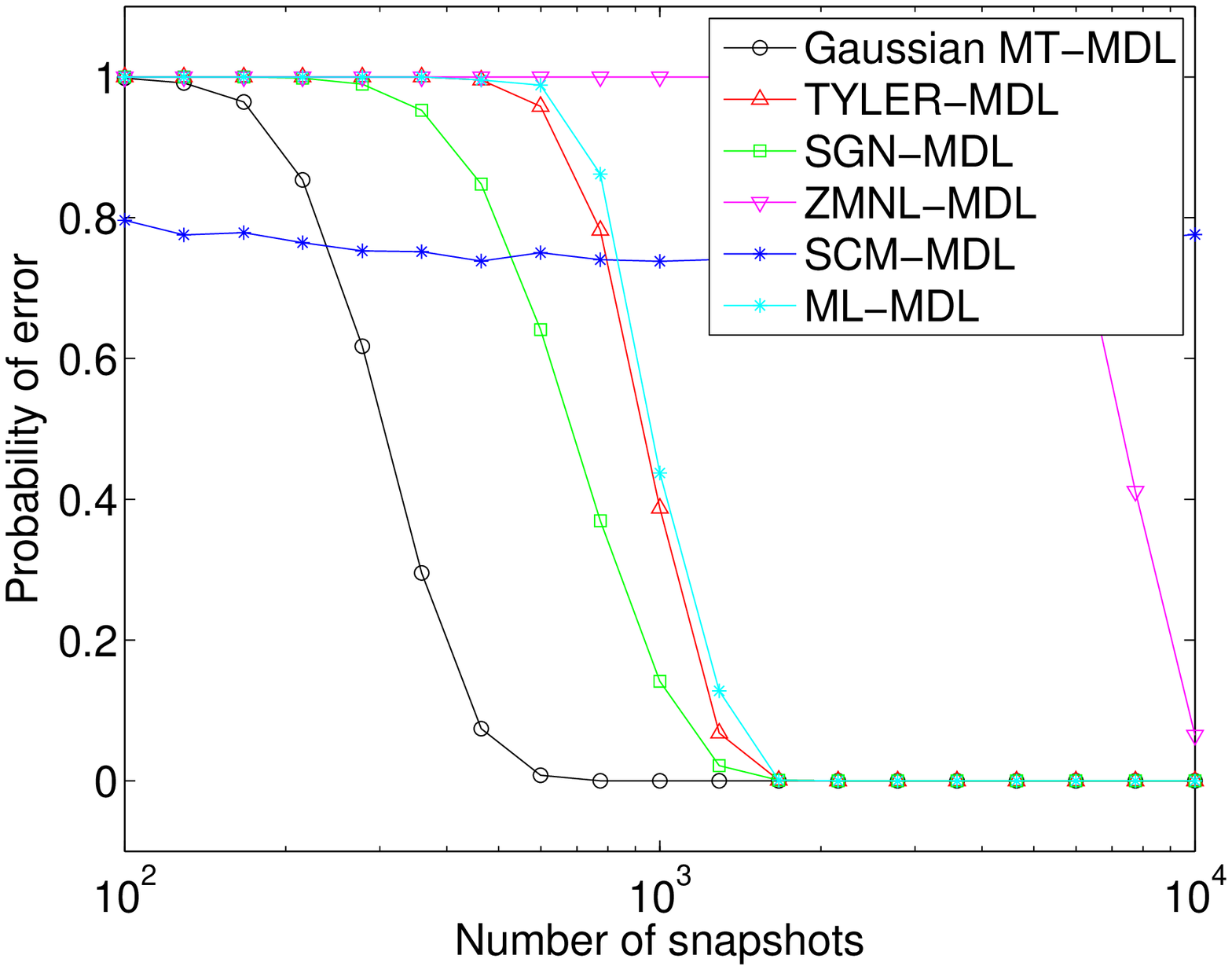}}}}
      \end{center}  
 \caption{\textbf{Non-coherent signals in spherical compound Gaussian noise with inverse-Gaussian texture and shape parameter $\lambda=0.1$:}
   (a) Average RMSE versus GSNR. (b) Probability of error for estimating the number of signals versus GSNR. (c) Average RMSE versus the number of snapshots.  
   (d) Probability of error for estimating the number of signals versus the number of snapshots. The performance measures versus GSNR were evaluated for $N=1000$ i.i.d. snapshots. The performance measures versus the number of     snapshots were evaluated for ${\rm{GSNR}}=-22$ [dB]. Notice that the Gaussian MT-MUSIC outperforms all other compared algorithms in the low GSNR and low sample size regimes.  Also notice that the Gaussian MT-MDL criterion leads to significantly lower error rates for estimating the number of signals.}
\label{IG_NON_COHERENT}  
\end{figure}
\subsection{Coherent signals}
In this example, we considered five coherent signals impinging on a $22$-element uniform linear array with $\lambda/2$ spacing from DOAs $\theta_{1}=-17^{\circ}$, $\theta_{2}=-3^{\circ}$, $\theta_{3}=2^{\circ}$, $\theta_{4}=13^{\circ}$ and $\theta_{5}=20^{\circ}$. The signals were generated according to the model (\ref{CoherentModel}), where $s\left(n\right)$ is a 4-QAM signal with power $\sigma^{2}_{S}$. The attenuation coefficients were set to $\eta_{1}=0.8\exp(i\pi/3)$, $\eta_{2}=1$, $\eta_{3}=0.9\exp(i\pi/4)$, $\eta_{4}=0.7\exp(i\pi/5)$ and $\eta_{5}=0.6\exp(i\pi/6)$. The dimension of the spatially smoothed covariance was set to $r=16$. The average RMSEs for estimating the DOAs and the error rates for estimating the number of signals versus GSNR are depicted in Figs. \ref{GAUSS_COHERENT}-\ref{IG_COHERENT} for each noise distribution. 
Notice that for the Gaussian noise case, there is no significant performance gap between the compared MUSIC algorithms. Regarding the estimation of the number of signals, the sign-covariance based modified MDL criterion results in better estimation performance. This may be attributed to the fact that in the sign-covariance based modified MDL criterion \cite{Visuri} the eigenvalues are estimated in a more stable manner. For the other noise distributions, the spatially smoothed Gaussian MT-MUSIC and the Gaussian modified MT-MDL based estimation of the number of signals outperform all other robust MUSIC and MDL generalizations in the low GSNR and low sample size regimes, with significantly lower breakdown thresholds. Again, as in the non-coherent case, this performance advantage may be attributed to the following facts. First, unlike the empirical sign-covariance, Tyler's scatter M-estimator and the ML-estimators of scatter corresponding to each non-Gaussian noise distribution, the influence function of the empirical Gaussian MT-covariance is very small for large norm outliers. Such outliers are likely in low GSNRs and become more frequent when the sample size is small. Second, unlike the ZMNL preprocessing based technique, the proposed measure-transformation approach preserves the noise subspace and effectively suppresses outliers without significant performance loss in estimating the DOAs and the number of signals.
\begin{figure}[htp]
  \begin{center}
    {{\subfigure[]{\includegraphics[scale = 0.235]{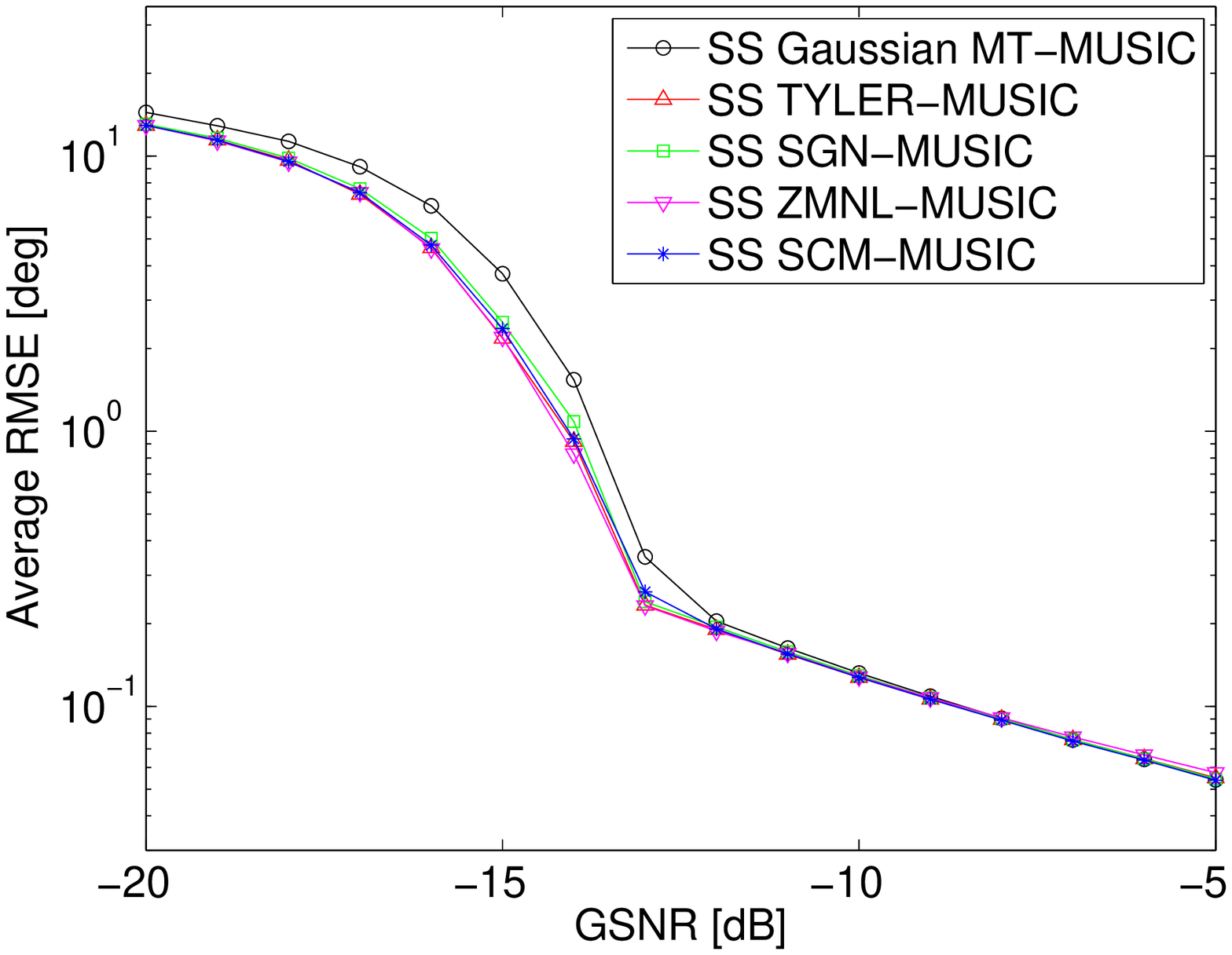}}}}
    {{\subfigure[]{\includegraphics[scale = 0.235]{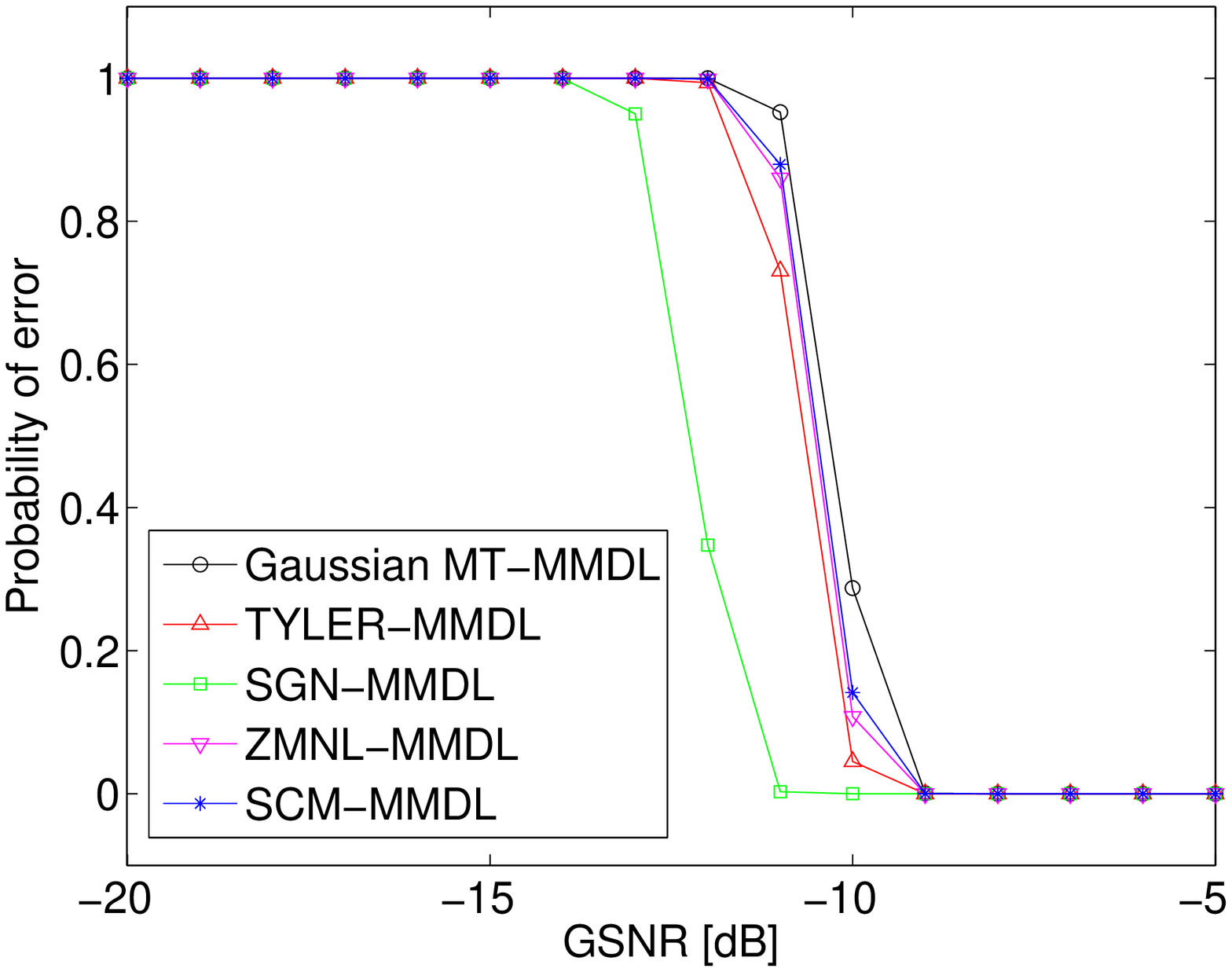}}}}
    {{\subfigure[]{\includegraphics[scale = 0.235]{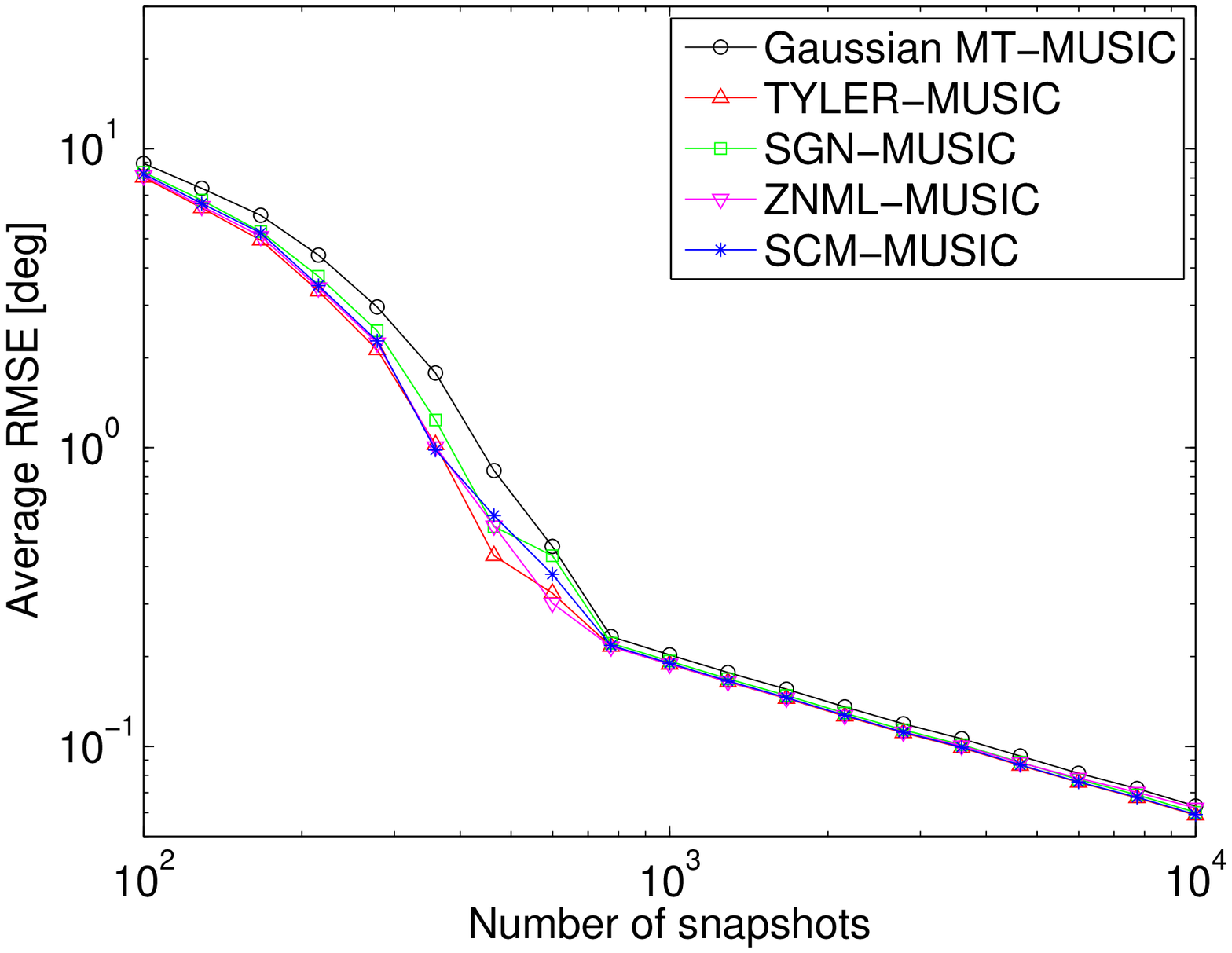}}}}
    {{\subfigure[]{\includegraphics[scale = 0.235]{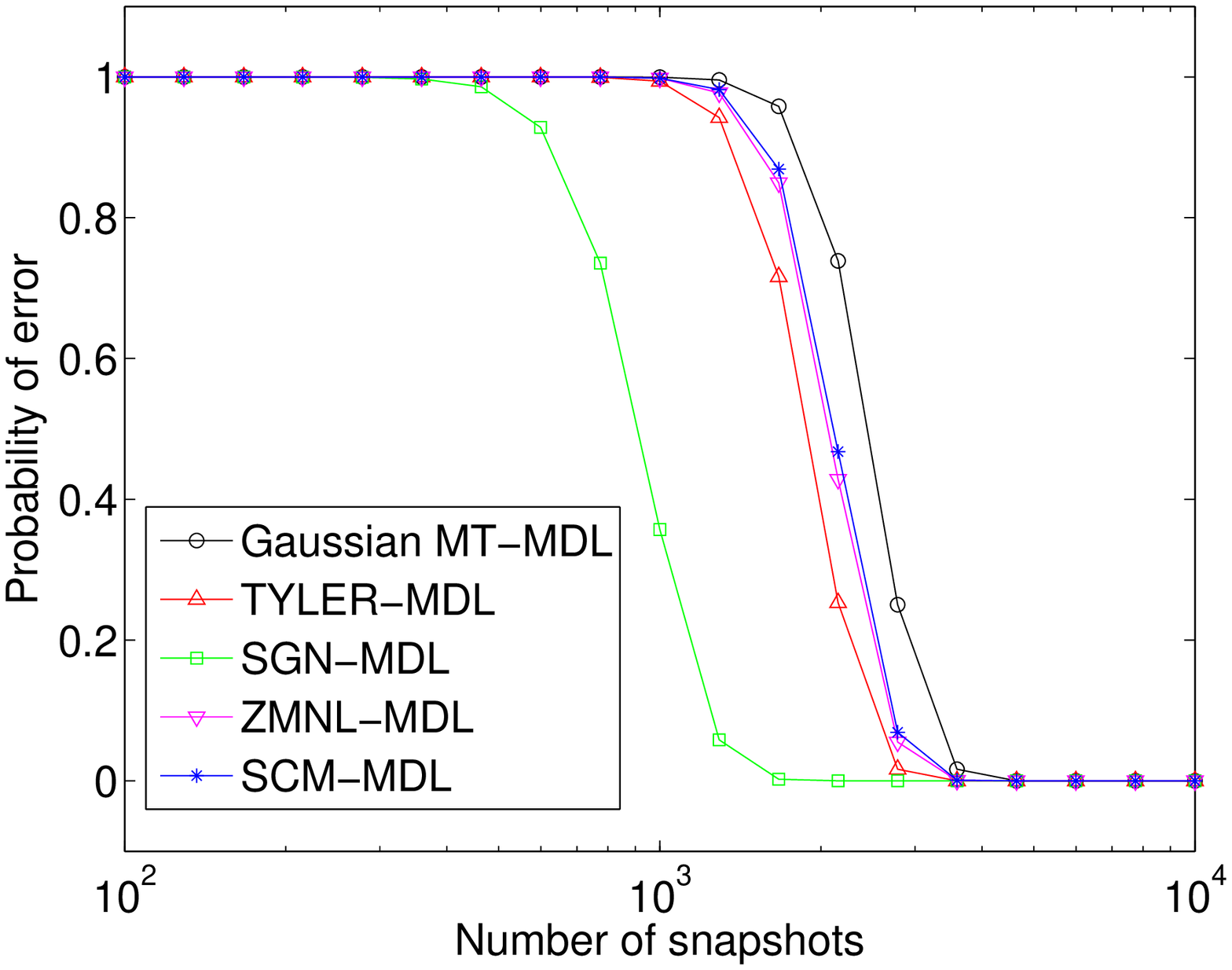}}}}
      \end{center}  
  \caption{\textbf{Coherent signals in Gaussian noise:}
   (a) Average RMSE versus GSNR. (b) Probability of error for estimating the number of signals versus GSNR. 
   (c) Average RMSE versus the number of snapshots.  
   (d) Probability of error for estimating the number of signals versus the number of snapshots. 
   The performance measures versus GSNR were evaluated for $N=1000$ i.i.d. snapshots. The performance measures versus the number of snapshots were evaluated for ${\rm{GSNR}}=-12$ [dB].
    Notice that there is no significant performance gap between the compared MUSIC algorithms. The sign-covariance based modified MDL criterion results in better estimation of the number of signals. This may be attributed to the fact that the sign-covariance based modified MDL criterion \cite{Visuri} involves more stable eigenvalues estimation.}
\label{GAUSS_COHERENT}
\end{figure}
\begin{figure}[htp]
  \begin{center}
    {{\subfigure[]{\includegraphics[scale = 0.235]{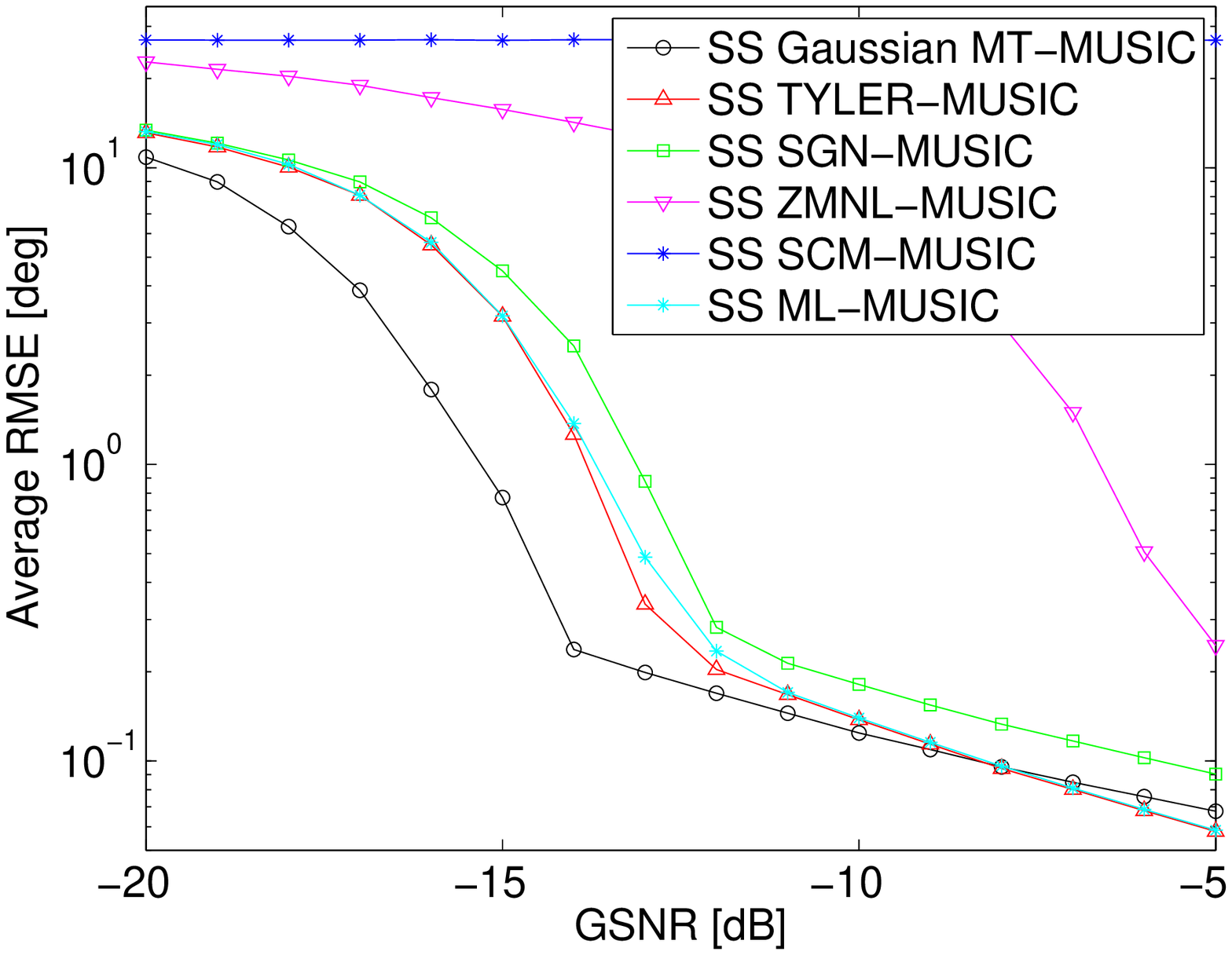}}}}
    {{\subfigure[]{\includegraphics[scale = 0.235]{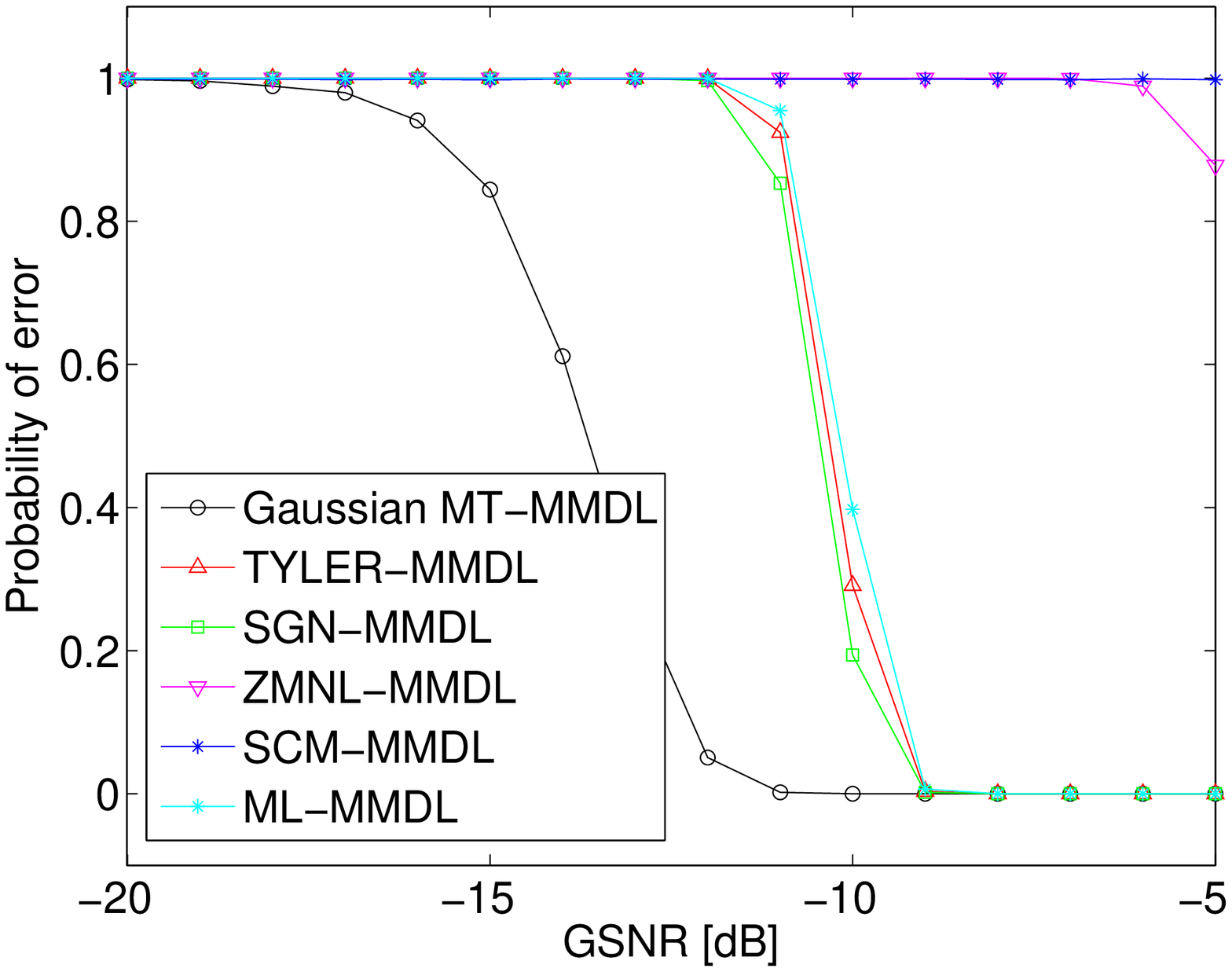}}}}
    {{\subfigure[]{\includegraphics[scale = 0.235]{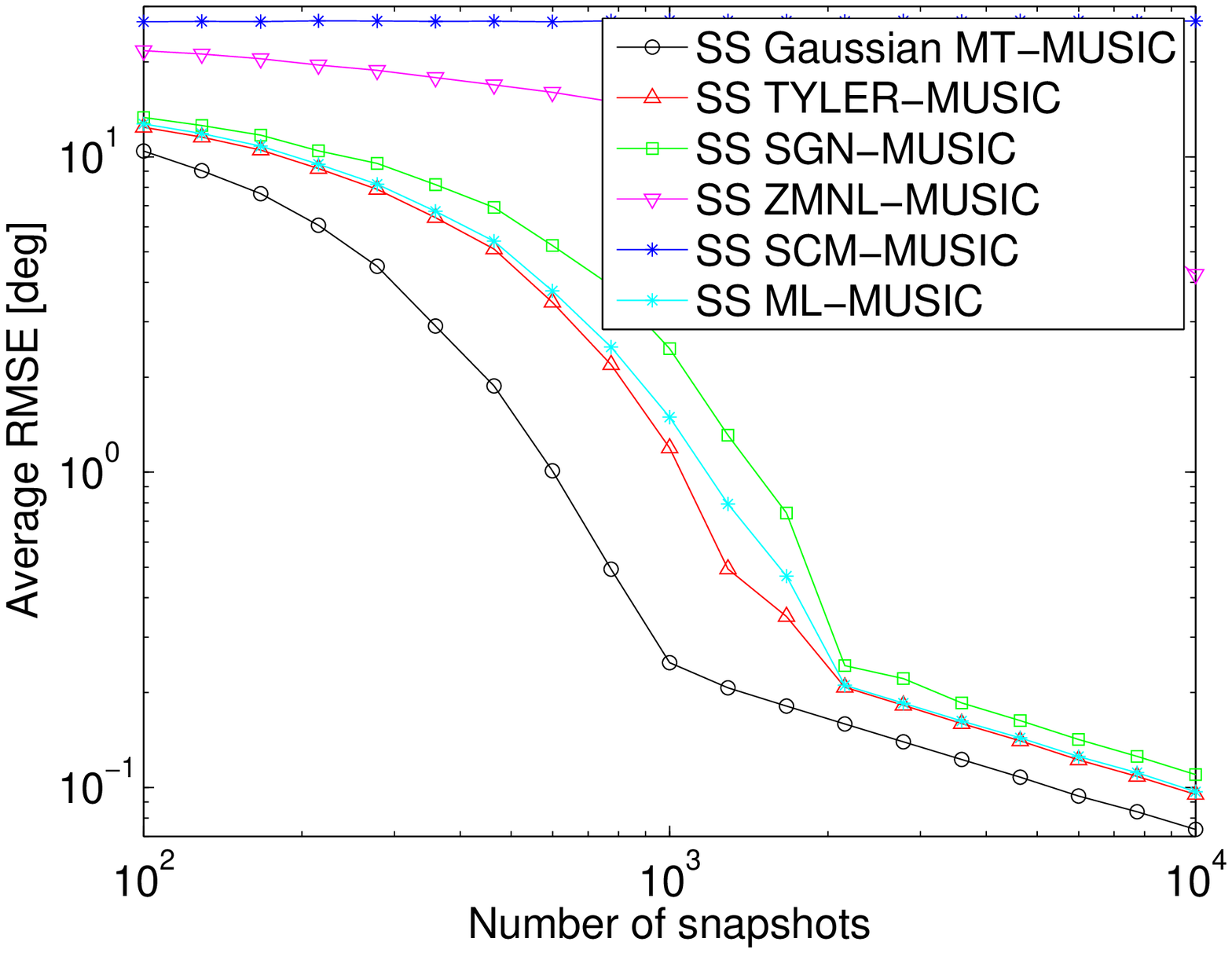}}}}
    {{\subfigure[]{\includegraphics[scale = 0.235]{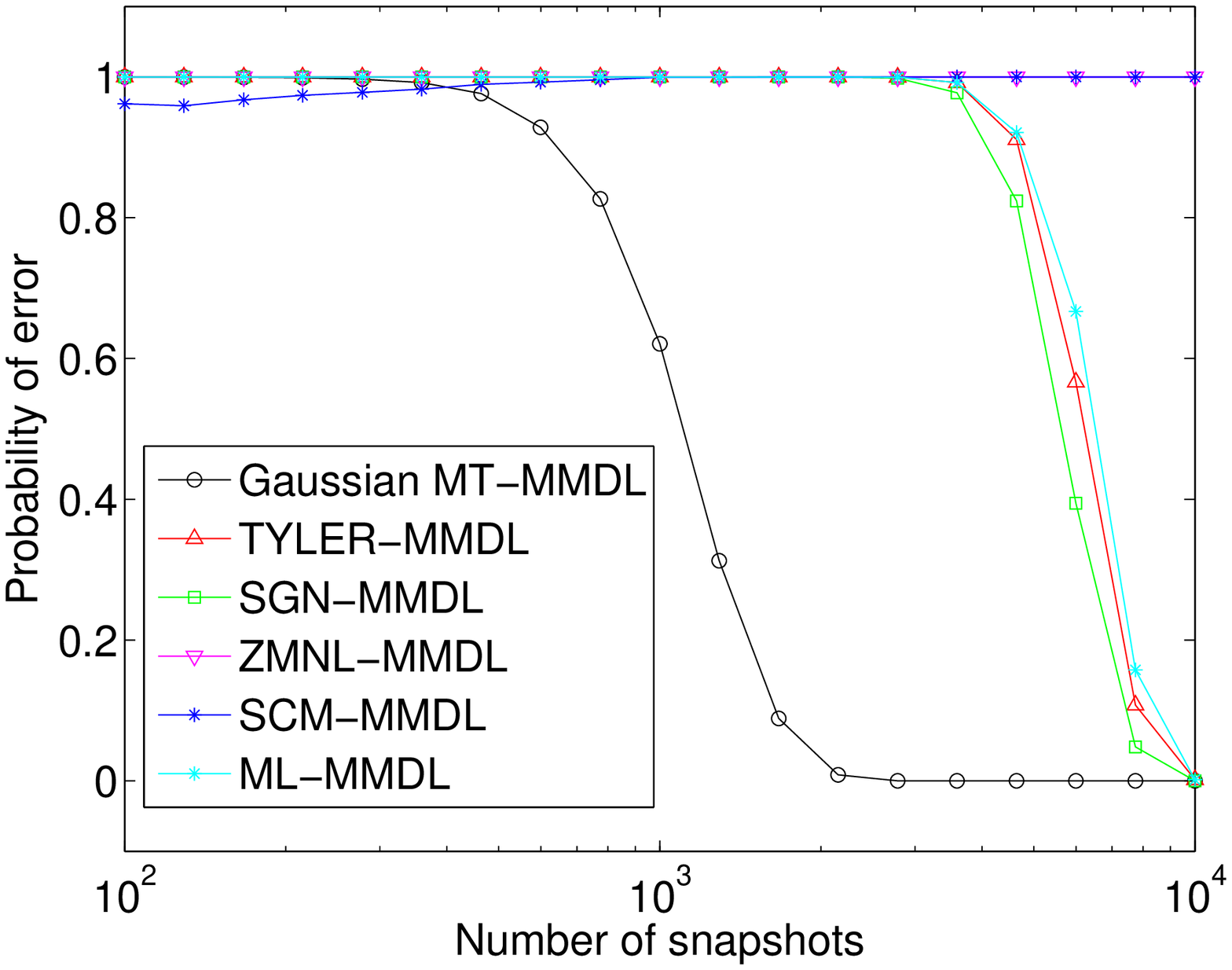}}}}
      \end{center}  
  \caption{\textbf{Coherent signals in Cauchy noise:}
   (a) Average RMSE versus GSNR. (b) Probability of error for estimating the number of signals versus GSNR. 
   (c) Average RMSE versus the number of snapshots.  
   (d) Probability of error for estimating the number of signals versus the number of snapshots. The performance measures versus GSNR were evaluated for $N=1000$ i.i.d. snapshots. The performance measures versus the number of     snapshots were evaluated for ${\rm{GSNR}}=-14$ [dB]. Note that similarly to the non-coherent case, the Gaussian MT-MUSIC estimator has significantly lower GSNR and sample size threshold regions than the other methods.  
  Also notice that the Gaussian MT-MMDL estimator of the number of signals outperforms all other MDL based estimators with significantly lower probability of error at the low GSNR and low sample size regimes.}
\label{CAUCHY_COHERENT}
\end{figure}
\begin{figure}[htp]
  \begin{center}
    {{\subfigure[]{\includegraphics[scale = 0.235]{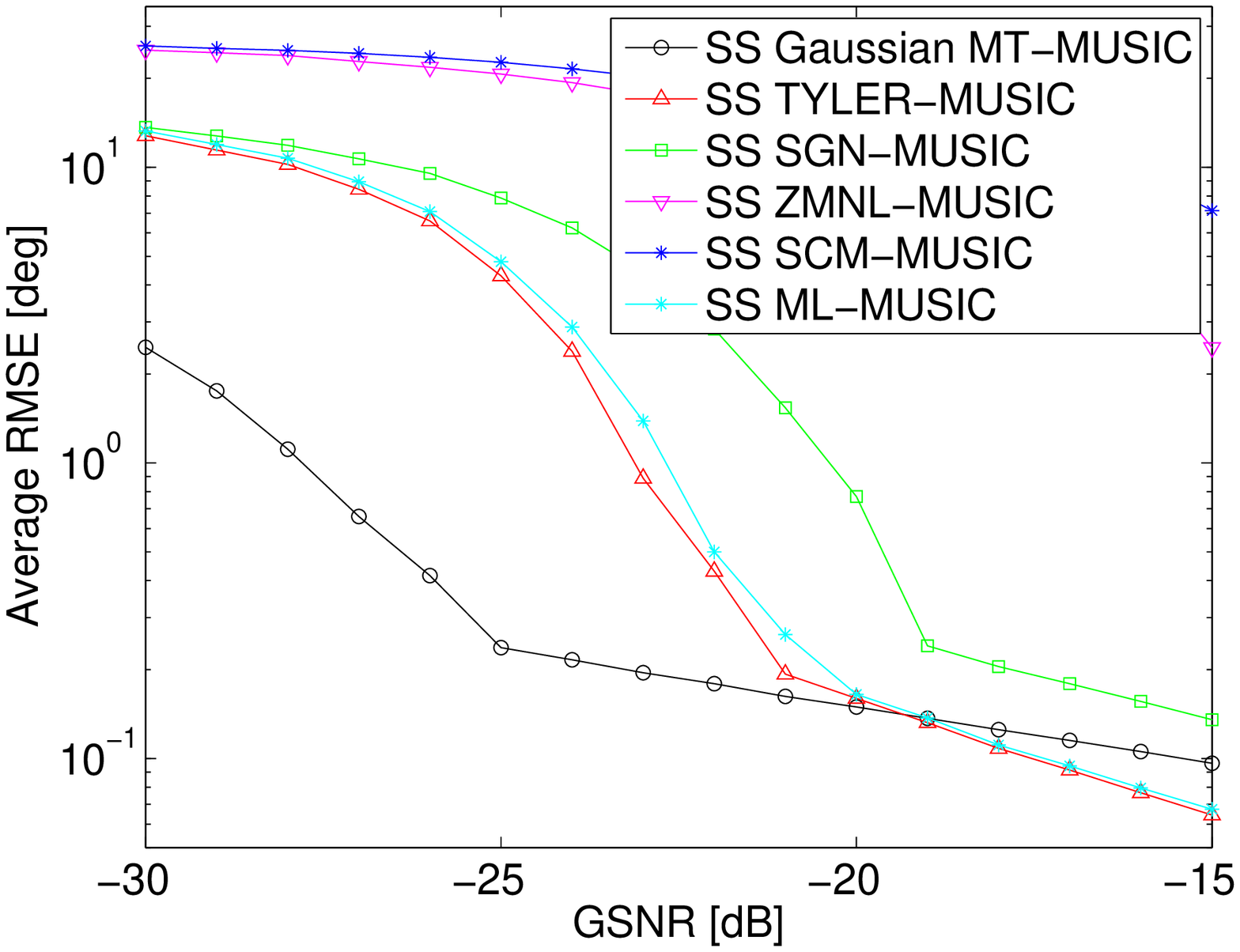}}}}
    {{\subfigure[]{\includegraphics[scale = 0.235]{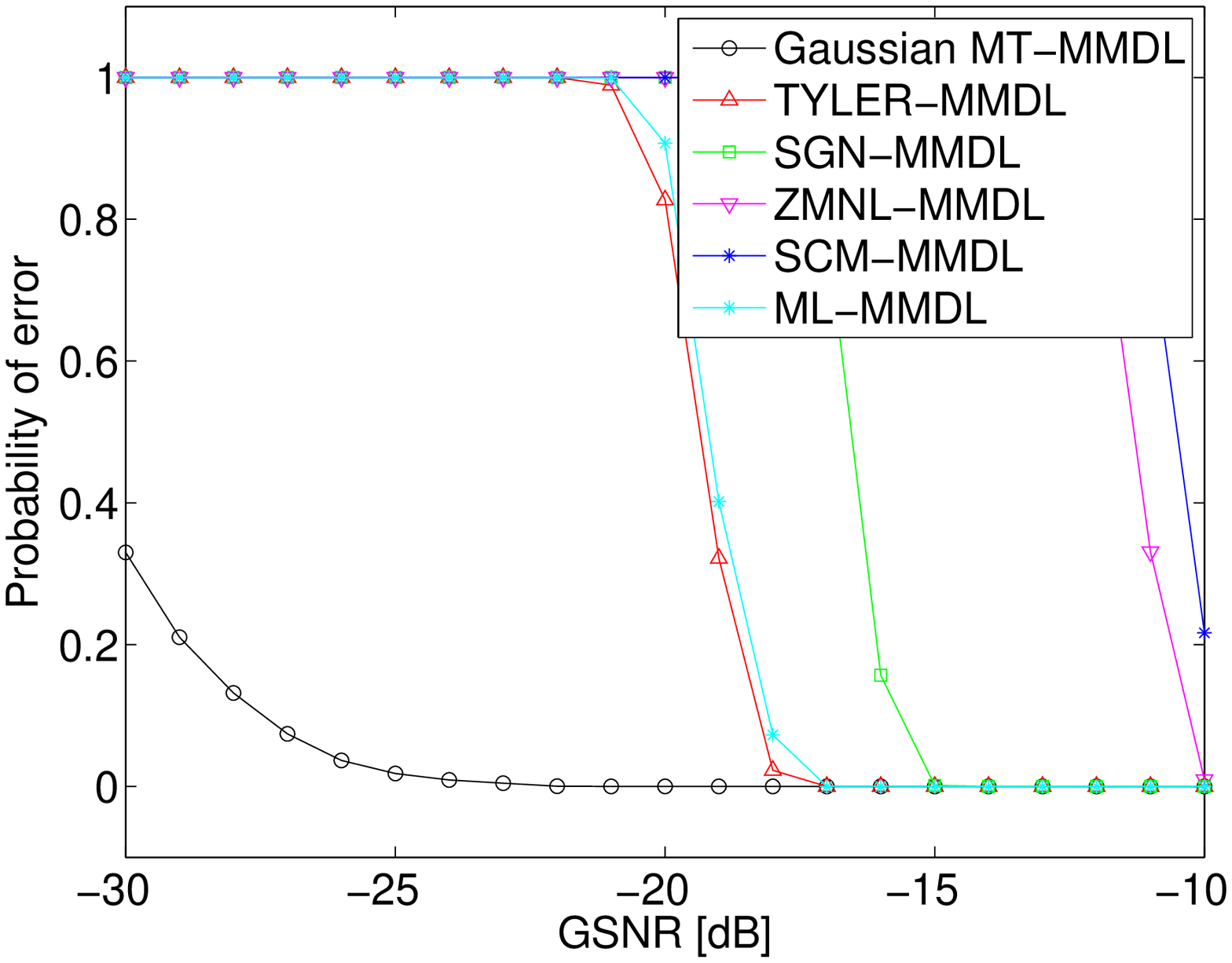}}}}
    {{\subfigure[]{\includegraphics[scale = 0.235]{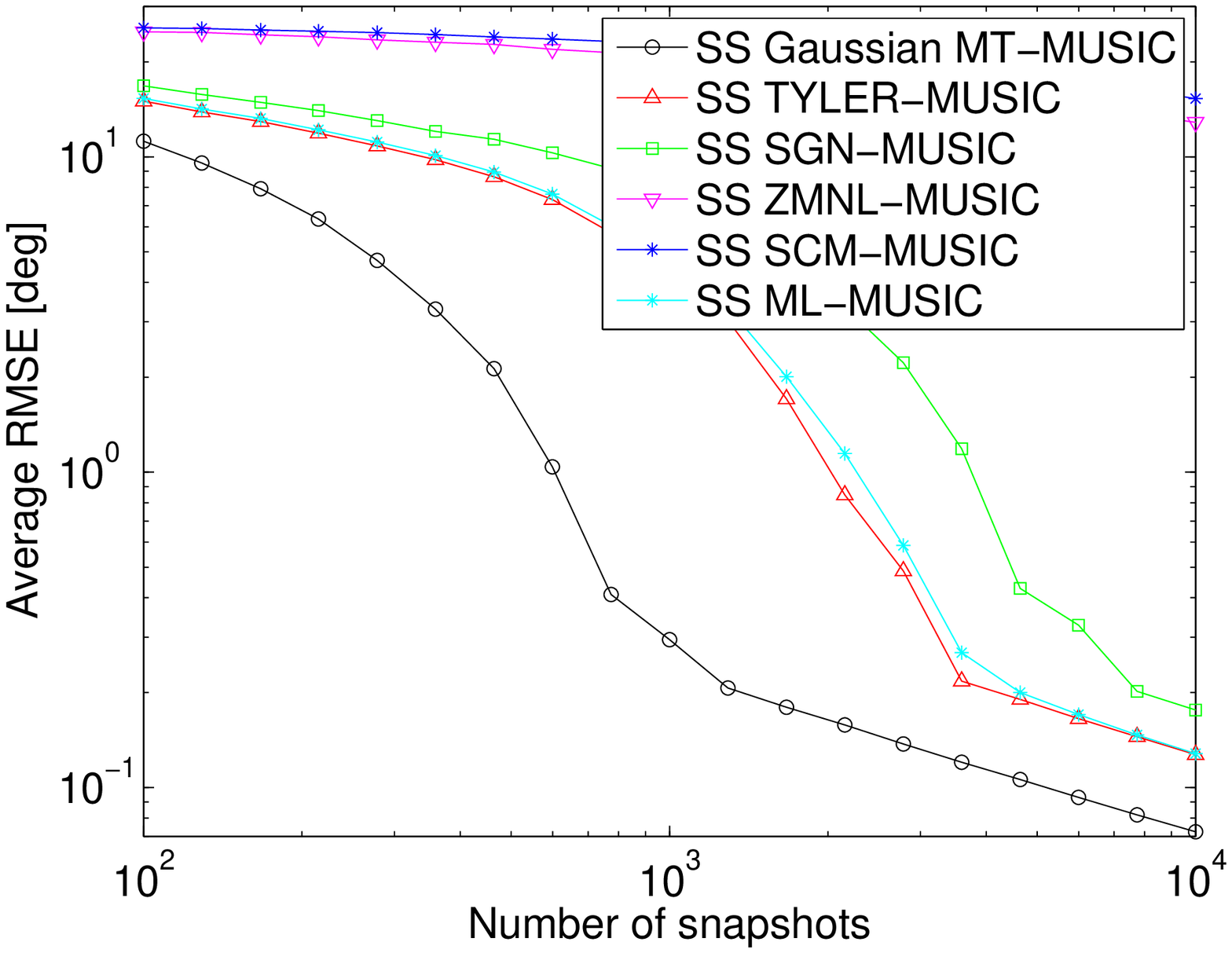}}}}
    {{\subfigure[]{\includegraphics[scale = 0.235]{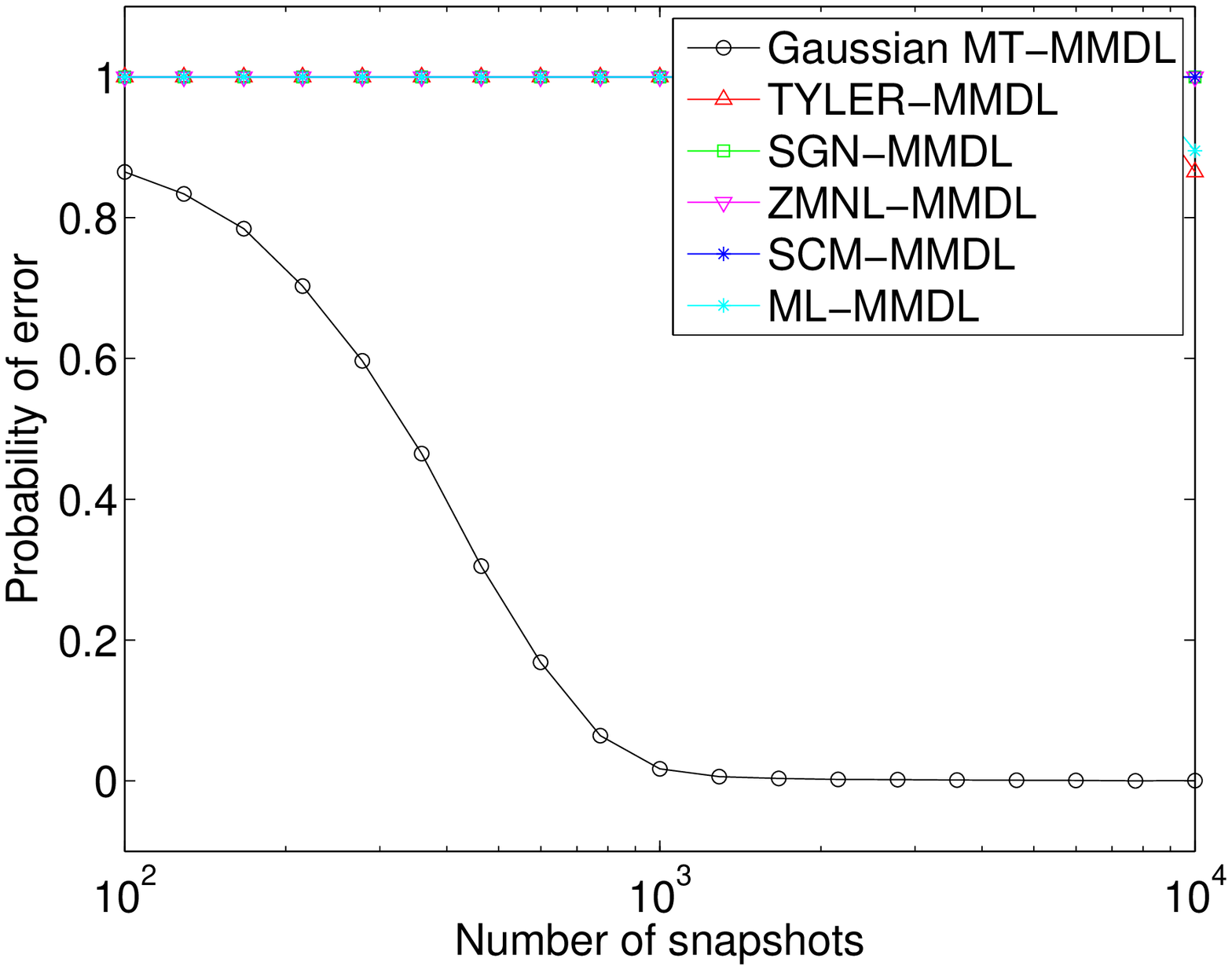}}}}
      \end{center}  
  \caption{\textbf{Coherent signals in K-distributed noise with shape parameter $\nu=0.75$:}
   (a) Average RMSE versus GSNR. (b) Probability of error for estimating the number of signals versus GSNR. 
   (c) Average RMSE versus the number of snapshots.  
   (d) Probability of error for estimating the number of signals versus the number of snapshots. The performance measures versus GSNR were evaluated for $N=1000$ i.i.d. snapshots. The performance measures versus the number of     snapshots were evaluated for ${\rm{GSNR}}=-25$ [dB]. Note that similarly to the non-coherent case, the Gaussian MT-MUSIC estimator has significantly lower GSNR and sample size threshold regions than the other methods.  
  Also notice that the Gaussian MT-MMDL estimator of the number of signals outperforms all other MDL based estimators.}
\label{K_COHERENT}
\end{figure}
\begin{figure}[htp]
  \begin{center}
    {{\subfigure[]{\includegraphics[scale = 0.235]{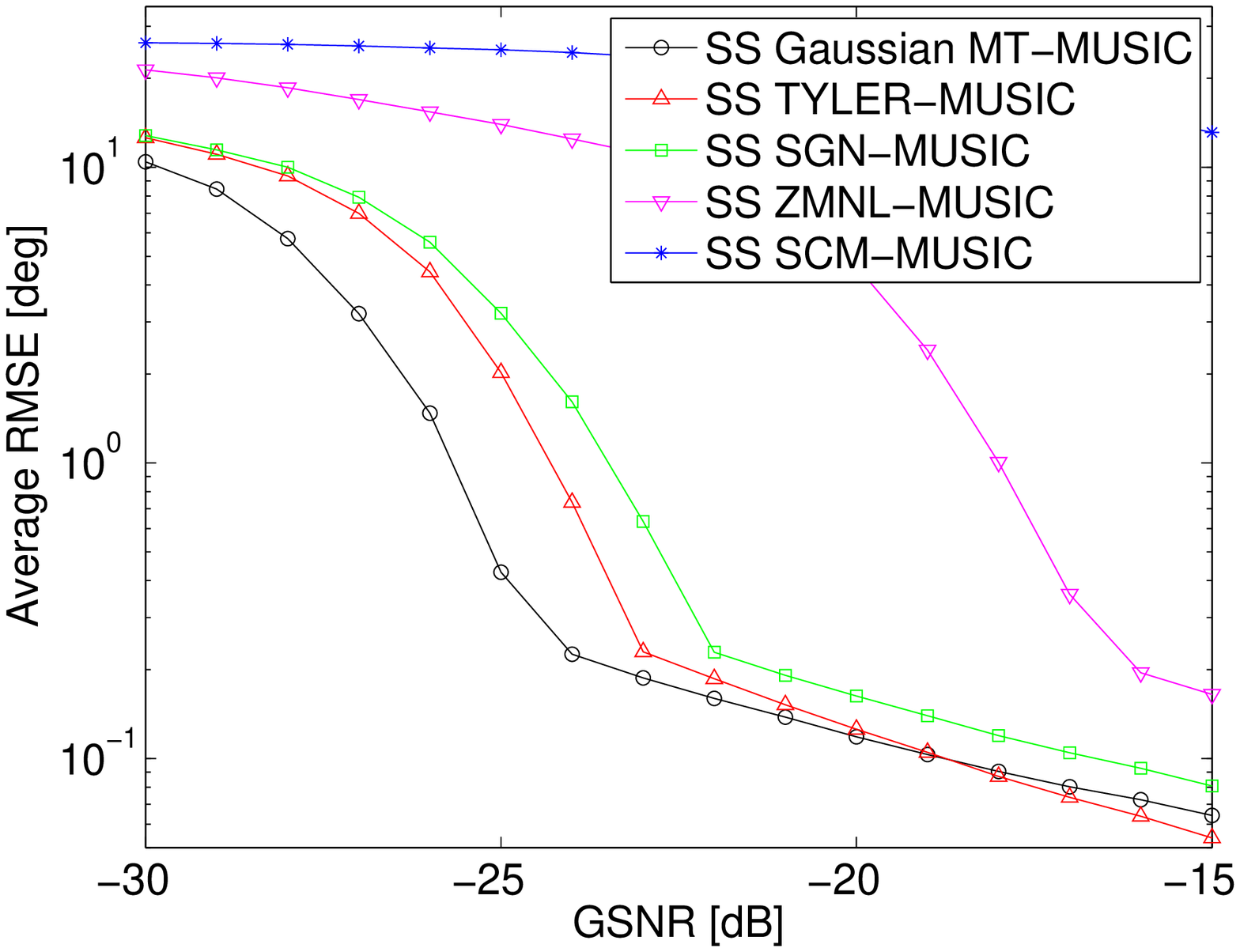}}}}
    {{\subfigure[]{\includegraphics[scale = 0.235]{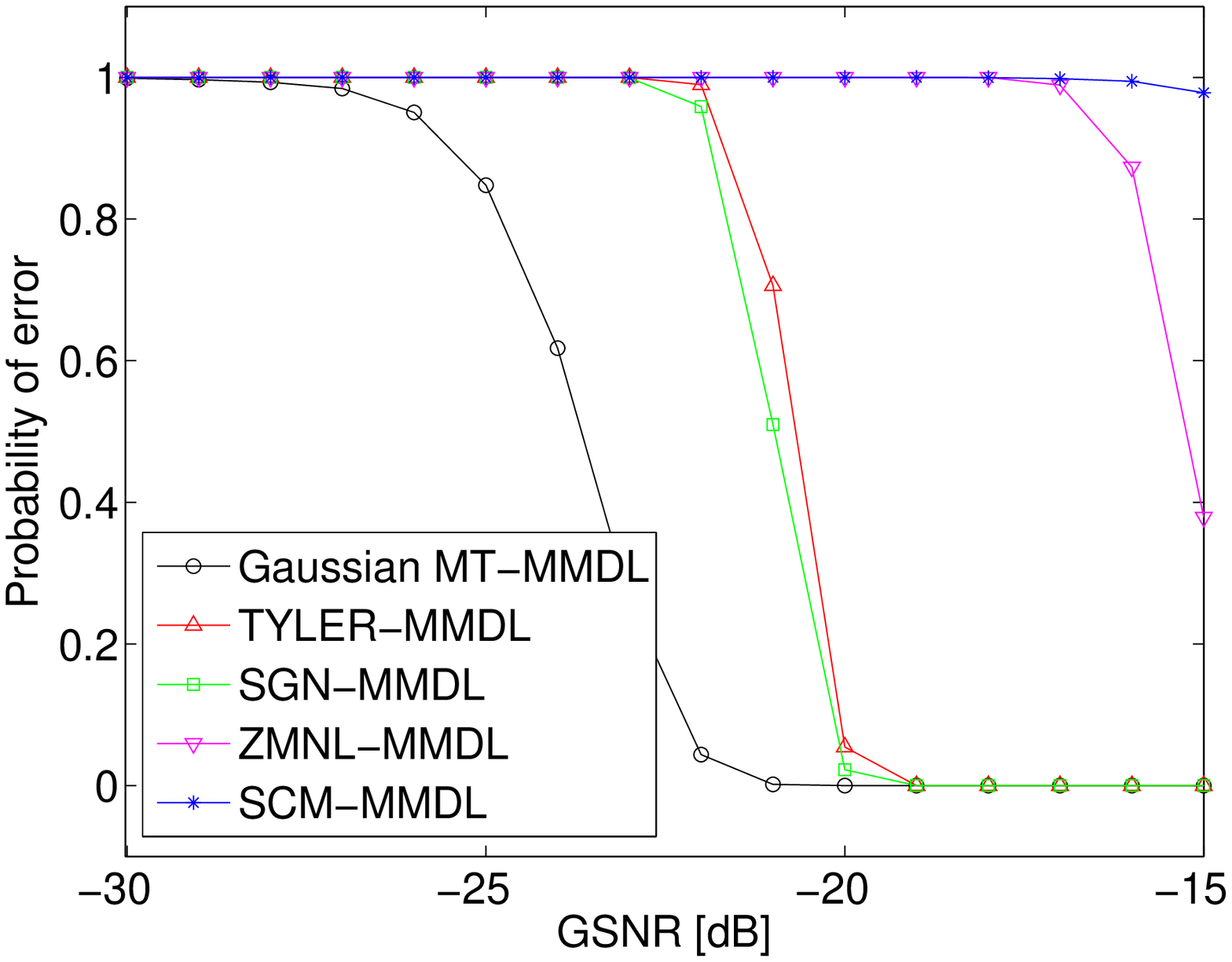}}}}
    {{\subfigure[]{\includegraphics[scale = 0.235]{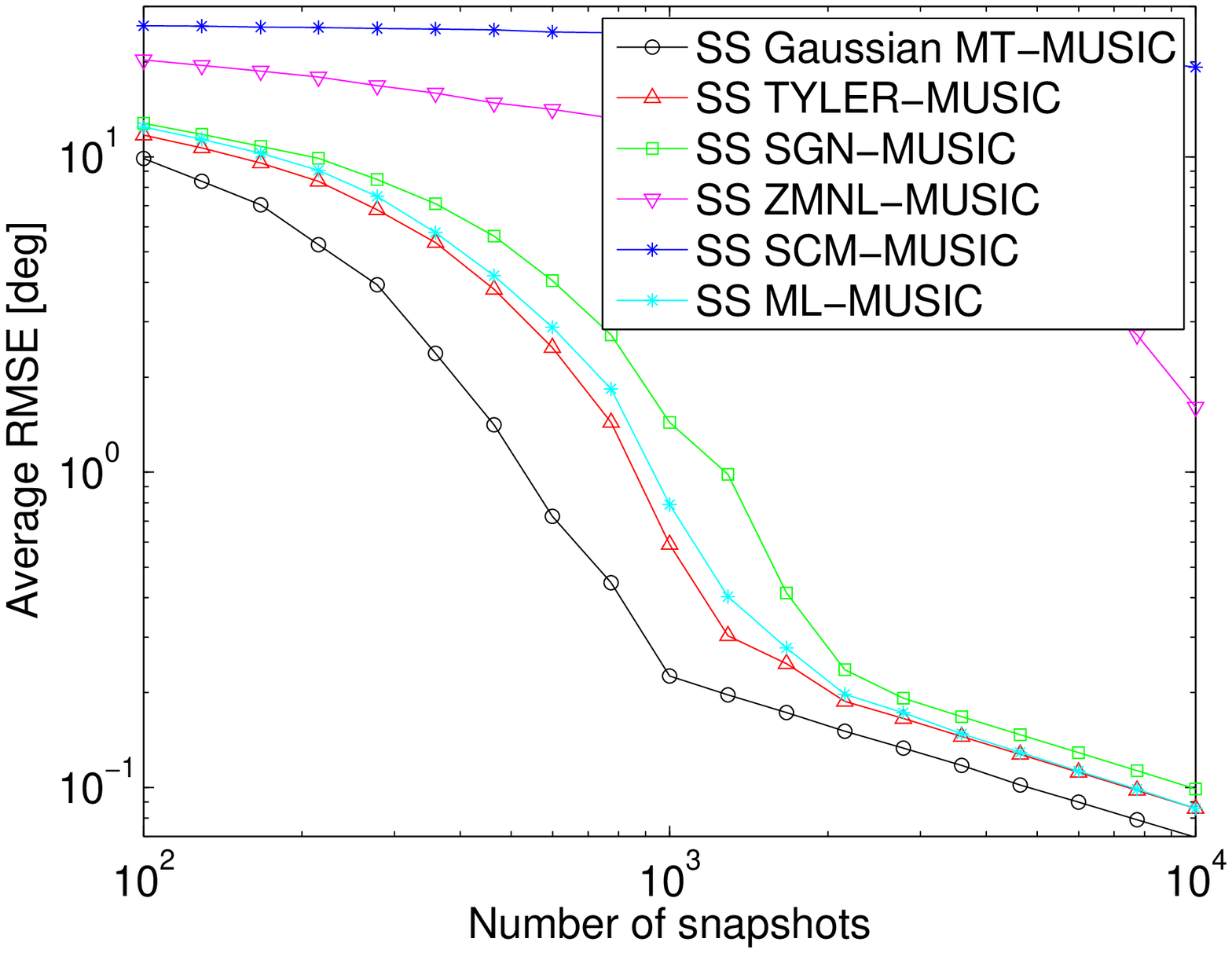}}}}
    {{\subfigure[]{\includegraphics[scale = 0.235]{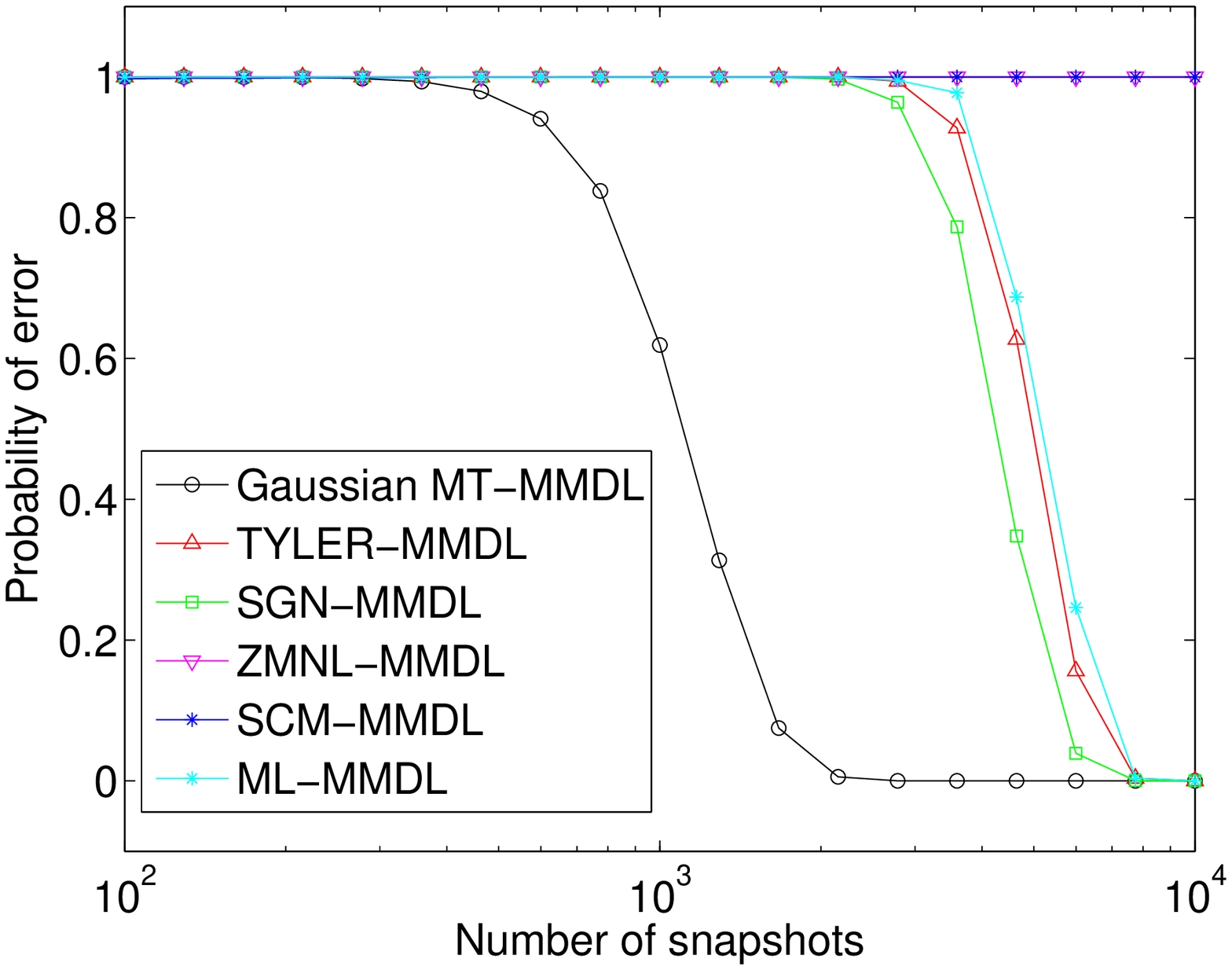}}}}
      \end{center}  
  \caption{\textbf{Coherent signals in spherical compound Gaussian noise with inverse-Gaussian texture and shape parameter $\lambda=0.1$:}
   (a) Average RMSE versus GSNR. (b) Probability of error for estimating the number of signals versus GSNR. 
   (c) Average RMSE versus the number of snapshots.  
   (d) Probability of error for estimating the number of signals versus the number of snapshots. The performance measures versus GSNR were evaluated for $N=1000$ i.i.d. snapshots. The performance measures versus the number of     snapshots were evaluated for ${\rm{GSNR}}=-24$ [dB]. Note that similarly to the non-coherent case, the Gaussian MT-MUSIC estimator has significantly lower GSNR and sample size threshold regions than the other methods.  
  Also notice that the Gaussian MT-MMDL estimator of the number of signals outperforms all other MDL based estimators with significantly lower probability of error at the low GSNR and low sample size regimes.}
\label{IG_COHERENT}  
\end{figure}
\section{Conclusion}
\label{Conclusions}
In this paper, a new framework for robust MUSIC was proposed that applies a transform to the probability distribution of the data prior to forming the sample covariance. 
Under the assumption of spherically contoured noise distribution, a new robust MUSIC algorithm, called Gaussian MT-MUSIC, was presented based on a Gaussian shaped measure transform (MT) function. Furthermore, a new robust generalization of the MDL criterion for estimating the number of signals, called MT-MDL, was derived that is based on replacing the eigenvalues of the SCM with those of the empirical MT-covariance. The proposed Gaussian MT-MUSIC algorithm was extended to the case of coherent signals by applying spatial smoothing to the empirical Gaussian MT-covariance. Exploration of other classes of MT-functions may result in additional robust MUSIC algorithms that have different useful properties. Furthermore, extending the MT-MDL criterion to sample-starved scenarios \cite{R1} or to cases where there is additional information on the sample eigenvalues distribution \cite{R2} are worthwhile topics for future research.
\appendix
\subsection{Proof Proposition \ref{Prop1}:}
\label{Prop1Proof}
\flushleft\textbf{Property \ref{P1}:}\\
Since $\varphi_{u}\left(\xvec\right)$ is nonnegative, then by Corollary 2.3.6 in \cite{MeasureTheory} $\qx$ is a measure on $\mathcal{S}_{\XCalsc}$. Furthermore, $\qx\left(\XCal\right)=1$ so that $\qx$ is a probability measure on $\mathcal{S}_{\XCalsc}$.\\
\textbf{Property \ref{P2}:}\\
Follows from definitions 4.1.1 and 4.1.3 in \cite{MeasureTheory}.\\
\textbf{Property \ref{P3}:}\\
Equivalently, we show that if the covariance of $\gvec\left(\Xmat\right)$ under $\qx$ is singular, then it must be singular under $\px$. According to (\ref{MeasureTransformRadNik}), the covariance of $\gvec\left(\Xmat\right)$ under $\qx$ is given by:
\begin{equation}
\nonumber
\bSigma^{\left(u\right)}_{\gvec(\xvec)}\triangleq{\rm{E}}\left[(\gvec\left(\Xmat\right)-\muvec^{\left(u\right)}_{\gvec(\xvec)})(\gvec(\Xmat)-\muvec^{\left(u\right)}_{\gvec(\xvec)})^{H}\varphi_{u}\left(\Xmat\right);\px\right],
\end{equation}
where $\muvec^{\left(u\right)}_{\gvec\left(\xvec\right)}\triangleq{\rm{E}}\left[\gvec\left(\Xmat\right)\varphi_{u}\left(\Xmat\right);\px\right]$ is the expectation of $\gvec\left(\Xmat\right)$ under $\qx$. Since $\bSigma^{\left(u\right)}_{\gvec\left(\xvec\right)}$ is singular, there exists a non-zero vector $\avec\in\Csp^{m}$ such that
\begin{equation}
\nonumber
\avec^{H}\bSigma^{\left(u\right)}_{\gvec\left(\xvec\right)}\avec={\rm{E}}\left[\left|\avec^{H}\left(\gvec\left(\Xmat\right)
-\muvec^{\left(u\right)}_{\gvec\left(\xvec\right)}\right)\right|^{2}\varphi_{u}\left(\Xmat\right);\px\right]
=0.
\end{equation}
Therefore, by (\ref{ExpDef}), (\ref{VarPhiDef}), the strict positiveness of $u\left(\cdot\right)$ and Proposition 2.3.9 in \cite{MeasureTheory}
\begin{equation}
\label{Zeroae}
{\rm{E}}\left[\left|\avec^{H}\left(\gvec\left(\Xmat\right)
-\muvec^{\left(u\right)}_{\gvec\left(\xvec\right)}\right)\right|^{2};\px\right]=0.
\end{equation}
The covariance of $\gvec\left(\Xmat\right)$ under $\px$ given by
\begin{equation} 
\nonumber
\bSigma_{\gvec\left(\xvec\right)}\triangleq{\rm{E}}\left[\left(\gvec\left(\Xmat\right)-\muvec_{\gvec\left(\xvec\right)}\right)\left(\gvec\left(\Xmat\right)-\muvec_{\gvec\left(\xvec\right)}\right)^{H};\px\right],
\end{equation}
where $\muvec_{\gvec\left(\xvec\right)}\triangleq{\rm{E}}\left[\gvec\left(\Xmat\right);\px\right]$ is the expectation of $\gvec\left(\Xmat\right)$ under $\px$. Hence,
one can verify that
\begin{eqnarray}
\label{SingSigma}
\avec^{H}\bSigma_{\gvec\left(\xvec\right)}\avec&+&\left|\avec^{H}\left(\muvec_{\gvec\left(\xvec\right)}-\muvec^{\left(u\right)}_{\gvec\left(\xvec\right)}\right)\right|^{2}
\\\nonumber&=&{\rm{E}}\left[\left|\avec^{H}\left(\gvec\left(\Xmat\right)-\muvec^{\left(u\right)}_{\gvec\left(\xvec\right)}\right)\right|^{2};\px\right]=0,
\end{eqnarray}
where the last equality stems from (\ref{Zeroae}).
Since the second summand in the l.h.s. of (\ref{SingSigma}) is nonnegative we conclude that $\avec^{H}\bSigma_{\gvec\left(\xvec\right)}\avec=0$, which implies that $\bSigma_{\gvec\left(\xvec\right)}$ is singular.\qed
\subsection{Proof Proposition \ref{ConsistentEst}:}
\label{Prop2Proof}
By the real-imaginary decompositions of the MT-covariance ${\bSigma}^{\left(u\right)}_{\Xmatsc}$ and its empirical version $\hat{\bSigma}^{\left(u\right)}_{\Xmatsc}$, given in Lemmas \ref{RICompSigma} and \ref{RICompSigmaHat} in Appendix \ref{RealImag}, it is sufficient to prove that $\hat{\bSigma}^{\left(g\right)}_{\Zmatsc}\rightarrow{\bSigma}^{\left(g\right)}_{\Zmatsc}$ a.s.  as $N\rightarrow\infty$.  According to (\ref{Rx_u_Est})-(\ref{hat_varphi})
\begin{eqnarray}
\label{SigmaUV_hat_Lim}
\nonumber
\lim\limits_{N\rightarrow\infty}\hat{\bSigma}^{\left(g\right)}_{\Zmatsc}&=&
\lim\limits_{N\rightarrow\infty}\frac{1}{N}\sum\limits_{n=1}^{N}\Zmat\left(n\right)\Zmat^{T}\left(n\right)\hat{\varphi}_{g}\left(\Zmat\left(n\right)\right)
\\&-&\lim\limits_{N\rightarrow\infty}\hat{\muvec}^{\left(g\right)}_{\Zmatsc}
\lim\limits_{N\rightarrow\infty}\hat{\muvec}^{\left(g\right)T}_{\Zmatsc},
\end{eqnarray}
where
\begin{equation}
\label{UV_Lim}
\begin{gathered}
\lim\limits_{N\rightarrow\infty}\frac{1}{N}\sum\limits_{n=1}^{N}\Zmat\left(n\right)\Zmat^{T}\left(n\right)\hat{\varphi}_{g}\left(\Zmat\left(n\right)\right)
\\=\frac{\lim\limits_{N\rightarrow\infty}\frac{1}{N}\sum\limits_{n=1}^{N}\Zmat\left(n\right)\Zmat^{T}\left(n\right)g\left(\Zmat\left(n\right)\right)}
{\lim\limits_{N\rightarrow\infty}\frac{1}{N}\sum\limits_{n=1}^{N}g\left(\Zmat\left(n\right)\right)}
\end{gathered}
\end{equation}
and
\begin{equation}
\label{muU_Lim}
\lim\limits_{N\rightarrow\infty}\hat{\muvec}^{\left(u\right)}_{\Zmatsc}=
\frac{\lim\limits_{N\rightarrow\infty}\frac{1}{N}\sum\limits_{n=1}^{N}\Zmat\left(n\right)g\left(\Zmat\left(n\right)\right)}
{\lim\limits_{N\rightarrow\infty}\frac{1}{N}\sum\limits_{n=1}^{N}g\left(\Zmat\left(n\right)\right)}.
\end{equation}
 
In the following, the limits of the series in the r.h.s. of (\ref{UV_Lim}) and (\ref{muU_Lim}) are obtained. Since $\left\{\Zmat\left(n\right)\right\}_{n=1}^{N}$ is a sequence of i.i.d. samples of $\Zmat$, then the random matrices 
$\left\{\Zmat\left(n\right)\Zmat^{T}\left(n\right)g\left(\Zmat\left(n\right),\Zmat\left(n\right)\right)\right\}_{n=1}^{N}$ in the r.h.s. of (\ref{UV_Lim}) define a sequence of i.i.d. samples of $\Zmat\Zmat^{T}g\left(\Zmat\right)$. Moreover, if the condition in (\ref{Cond11}) is satisfied, then for any pair of entries $Z_{k}$, $Z_{l}$ of $\Zmat$ we have that
\begin{eqnarray}
\nonumber
&&{\rm{E}}\left[\left|Z_{k}Z_{l}g\left(\Zmat\right)\right|;P_{\Zmatsc}\right]
\\\nonumber
&=&
{\rm{E}}\left[\left|Z_{k}g^{1/2}\left(\Zmat\right)\right|\left|Z_{l}g^{1/2}\left(\Zmat\right)\right|;P_{\Zmatsc}\right]
\\\nonumber&\leq&
\left({\rm{E}}\left[\left|Z_{k}\right|^{2}g\left(\Zmat\right);P_{\Zmatsc}\right]
{\rm{E}}\left[\left|Z_{l}\right|^{2}g\left(\Zmat\right);P_{\Zmatsc}\right]\right)^{1/2}\\\nonumber
&\leq&
{\rm{E}}\left[\left\|\Xmat\right\|^{2}_{2}u\left(\Xmat\right);\px\right]<\infty,
\end{eqnarray}
where the first semi-inequality stems from the H\"older inequality for random variables \cite{MeasureTheory}, and the second one stems from the definitions of $\Zmat$ and $g\left(\Zmat\right)$ in Lemma \ref{RICompSigma}  in Appendix \ref{RealImag}. Therefore, by Khinchine's strong law of large numbers (KSLLN) \cite{Folland}
\begin{equation}
\label{NomxyLim}
\lim\limits_{N\rightarrow\infty}\frac{1}{N}\sum\limits_{n=1}^{N}\Zmat\left(n\right)\Zmat^{T}\left(n\right)g\left(\Zmat\left(n\right)\right)
={\rm{E}}\left[\Zmat\Zmat^{T}g\left(\Zmat\right);P_{\Zmatsc}\right]\hspace{0.03cm}{\rm{a.s.}}
\end{equation}
Similarly, it can be shown that if the condition in (\ref{Cond11}) is satisfied, then by the KSLLN 
\begin{equation}
\label{NomxLim}
\lim\limits_{N\rightarrow\infty}\frac{1}{N}\sum\limits_{n=1}^{N}\Zmat\left(n\right)g\left(\Zmat\left(n\right)\right)
={\rm{E}}\left[\Zmat{g}\left(\Zmat\right);P_{\Zmatsc}\right]\hspace{0.2cm}{\rm{a.s.}},
\end{equation}
and
\begin{equation}
\label{DenLim}
\lim\limits_{N\rightarrow\infty}\frac{1}{N}\sum\limits_{n=1}^{N}g\left(\Zmat\left(n\right)\right)
={\rm{E}}\left[g\left(\Zmat\right);P_{\Zmatsc}\right]\hspace{0.2cm}{\rm{a.s.}}
\end{equation}
\begin{Remark}
\label{Remark3}
By (\ref{DenLim}), the definition of $g\left(\Zmat\right)$ in Lemma \ref{RICompSigma}  in Appendix \ref{RealImag}, and the assumption in (\ref{Cond}) the denominator in the r.h.s. of (\ref{UV_Lim}) and (\ref{muU_Lim})  is non-zero almost surely.
\end{Remark}

Therefore, since the sequences in the l.h.s. of (\ref{UV_Lim}) and (\ref{muU_Lim}) are obtained by continuous mappings of the elements of the sequences in their r.h.s., then by  (\ref{NomxyLim})-(\ref{DenLim}), and the Mann-Wald Theorem \cite{MannWald}
\begin{eqnarray}
\label{XY_Lim_Closed}
&&\lim\limits_{N\rightarrow\infty}\frac{1}{N}\sum\limits_{n=1}^{N}\Zmat\left(n\right)\Zmat^{T}\left(n\right)\hat{\varphi}_{g}\left(\Zmat\left(n\right)\right)
\\\nonumber
&=&\frac{{\rm{E}}\left[\Zmat\Zmat^{T}g\left(\Zmat\right);P_{\Zmatsc}\right]}{{\rm{E}}\left[g\left(\Zmat\right);P_{\Zmatsc}\right]}
={\rm{E}}\left[\Zmat\Zmat^{T}\varphi_{g}\left(\Zmat\right);P_{\Zmatsc}\right]\hspace{0.2cm}{\rm{a.s.}}
\end{eqnarray}
and
\begin{equation}
\label{mux_Lim_Closed}
\lim\limits_{N\rightarrow\infty}\hat{\muvec}^{\left(g\right)}_{\Zmatsc}
=\frac{{\rm{E}}\left[\Zmat{g}\left(\Zmat\right);P_{\Zmatsc}\right]}{{\rm{E}}\left[g\left(\Zmat\right);P_{\Zmatsc}\right]}
={\rm{E}}\left[\Zmat\varphi_{g}\left(\Zmat\right);P_{\Zmatsc}\right]\hspace{0.2cm}{\rm{a.s.}},
\end{equation}
where the last equalities in (\ref{XY_Lim_Closed}) and (\ref{mux_Lim_Closed}) follow from the definition of $\varphi_{g}\left(\cdot\right)$ in (\ref{VarPhiDef}).

Thus, since the sequence in the l.h.s. of (\ref{SigmaUV_hat_Lim}) is obtained by continuous mappings of the elements of the sequences in its r.h.s., then by (\ref{XY_Lim_Closed}) and (\ref{mux_Lim_Closed}), the Mann-Wald Theorem, and (\ref{MTCovZ}) it is concluded that $\hat{\bSigma}^{\left(g\right)}_{\Zmatsc}\rightarrow{\bSigma}^{\left(g\right)}_{\Zmatsc}$ a.s. as $N\rightarrow\infty$.\qed
\subsection{Real-imaginary decomposition of the MT-covariance and its empirical estimate}
\label{RealImag}
The real-imaginary decomposition of the MT-covariance of $\Xmat$ under $\qx$ is given in the following Lemma that follows directly from (\ref{VarPhiDef}), (\ref{MTCovZ}) and (\ref{MTMean}).
\begin{Lemma} 
\label{RICompSigma} 
Let $\Umat$ and $\Vmat$ denote the real and imaginary components of $\Xmat$, respectively. Define the real random vector $\Zmat\triangleq\left[\Umat^{T},\Vmat^{T}\right]^{T}$ and the MT-function $g\left(\Zmat\right)\triangleq{u}\left(\Xmat\right)$.
Let $\left[{\bSigma}^{\left(g\right)}_{\Zmatsc}\right]_{1,1}$, $\left[{\bSigma}^{\left(g\right)}_{\Zmatsc}\right]_{1,2}$, $\left[{\bSigma}^{\left(g\right)}_{\Zmatsc}\right]_{2,1}$ and $\left[{\bSigma}^{\left(g\right)}_{\Zmatsc}\right]_{2,2}$ denote the $p\times{p}$ submatrices of the real-valued MT-covariance ${\bSigma}^{\left(g\right)}_{\Zmatsc}$ satisfying
$$
{\bSigma}^{\left(g\right)}_{\Zmatsc}
=\left[{\begin{array}{*{20}c} \left[{\bSigma}^{\left(g\right)}_{\Zmatsc}\right]_{1,1} & \left[{\bSigma}^{\left(g\right)}_{\Zmatsc}\right]_{1,2}   \\
\left[{\bSigma}^{\left(g\right)}_{\Zmatsc}\right]_{2,1} & \left[{\bSigma}^{\left(g\right)}_{\Zmatsc}\right]_{2,2}\end{array}}\right].
$$
The real-imaginary decomposition of the MT-covariance of $\Xmat$ under $\qx$ takes the form:
\begin{eqnarray}
\label{RIMTCov}
{\bSigma}^{\left(u\right)}_{\Xmatsc}&=&\left(\left[{\bSigma}^{\left(g\right)}_{\Zmatsc}\right]_{1,1} + \left[{\bSigma}^{\left(g\right)}_{\Zmatsc}\right]_{2,2}\right)
\\\nonumber
&+& i\left(\left[{\bSigma}^{\left(g\right)}_{\Zmatsc}\right]_{2,1} - \left[{\bSigma}^{\left(g\right)}_{\Zmatsc}\right]_{1,2}\right).
\end{eqnarray}
\end{Lemma}
Similarly, the real-imaginary decomposition of the empirical MT-covariance of $\Xmat$ under $\qx$ is given in the following Lemma that follows directly from (\ref{Rx_u_Est})-(\ref{hat_varphi}).
\begin{Lemma}
\label{RICompSigmaHat}  
Let $\Xmat\left(n\right)$, $n=1,\ldots,N$ denote a sequence of samples from $\px$, and let $\Umat\left(n\right)$, $\Vmat\left(n\right)$ denote the real and imaginary components of $\Xmat\left(n\right)$, respectively.
Define the real random vector $\Zmat\left(n\right)\triangleq\left[\Umat^{T}\left(n\right),\Vmat^{T}\left(n\right)\right]^{T}$ and the MT-function $g\left(\Zmat\left(n\right)\right)\triangleq{u}\left(\Xmat\left(n\right)\right)$.
Let $\left[\hat{\bSigma}^{\left(g\right)}_{\Zmatsc}\right]_{1,1}$, $\left[\hat{\bSigma}^{\left(g\right)}_{\Zmatsc}\right]_{1,2}$, $\left[\hat{\bSigma}^{\left(g\right)}_{\Zmatsc}\right]_{2,1}$ and $\left[\hat{\bSigma}^{\left(g\right)}_{\Zmatsc}\right]_{2,2}$ denote the $p\times{p}$ submatrices of the real-valued empirical MT-covariance $\hat{\bSigma}^{\left(g\right)}_{\Zmatsc}$ satisfying
$$
\hat{\bSigma}^{\left(g\right)}_{\Zmatsc}
=\left[{\begin{array}{*{20}c} \left[\hat{\bSigma}^{\left(g\right)}_{\Zmatsc}\right]_{1,1} & \left[\hat{\bSigma}^{\left(g\right)}_{\Zmatsc}\right]_{1,2}   \\
\left[\hat{\bSigma}^{\left(g\right)}_{\Zmatsc}\right]_{2,1} & \left[\hat{\bSigma}^{\left(g\right)}_{\Zmatsc}\right]_{2,2}\end{array}}\right].
$$
The real-imaginary decomposition of the empirical MT-covariance of $\Xmat$ under $\qx$ takes the form:
\begin{eqnarray}
\label{RIMTCov}
\hat{\bSigma}^{\left(u\right)}_{\Xmatsc}&=&\left(\left[\hat{\bSigma}^{\left(g\right)}_{\Zmatsc}\right]_{1,1} + \left[\hat{\bSigma}^{\left(g\right)}_{\Zmatsc}\right]_{2,2}\right)\\\nonumber
&+& i\left(\left[\hat{\bSigma}^{\left(g\right)}_{\Zmatsc}\right]_{2,1} - \left[\hat{\bSigma}^{\left(g\right)}_{\Zmatsc}\right]_{1,2}\right).
\end{eqnarray}
\end{Lemma}
\subsection{Proof Proposition \ref{RobustnessConditions}:}
\label{InfBoundProof}
The influence function (\ref{MT_COV_INF}) can be written as:
\begin{equation}
\begin{gathered}
\label{IFREW}
{\rm{IF}}_{{\tiny{\Psimat}}_{\xvec}^{(u)}}\left(\yvec;\px\right)={c{u\left(\yvec\right)}\left({\left\|\yvec-\muvec^{\left(u\right)}_{\Xmatsc}\right\|}^{2}_{2}\Gmat\left(\yvec\right) - {\bSigma}^{\left(u\right)}_{\Xmatsc}\right)},
\end{gathered}
\end{equation}
where $c\triangleq{\rm{E}}^{-1}[u(\Xmat);\px]$, $\Gmat\left(\yvec\right)\triangleq\psivec\left(\yvec\right)\psivec^{H}\left(\yvec\right)$ and
\begin{equation}
\nonumber
\psivec\left(\yvec\right)\triangleq\frac{\yvec-\muvec^{\left(u\right)}_{\Xmatsc}}{{\left\|\yvec-\muvec^{\left(u\right)}_{\Xmatsc}\right\|}_{2}}.
\end{equation}
Since ${\left\|\psivec\left(\yvec\right)\right\|}_{2}={1}$ for any $\yvec\in\Csp^{p}$, the real and imaginary components of $\Gmat\left(\yvec\right)$ in (\ref{IFREW}) are bounded. Thus, the influence function ${\rm{IF}}_{{\tiny{\Psimat}}_{\xvec}^{(u)}}\left(\yvec;\px\right)$ is bounded if $u\left(\yvec\right)$ and $u\left(\yvec\right)\left\|\yvec\right\|^{2}_{2}$ are bounded. \qed
\subsection{Proof Proposition \ref{InfFuncLim}:}
\label{InfFuncLimProof}
According to (\ref{MT_COV_INF}) 
\begin{eqnarray}
\label{InfFuncGauss}
&&\left\|{\rm{IF}}_{{\tiny{\Psimat}}_{\xvec}^{(u)}}\left(\yvec;\px\right)\right\|^{2}_{\textrm{Fro}}
\\\nonumber
&=&c{u}^{2}\left(\yvec\right)\Bigg(\left\|\yvec-\muvec^{\left(u\right)}_{\Xmatsc}\right\|_{2}^{4} 
-2\left\|\yvec-\muvec^{\left(u\right)}_{\Xmatsc}\right\|^{2}_{{\bSigmasc^{\left(u\right)}_{\xvec}}}
\\\nonumber&+&{\rm{tr}}\left[\left(\bSigma^{\left(u\right)}_{\Xmatsc}\right)^{2}\right]\Bigg)\\\nonumber
 &\leq&{c}\Bigg(\left(\sqrt{u\left(\yvec\right)}\left\|\yvec\right\|_{2} + \sqrt{u\left(\yvec\right)}\left\|\muvec^{\left(u\right)}_{\Xmatsc}\right\|_{2}\right)^{4}
\\\nonumber&+&{u}^{2}\left(\yvec\right){\rm{tr}}\left[\left(\bSigma^{\left(u\right)}_{\Xmatsc}\right)^{2}\right]\Bigg),
\end{eqnarray}
where $c\triangleq{\rm{E}}^{-2}[u(\Xmat);\px]$, $\left\|\yvec\right\|^{2}_{\bSigmasc}\triangleq\yvec^{H}\bSigma\yvec$, and the semi-inequality follows from the triangle-inequality and the positive semi-definiteness of $\bSigma^{\left(u\right)}_{\Xmatsc}$.
By (\ref{GaussKernel}),  
$$u\left(\yvec\right)=\uGausss\left(\yvec;\tau\right)=\left(\pi\tau^{2}\right)^{-p}\phi\left(r\right)$$ and 
$$u\left(\yvec\right)\left\|\yvec\right\|^{2}_{2}=\left(\pi\tau^{2}\right)^{-p}\left(\frac{\tau^{4}}{4}\frac{\partial^{2}\phi\left(r\right)}{\partial{r^{2}}} + \frac{\tau^{2}}{2}\phi\left(r\right)\right),$$ where $r\triangleq\left\|\yvec\right\|_{2}$ and 
$\phi\left(r\right)\triangleq{\exp}\left({-r^{2}/\tau^{2}}\right)$. Therefore, since for any fixed $\tau$ we have $\phi\left(r\right)\rightarrow{0}$ and $\frac{\partial^{2}\phi\left(r\right)}{\partial{r^{2}}}\rightarrow{0}$ as $r\rightarrow\infty$, we conclude that (\ref{IFLim}) holds.  \qed
\subsection{Proof of Theorem \ref{CompGaussStruct}}
\label{CompGaussStructProof}
According to  (\ref{ArrayModel}) and (\ref{CompGauss}), the conditional probability distribution, $P_{\Xmatsc|\nu,\Smatsc}$, of $\Xmat$ given $\nu$ and $\Smat$  is proper complex normal with location parameter $\muvec_{\Xmatsc|\nu,\Smatsc}=\Amat\Smat$ and covariance matrix $\bSigma_{\Xmatsc|\nu,\Smatsc}=\nu^{2}\Imat$. Therefore, using (\ref{ExpDef}), (\ref{GaussKernel}) and the law of total expectation one can verify that:
\begin{equation}
\label{ExpCopm1}
{\rm{E}}\left[\uGausss\left(\Xmat;\tau\right);\px\right]
={\rm{E}}\left[g\left(\alpha,\Smat;\tau\right);P_{\alpha,\Smatsc}\right]={\rm{E}}\left[h\left(\alpha;\tau\right);P_{\alpha}\right],
\end{equation}
\begin{equation}
\label{ExpCopm2}
{\rm{E}}\left[\Xmat\uGausss\left(\Xmat;\tau\right);\px\right]
=\Amat{\rm{E}}\left[\alpha^{2}\Smat{g}\left(\alpha,\Smat;\tau\right);P_{\alpha,\Smatsc}\right],
\end{equation}
and
\begin{eqnarray}
\label{ExpCopm3}
\nonumber
&&{\rm{E}}\left[\Xmat\Xmat^{H}\uGausss\left(\Xmat;\tau\right);\px\right]\\\nonumber
&=&\Amat{\rm{E}}\left[\alpha^{4}\Smat\Smat^{H}{g}\left(\alpha,\Smat;\tau\right);P_{\alpha,\Smatsc}\right]\Amat^{H}
\\&+&{\rm{E}}\left[\alpha^{2}\nu^{2}h\left(\alpha;\tau\right);P_{\nu}\right]\Imat,
\end{eqnarray}
where $\alpha\triangleq\sqrt{\frac{\tau^{2}}{\tau^{2}+\nu^{2}}}$, $$g\left(\alpha,\Smat;\tau\right)\triangleq{\left(\frac{\pi\tau^{2}}{\alpha^{2}}\right)}^{-p}\exp{(-{\alpha^{2}\|\Amat\Smat\|^{2}_{2}}/{\tau^{2}})},$$ and
$h\left(\alpha;\tau\right)\triangleq{\rm{E}}\left[g\left(\alpha,\Smat;\tau\right);P_{\Smatsc}\right]$.
Notice that by (\ref{CompGauss}) the expectation in the second summand of (\ref{ExpCopm3}) can be rewritten as:
\begin{eqnarray}
\label{ExpCopm4}
{\rm{E}}\left[\alpha^{2}\nu^{2}h\left(\alpha;\tau\right);P_{\nu}\right]&=&{\rm{E}}\left[\left|\alpha{W}\right|^{2}h\left(\alpha,\tau\right);P_{\alpha,W}\right]
\\\nonumber&-&\left|{\rm{E}}\left[\alpha{W}h\left(\alpha,\tau\right);P_{\alpha,W}\right]\right|^{2},
\end{eqnarray}
where the scalar $W$ denotes any of the identically distributed components of the noise vector $\Wmat$. Thus, using (\ref{VarPhiDef}), (\ref{MTCovZ}) and (\ref{ExpCopm1})-(\ref{ExpCopm4}) the relation (\ref{GaussMTCov}) is easily obtained. Finally, since the MT-function $g(\cdot,\cdot;\cdot)$ is strictly positive and $\alpha^{2}\Smat$ has a non-singular covariance under the joint probability measure $P_{\alpha,\Smat}$, by Property \ref{P3} in Proposition \ref{Prop1} the MT-covariance $\bSigma^{(g)}_{\alpha^{2}\Smatsc}(\tau)$ must be non-singular. \qed
\subsection{Proof of Theorem \ref{TauProp}}
\label{TauPropProof}
 Since $P_{\Xmatsc}$ is proper complex normal with location parameter $\muvec_{\Xmatsc}=\zerovec$ and covariance matrix $\bSigma_{\Xmatsc}$,  by (\ref{VarPhiDef}) and (\ref{MeasureTransformRadNik}) the transformed probability distribution $Q^{\left(\uGausss\right)}_{\Xmatsc}$ generated by the Gaussian MT-function (\ref{GaussKernel}) is proper complex normal with location parameter $\muvec_{\Xmatsc}^{\left(\uGausss\right)}=\zerovec$ and covariance matrix 
\begin{equation}
\label{GaussCovGauss}
\bSigma_{\Xmatsc}^{\left(\uGausss\right)}\left(\tau\right)=\left(\bSigma^{-1}_{\Xmatsc}+\tau^{-2}\Imat\right)^{-1}.
\end{equation}

According to the Gaussian Fisher information formula, stated in \cite{Schreier}, the Fisher informations for estimating $\theta_{k}$ under $\px$ and $Q^{\left(\uGausss\right)}_{\Xmatsc}$ are given by:
\begin{equation}
\label{GFIMP}
F\left(\theta_{k};\px\right)={\rm{tr}}\left[\left(\bSigma^{-1}_{\Xmatsc}\frac{\partial\bSigma_{\Xmatsc}}{\partial\theta_{k}}\right)^{2}\right]
\end{equation}
and
\begin{equation}
\label{GFIMQ}
F\left(\theta_{k};Q^{\left(\uGausss\right)}_{\Xmatsc}\right)={\rm{tr}}\left[\left(\left(\bSigma^{\left(\uGausss\right)}_{\Xmatsc}\left(\tau\right)\right)^{-1}\frac{\partial\bSigma^{\left(\uGausss\right)}_{\Xmatsc}\left(\tau\right)}{\partial\theta_{k}}\right)^{2}\right],
\end{equation}
respectively. By (\ref{GaussCovGauss}) and the matrix identity $\frac{\partial{\Cmat^{-1}}}{\partial\alpha}=-\Cmat^{-1}\frac{\partial{\Cmat}}{\partial\alpha}\Cmat^{-1}$ \cite{Cook}, where $\Cmat$ is some invertible complex matrix and $\alpha\in\Rsp$, we have 
\begin{equation}
\label{DbSigma}
\frac{\partial\bSigma_{\Xmatsc}^{\left(\uGausss\right)}}{\partial{\theta_{k}}}=\bSigma_{\Xmatsc}^{\left(\uGausss\right)}\bSigma_{\Xmatsc}^{-1}\frac{\partial\bSigma_{\Xmatsc}}{\partial{\theta_{k}}}\bSigma_{\Xmatsc}^{-1}
\bSigma_{\Xmatsc}^{\left(\uGausss\right)}.
\end{equation}

Therefore, using (\ref{GaussCovGauss})-(\ref{DbSigma}) and Lemma \ref{Lemma1} in Appendix \ref{TraceIneqApp} we obtain that:
\begin{eqnarray}
\label{GFIMQ3}
&&F\left(\theta_{k};Q^{\left(\uGausss\right)}_{\Xmatsc}\right)
\\\nonumber&\leq&
{\rm{tr}}\left[\left(\bSigma^{-1}_{\Xmatsc}\frac{\partial\bSigma_{\Xmatsc}}{\partial\theta_{k}}\right)^{2}\right]\lambda^{2}_{\rm{max}}\left(\bSigma^{-1}_{\Xmatsc}\bSigma^{\left(\uGausss\right)}_{\Xmatsc}\left(\tau\right)\right)
\\\nonumber&=&
F\left(\theta_{k};\px\right)\left(\frac{\tau^{2}}{\tau^{2}+\lambda_{\rm{min}}\left(\bSigma_{\Xmatsc}\right)}\right)^{2}
\end{eqnarray}
and
\begin{eqnarray}
\label{GFIMQ4}
&&F\left(\theta_{k};Q^{\left(\uGausss\right)}_{\Xmatsc}\right)
\\\nonumber&\geq&
{\rm{tr}}\left[\left(\bSigma^{-1}_{\Xmatsc}\frac{\partial\bSigma_{\Xmatsc}}{\partial\theta_{k}}\right)^{2}
\right]\lambda^{2}_{\rm{min}}\left(\bSigma^{-1}_{\Xmatsc}\bSigma^{\left(\uGausss\right)}_{\Xmatsc}\left(\tau\right)\right)
\\\nonumber&=&
F\left(\theta_{k};\px\right)\left(\frac{\tau^{2}}{\tau^{2}+\lambda_{\rm{max}}\left(\bSigma_{\Xmatsc}\right)}\right)^{2}.
\end{eqnarray}
The relations in (\ref{FIMRatio}) are obtained from (\ref{GFIMQ3}) and (\ref{GFIMQ4}).\qed
\subsection{Some useful trace inequalities}
\label{TraceIneqApp}  
\begin{Lemma}
\label{Lemma1}
Let $\Amat$, $\Bmat$ and $\Cmat$ denote Hermitian matrices with the same dimensions, and assume that $\Amat$ and $\Cmat$ are positive-definite and $\Bmat$ is positive-semidefinite.  
\begin{equation}
\label{TrIneq}
{{\rm{tr}}\left[\Amat\Bmat\right]}\lambda_{\rm{min}}\left(\Amat\Cmat\right){\leq}{\rm{tr}}\left[\Amat\Bmat\Amat\Cmat\right]\leq{\rm{tr}}\left[\Amat\Bmat\right]\lambda_{\rm{max}}\left(\Amat\Cmat\right),
\end{equation}
where $\lambda_{\rm{min}}\left(\cdot\right)$ and $\lambda_{\rm{max}}\left(\cdot\right)$ denote the minimal and maximal eigenvalues, respectively.
\end{Lemma}
\begin{proof}
By the invariance of the trace operator to multiplication order of two matrices, inequalities (I) and (II) in \cite{TraceIneq}, and the fact that $\Amat^{1/2}\Cmat\Amat^{1/2}$ is Hermitian and similar to $\Amat\Cmat$
we conclude that 
\begin{eqnarray}
\label{UpBound}
{\rm{tr}}\left[\Amat\Bmat\Amat\Cmat\right]
&\leq&
{\rm{tr}}\left[\Amat^{1/2}\Bmat\Amat^{1/2}\right]\left\|\Amat^{1/2}\Cmat\Amat^{1/2}\right\|_{S}
\\\nonumber&=&
{\rm{tr}}\left[\Amat\Bmat\right]\lambda_{\rm{max}}\left(\Amat\Cmat\right)
\end{eqnarray}
and
\begin{eqnarray}
\label{LowBound}
\nonumber{\rm{tr}}\left[\Amat\Bmat\Amat\Cmat\right]
&\geq&
{\rm{tr}}\left[\Amat^{1/2}\Bmat\Amat^{1/2}\right]\left\|\left(\Amat^{1/2}\Cmat\Amat^{1/2}\right)^{-1}\right\|^{-1}_{S}
\\&=&
{\rm{tr}}\left[\Amat\Bmat\right]\lambda_{\rm{min}}\left(\Amat\Cmat\right),
\end{eqnarray}
where $\left\|\cdot\right\|_{S}$ denotes the spectral-norm \cite{Horn}. 
\end{proof}
\subsection{Proof of Theorem \ref{Th2}}
\label{Th2Proof}
Assume that 
\begin{equation}
\label{AsProp}
\hat{\bSigma}^{\left(g\right)}_{\Zmatsc}-\bSigma^{\left(g\right)}_{\Zmatsc}=O\left(\sqrt{N^{-1}\log\log{N}}\right)\hspace{0.2cm}{\rm{a.s.}},
\end{equation}
where $\bSigma^{\left(g\right)}_{\Zmatsc}$ and $\hat{\bSigma}^{\left(g\right)}_{\Zmatsc}$ are defined in the real-imaginary decompositions of $\bSigma^{\left(u\right)}_{\Xmatsc}$ and $\hat{\bSigma}^{\left(u\right)}_{\Xmatsc}$, respectively, in Appendix  \ref{RealImag}. Then, $\hat{\bSigma}^{\left(u\right)}_{\Xmatsc}-\bSigma^{\left(u\right)}_{\Xmatsc}=O\left(\sqrt{N^{-1}\log\log{N}}\right)$ a.s., and therefore, by Lemma 3.2 in \cite{Zhao}
we obtain that the descendingly ordered  eigenvalues of  $\hat{\bSigma}^{\left(u\right)}_{\Xmatsc}$ and $\bSigma^{\left(u\right)}_{\Xmatsc}$ satisfy $\hat{\lambda}^{\left(u\right)}_{k}-\lambda^{\left(u\right)}_{k}=O\left(\sqrt{N^{-1}\log\log{N}}\right)$ a.s. for $k=1,\ldots,p$. Using this result and applying that  $\lambda^{\left(u\right)}_{k}$, $k=1,\ldots,p$ satisfy (\ref{EigCond}), the strong consistency of $\hat{q}$ follows directly from the proof of Theorem 3.1 in \cite{Zhao}.
Therefore, in order to complete the proof,  we show that under the condition (\ref{Cond3}) the assumption in (\ref{AsProp}) is satisfied. 

Similarly to (\ref{StatFunc}), the empirical MT-covariance $\hat{\bSigma}^{\left(g\right)}_{\Zmatsc}$ can be written as a statistical functional ${\Psimat}_{\Zmatsc}^{(g)}[\hat{P}_{\Zmatsc}]$ of the empirical probability measure $\hat{P}_{\Zmatsc}=\frac{1}{N}\sum\limits_{n=1}^{N}\delta_{\Zmatsc_{n}}$, 
where ${\Psimat}_{\Zmatsc}^{(g)}[P_{\Zmatsc}]=\bSigma^{(g)}_{\Zmatsc}$. The Taylor expansion of ${\Psimat}_{\Zmatsc}^{(g)}[\hat{P}_{\Zmatsc}]$ about ${P}_{\Zmatsc}$ is given by  \cite{Serfling}:
\begin{equation}
\label{TaylorExp}
{\Psimat}_{\Zmatsc}^{(g)}[\hat{P}_{\Zmatsc}]={\Psimat}_{\Zmatsc}^{(g)}[{P}_{\Zmatsc}]+\left.\frac{\partial{\Psimat}_{\Zmatsc}^{(g)}[(1-\epsilon){P}_{\Zmatsc}+\epsilon\hat{P}_{\Zmatsc}]}{\partial\epsilon}\right|_{\epsilon=0} + \Rmat^{\left(g\right)}_{\Zmatsc},
\end{equation}
where $\Rmat^{\left(g\right)}_{\Zmatsc}$ denotes the reminder term. 

Using (\ref{MTCovZ})-(\ref{Mu_u_Est}) and (\ref{StatFunc}), one can verify that 
\begin{eqnarray}
\label{Grad}
\nonumber
&&\left.\frac{\partial{\Psimat}_{\Zmatsc}^{(g)}[(1-\epsilon){P}_{\Zmatsc}+\epsilon\hat{P}_{\Zmatsc}]}{\partial\epsilon}\right|_{\epsilon=0}
\\\nonumber&=&\frac{{\rm{E}}\left[g\left(\Zmat\right);\hat{P_{\Zmatsc}}\right]}{{\rm{E}}\left[g\left(\Zmat\right);{P_{\Zmatsc}}\right]}
\Bigg(\hat{\bSigma}^{\left(g\right)}_{\Zmatsc}-{\bSigma}^{\left(g\right)}_{\Zmatsc}
\\\nonumber&+&\left(\hat{\muvec}^{\left(g\right)}_{\Zmatsc}-\muvec^{\left(g\right)}_{\Zmatsc}\right)\left(\hat{\muvec}^{\left(g\right)}_{\Zmatsc}-\muvec^{\left(g\right)}_{\Zmatsc}\right)^{T}\Bigg)
=
\frac{1}{N}\sum\limits_{n=1}^{N}\Hmat\left(\Zmat\left(n\right)\right)
\end{eqnarray}
is a V-statistic \cite{Serfling} with zero-mean kernel $\Hmat(\Zmat)\triangleq\frac{g(\Zmat)}{{\rm{E}}[g(\Zmat);P_{\Zmatsc}]}((\Zmat-\muvec^{(g)}_{\Zmatsc})(\Zmat-\muvec^{(g)}_{\Zmatsc})^{T}-\bSigma^{(g)}_{\Zmatsc})$. By condition (\ref{Cond3}) and the definitions of $\Zmat$ and $g\left(\cdot\right)$ in Lemma {\ref{RICompSigma}} in Appendix \ref{RealImag}, $\Hmat\left(\Zmat\right)$ must have finite variance entries. 

By (\ref{TaylorExp}) and (\ref{Grad}) we have that 
\begin{eqnarray}
\label{Reminder}
\Rmat^{\left(g\right)}_{\Zmatsc}&=&\left(1-\frac{{\rm{E}}\left[g\left(\Zmat\right);\hat{P_{\Zmatsc}}\right]}{{\rm{E}}\left[g\left(\Zmat\right);{P_{\Zmatsc}}\right]}\right)\left(\hat{\bSigma}^{\left(g\right)}_{\Zmatsc}-{\bSigma}^{\left(g\right)}_{\Zmatsc}\right)
\\\nonumber&-&\frac{{\rm{E}}\left[g\left(\Zmat\right);\hat{P_{\Zmatsc}}\right]}{{\rm{E}}\left[g\left(\Zmat\right);{P_{\Zmatsc}}\right]}\left(\hat{\muvec}^{\left(g\right)}_{\Zmatsc}-\muvec^{\left(g\right)}_{\Zmatsc}\right)\left(\hat{\muvec}^{\left(g\right)}_{\Zmatsc}-\muvec^{\left(g\right)}_{\Zmatsc}\right)^{T}.
\end{eqnarray}
Therefore, each entry of $\Rmat^{\left(g\right)}_{\Zmatsc}$ satisfies:
\begin{eqnarray}
\label{RemIneq}
\left|\left[\Rmat^{\left(g\right)}_{\Zmatsc}\right]_{k,l}\right|&\leq&{\left|C\right|}\left|\left[\hat{\bSigma}^{\left(g\right)}_{\Zmatsc}\right]_{k,l}-\left[{\bSigma}^{\left(g\right)}_{\Zmatsc}\right]_{k,l}\right|
\\\nonumber&+&\left|D\right|\left| \left[\hat{\muvec}^{\left(g\right)}_{\Zmatsc}\right]_{l}-\left[\muvec^{\left(g\right)}_{\Zmatsc}\right]_{l} \right|\hspace{0.2cm}{\rm{a.s.}},
\end{eqnarray}
where
\begin{eqnarray}
\label{Cdef}
C&\triangleq&1-\frac{{\rm{E}}\left[g\left(\Zmat\right);\hat{P_{\Zmatsc}}\right]}{{\rm{E}}\left[g\left(\Zmat\right);{P_{\Zmatsc}}\right]}
\\\nonumber&=&c\frac{1}{N}\sum\limits_{n=1}^{N}\left(g\left(\Zmat\left(n\right)\right) - {\rm{E}}\left[g\left(\Zmat\right);{P_{\Zmatsc}}\right]\right),
\end{eqnarray}
\begin{eqnarray}
\label{Ddef}
D&\triangleq&\frac{{\rm{E}}\left[g\left(\Zmat\right);\hat{P_{\Zmatsc}}\right]}{{\rm{E}}\left[g\left(\Zmat\right);{P_{\Zmatsc}}\right]}\left(\left[\hat{\muvec}^{\left(g\right)}_{\Zmatsc}\right]_{k}-\left[\muvec^{\left(g\right)}_{\Zmatsc}\right]_{k} \right)
\\\nonumber&=&c\frac{1}{N}\sum\limits_{n=1}^{N}\left(Z_{k}\left(n\right) - \frac{{\rm{E}}\left[Z_{k}g\left(\Zmat\right);P_{\Zmatsc}\right]}{{\rm{E}}\left[g\left(\Zmat\right);P_{\Zmatsc}\right]}\right)g\left(\Zmat\left(n\right)\right),
\end{eqnarray}
$c\triangleq{\rm{E}}^{-1}[g(\Zmat);\pz]$,  and second equality in (\ref{Ddef}) stems from  (\ref{MTMean}) and (\ref{Mu_u_Est}), with $Z_{k}$ and $Z_{k}(n)$ denoting the $k$-th entry of $\Zmat$ and $\Zmat(n)$, respectively.

Under the condition (\ref{Cond3}), and the definitions of $\Zmat$, $\Zmat\left(n\right)$ and $g\left(\cdot\right)$ in  Lemmas {\ref{RICompSigma}}, {\ref{RICompSigmaHat}} in Appendix \ref{RealImag}, the summands in (\ref{Cdef}) and (\ref{Ddef}) have finite variances. Therefore, by the i.i.d. assumption and the law of iterated logarithm (LIL) \cite{Serfling} we have that
\begin{equation}
\label{CO}
C=O\left(\sqrt{N^{-1}\log\log{N}}\right)\hspace{0.2cm}{\rm{a.s.}}
\end{equation}
and
\begin{equation}
\label{DO}
D=O\left(\sqrt{N^{-1}\log\log{N}}\right)\hspace{0.2cm}{\rm{a.s.}}
\end{equation}
Furthermore, under condition (\ref{Cond11}) that follows from (\ref{Cond3}) $\hat{\muvec}^{\left(u\right)}_{\Zmatsc}$ and $\hat{\bSigma}^{\left(u\right)}_{\Zmatsc}$ are strongly consistent, as shown in Appendix \ref{Prop2Proof}. Hence, by (\ref{RemIneq}), (\ref{CO}) and (\ref{DO}), we conclude that the reminder (\ref{Reminder}) satisfies $\Rmat^{\left(g\right)}_{\Zmatsc}=o\left(\sqrt{N^{-1}\log\log{N}}\right)$ a.s. Thus, by Theorem 6.4.2 in \cite{Serfling} we conclude that (\ref{AsProp}) holds.\qed
\subsection{Proof of Proposition \ref{FBSTh}}
\label{FBSThProof}
Similarly to the proof of Theorem \ref{CompGaussStruct}, it can be shown that under the array model (\ref{ArrayModel}), the coherent signals model (\ref{CoherentModel}), and the compound Gaussian noise assumption (\ref{CompGauss}),
the Gaussian MT-covariance matrix of the array output takes the form:
\begin{equation}
\label{SigmaCorr}
\bSigma^{(\uGausss)}_{\Xmatsc}\left(\tau\right)=\dvec\dvec^{H}\sigma^{2(g)}_{\alpha^{2}{s}}\left(\tau\right)+\sigma^{2(h)}_{\alpha\Wmatsc}(\tau)\Imat,
\end{equation}
where $\dvec\triangleq\Amat\bxi$ and $\sigma^{2(g)}_{\alpha^{2}{s}}\left(\tau\right)$ is the variance of $\alpha^{2}\left(n\right)s\left(n\right)$, $\alpha\triangleq\sqrt{\frac{\tau^{2}}{\tau^{2}+\nu^{2}}}$, under the transformed joint probability measure $Q^{\left(g\right)}_{\alpha,s}$ with the MT-function $$g\left(\alpha,s;\tau\right)\triangleq{\left(\frac{\pi\tau^{2}}{\alpha^{2}}\right)}^{-p}\exp{(-{\alpha^{2}\|\dvec\|^{2}_{2}}\left|s\right|^{2}/{\tau^{2}})}.$$ The term  $\sigma^{2(h)}_{\alpha{\Wmatsc}}(\tau)$ is the variance of $\alpha\left(n\right)\Wmat\left(n\right)$ under the transformed joint probability measure $Q^{\left(h\right)}_{\alpha,\Wmat}$ with the MT-function $h\left(\alpha;\tau\right)\triangleq{\rm{E}}\left[g\left(\alpha,s;\tau\right);P_{s}\right]$. 

The Gaussian MT-covariance (\ref{SigmaCorr}) is structured similarly to the standard covariance $\bSigma_{\Xmatsc}$ for coherent signals \cite{Pillai}. Therefore, by the ULA assumption and (\ref{FOR})-(\ref{FBGMTC}), the proof follows using the same argumentations in \cite{Pillai} for the spatially smoothed version of $\bSigma_{\Xmatsc}$. \qed

\bibliographystyle{IEEEbib}

\bibliography{strings,refs}

\end{document}